\documentclass{article}
\usepackage[utf8]{inputenc}
\usepackage[margin=2.5cm]{geometry}
\usepackage{soul}
\usepackage[dvipsnames]{xcolor}
\usepackage{graphicx}
\usepackage{multirow}
\usepackage{xspace}
\usepackage[most]{tcolorbox}
\usepackage{booktabs}
\usepackage{pifont}
\usepackage{balance}
\usepackage{amsmath,amsfonts}
\usepackage{amssymb}
\usepackage{svg}
\usepackage{url}
\usepackage{cryptocode}
\usepackage{booktabs}
\usepackage{graphicx}
\usepackage{subfig}
\usepackage{float}
\usepackage{circledsteps}
\usepackage{algpseudocode}
\usepackage{algorithm}
\usepackage{algcompatible}
\usepackage{multicol}
\usepackage{boxedminipage}
\usepackage{soul}

\usepackage[utf8]{inputenc}
\usepackage{hyperref}
\hypersetup{
    colorlinks=false, 
    linkbordercolor=green, 
    citebordercolor=green, 
    filebordercolor=green, 
    urlbordercolor=green 
}
\usepackage{graphicx}
\usepackage{pdfcolfoot}
\usepackage{lipsum}
\usepackage{amsmath}
\usepackage{amsthm}
\usepackage{amssymb}
\usepackage{amsfonts}
\usepackage{bbm}
\usepackage{algorithm}
\usepackage{algpseudocode}
\usepackage{mathtools}
\usepackage{soul}
\usepackage[normalem]{ulem}
\usepackage{cancel}
\usepackage{cleveref}
\usepackage{ifthen}
\usepackage[stamp=false]{draftwatermark}
\SetWatermarkText{\textsf{DRAFT}}
\SetWatermarkScale{7}
\SetWatermarkColor[gray]{0.93}
\usepackage{titlesec}
\usepackage{xstring}

\floatstyle{ruled}
\newfloat{algo}{htbp}{algo}
\floatname{algo}{Algorithm}

\usepackage{numbertabbing}

\usepackage{IEEEtrantools}
\allowdisplaybreaks

\usepackage[shortlabels,inline]{enumitem}

\usepackage{tikz}
\usetikzlibrary{calc}
\usetikzlibrary{arrows}
\usetikzlibrary{arrows.meta}
\usetikzlibrary{patterns}
\usetikzlibrary{positioning}
\usetikzlibrary{decorations.pathreplacing}
\usetikzlibrary{shapes.misc}
\usetikzlibrary{spy}

\usepackage{pgfplots}
\pgfplotsset{compat=1.14}

\definecolor{myParula01Blue}{RGB}{0,114,189}
\definecolor{myParula02Orange}{RGB}{217,83,25}
\definecolor{myParula03Yellow}{RGB}{237,177,32}
\definecolor{myParula04Purple}{RGB}{126,47,142}
\definecolor{myParula05Green}{RGB}{119,172,48}
\definecolor{myParula06LightBlue}{RGB}{77,190,238}
\definecolor{myParula07Red}{RGB}{162,20,47}

\tikzset{myparula11/.style={color=myParula01Blue,solid,mark=+,mark options={solid}}}
\tikzset{myparula12/.style={color=myParula01Blue,densely dashed,mark=x,mark options={solid}}}
\tikzset{myparula13/.style={color=myParula01Blue,densely dotted,mark=o,mark options={solid}}}
\tikzset{myparula14/.style={color=myParula01Blue,dashdotted,mark=triangle,mark options={solid}}}
\tikzset{myparula15/.style={color=myParula01Blue,dashdotdotted,mark=square,mark options={solid}}}

\tikzset{myparula21/.style={color=myParula02Orange,solid,mark=+,mark options={solid}}}
\tikzset{myparula22/.style={color=myParula02Orange,densely dashed,mark=x,mark options={solid}}}
\tikzset{myparula23/.style={color=myParula02Orange,densely dotted,mark=o,mark options={solid}}}
\tikzset{myparula24/.style={color=myParula02Orange,dashdotted,mark=triangle,mark options={solid}}}
\tikzset{myparula25/.style={color=myParula02Orange,dashdotdotted,mark=square,mark options={solid}}}

\tikzset{myparula31/.style={color=myParula03Yellow,solid,mark=+,mark options={solid}}}
\tikzset{myparula32/.style={color=myParula03Yellow,densely dashed,mark=x,mark options={solid}}}
\tikzset{myparula33/.style={color=myParula03Yellow,densely dotted,mark=o,mark options={solid}}}
\tikzset{myparula34/.style={color=myParula03Yellow,dashdotted,mark=triangle,mark options={solid}}}
\tikzset{myparula35/.style={color=myParula03Yellow,dashdotdotted,mark=square,mark options={solid}}}

\tikzset{myparula41/.style={color=myParula04Purple,solid,mark=+,mark options={solid}}}
\tikzset{myparula42/.style={color=myParula04Purple,densely dashed,mark=x,mark options={solid}}}
\tikzset{myparula43/.style={color=myParula04Purple,densely dotted,mark=o,mark options={solid}}}
\tikzset{myparula44/.style={color=myParula04Purple,dashdotted,mark=triangle,mark options={solid}}}
\tikzset{myparula45/.style={color=myParula04Purple,dashdotdotted,mark=square,mark options={solid}}}

\tikzset{myparula51/.style={color=myParula05Green,solid,mark=+,mark options={solid}}}
\tikzset{myparula52/.style={color=myParula05Green,densely dashed,mark=x,mark options={solid}}}
\tikzset{myparula53/.style={color=myParula05Green,densely dotted,mark=o,mark options={solid}}}
\tikzset{myparula54/.style={color=myParula05Green,dashdotted,mark=triangle,mark options={solid}}}
\tikzset{myparula55/.style={color=myParula05Green,dashdotdotted,mark=square,mark options={solid}}}

\tikzset{myparula61/.style={color=myParula06LightBlue,solid,mark=+,mark options={solid}}}
\tikzset{myparula62/.style={color=myParula06LightBlue,densely dashed,mark=x,mark options={solid}}}
\tikzset{myparula63/.style={color=myParula06LightBlue,densely dotted,mark=o,mark options={solid}}}
\tikzset{myparula64/.style={color=myParula06LightBlue,dashdotted,mark=triangle,mark options={solid}}}
\tikzset{myparula65/.style={color=myParula06LightBlue,dashdotdotted,mark=square,mark options={solid}}}

\tikzset{myparula71/.style={color=myParula07Red,solid,mark=+,mark options={solid}}}
\tikzset{myparula72/.style={color=myParula07Red,densely dashed,mark=x,mark options={solid}}}
\tikzset{myparula73/.style={color=myParula07Red,densely dotted,mark=o,mark options={solid}}}
\tikzset{myparula74/.style={color=myParula07Red,dashdotted,mark=triangle,mark options={solid}}}
\tikzset{myparula75/.style={color=myParula07Red,dashdotdotted,mark=square,mark options={solid}}}

\theoremstyle{plain}
\newtheorem{theorem}{Theorem}
\newtheorem*{theorem*}{Theorem}
\newtheorem{corollary}{Corollary}
\newtheorem{lemma}{Lemma}

\newtheorem{property}{Property}
\newtheorem{constraint}{Constraint}

\crefname{constraint}{constraint}{constraints}
\crefname{eqconstraint}{constraint}{constraints}
\creflabelformat{eqconstraint}{(#2#1#3)}

\crefname{indhypothesis}{Inductive Hypothesis}{Inductive Hypotheses}
\Crefname{theorem}{Theorem}{Theorems}
\newcounter{countersetupcrefforlines}
\loop
  \edef\temp{\noexpand\crefname{linecounter\arabic{countersetupcrefforlines}}{line}{lines}}
  \temp
  \stepcounter{countersetupcrefforlines}
  \ifnum\value{countersetupcrefforlines}<11
\repeat
\undef{\temp}

\theoremstyle{definition}
\newtheorem{definition}{Definition}

\makeatletter
\newcommand{\customlabel}[3][]{%
   \protected@write \@auxout {}{\string \newlabel {#2}{{#3}{\thepage}{#3}{#2}{}} }%
   \protected@write \@auxout {}{\string \newlabel {#2@cref}{{[#1][#3][]#3}{[1][\thepage][]\thepage}}}%
   \hypertarget{#2}{}%
}
\makeatother

\titleclass{\subsubsubsection}{straight}[\subsubsection]
\newcounter{subsubsubsection}[subsubsection]
\renewcommand\thesubsubsubsection{\thesubsubsection.\arabic{subsubsubsection}}

\titleformat{\subsubsubsection}
  {\normalfont\normalsize\bfseries}{\thesubsubsubsection}{1em}{}
\titlespacing*{\subsubsubsection}{0pt}{1.5ex plus .2ex}{0.5em}

\usepackage{xspace}

\newcommand{\ie}[0]{\emph{i.e.}\xspace}
\newcommand{\eg}[0]{\emph{e.g.}\xspace}

\newcommand{\wlogen}[0]{w.l.o.g.\xspace}

 \newcommand{\fra}[1]{}
 \newcommand{\rs}[1]{}
  \newcommand{\lz}[1]{}
 \newcommand{\rsnotfn}[1]{}
 \newcommand{\tht}[1]{}
 \newcommand{\thtnotfn}[1]{}

\crefname{algo}{algorithm}{algorithms}
\crefname{g@linecounter}{line}{lines}
\crefname{condition}{condition}{conditions}

\tikzset{blockchain/.style={
        x=1.25cm,
        y=1.25cm,
        node distance=0.5cm,
        block/.style = {
            minimum width=0.75cm,
            minimum height=0.75cm,
            draw,
            shade,
            top color=white,
            bottom color=black!10,
        },
        block-adv/.style = {
            block,
            bottom color=myParula07Red!50,
            draw=myParula07Red!50!black,
        },
        block-hon/.style = {
            block,
            bottom color=myParula05Green!50,
            draw=myParula05Green!50!black,
        },
        link/.style = {
            -latex,
        },
        link-adv/.style = {
            link,
        },
        link-hon/.style = {
            link,
        },
    }
}

\usepackage{pifont}
\usepackage[flushleft]{threeparttable}

\newcommand{\txpool}[0]{\ensuremath{\mathit{txpool}}}
\newcommand{\tx}[0]{\ensuremath{\mathit{tx}}}
\newcommand{\chain}[0]{\ensuremath{\mathsf{ch}}}
\newcommand{\chainfrozen}[1][]{\ensuremath{\mathsf{ch}^{\mathrm{frozen}\ifthenelse{\equal{#1}{}}{}{,#1}}}}
\newcommand{\Chain}[0]{\ensuremath{\mathsf{Ch}}}
\newcommand{\chainava}[0]{\ensuremath{\mathsf{chAva}}}
\newcommand{\chainfin}[0]{\ensuremath{\mathsf{chFin}}}

\newcommand{\chaincanmfc}[0]{{\chain^{\mathsf{MFC}}}}
\newcommand{\chaincanrlmd}[0]{{\chain^{\mathsf{RLMD}}}}
\newcommand{\fastcands}[0]{\ensuremath{\mathsf{fast}^\mathrm{cands}}}
\newcommand{\fastcand}[0]{\ensuremath{\mathsf{fast}^\mathrm{cand}}}

\newcommand{\rlmdghost}[0]{\ensuremath{\operatorname{\textsc{RLMD-GHOST}}}}
\newcommand{\rlmd}[0]{\ensuremath{\operatorname{\textsc{RLMD}}}}

\newcommand{\hfc}[0]{\ensuremath{\operatorname{\textsc{HFC}}}}

\newcommand{\mfc}[0]{\ensuremath{\operatorname{\textsc{MFC}}}}
\newcommand{\mfcvote}[1]{\ensuremath{\mfc^{\voting{#1}}}}
\newcommand{\mfcpropose}[1]{\ensuremath{\mfc^{\proposing{#1}}}}
\newcommand{\specifyExec}[2]{\overset{#1}{#2}{}}
\newcommand{\mfcFFG}{\specifyExec{\FFGExec}{\mfc}{}}
\newcommand{\mfcNoFFG}{\specifyExec{\NoFFGExec}{\mfc}{}}
\newcommand{\mfcvoteFFG}[1]{\ensuremath{\mfcFFG^{\voting{#1}}}}
\newcommand{\mfcproposeFFG}[1]{\ensuremath{\mfcFFG^{\proposing{#1}}}}
\newcommand{\mfcvoteNoFFG}[1]{\ensuremath{\mfcNoFFG^{\voting{#1}}}}
\newcommand{\mfcproposeNoFFG}[1]{\ensuremath{\mfcNoFFG{}^{\proposing{#1}}}}

\newcommand{\mfcproposeExec}[2]{\ensuremath{\specifyExec{#1}{\mfc}^{\proposing{#2}}}}

\newcommand{\rlmdvote}[1]{\ensuremath{\rlmd^{\voting{#1}}}}
\newcommand{\rlmdpropose}[1]{\ensuremath{\rlmd^{\proposing{#1}}}}

\newcommand{\hfcvp}[0]{\rlmdghost}
\newcommand{\hfcvote}[1]{\ensuremath{\hfcvp^{\voting{#1}}}}
\newcommand{\hfcpropose}[1]{\ensuremath{\hfcvp^{\proposing{#1}}}}

\newcommand{\hfcFFG}{\specifyExec{\FFGExec}{\rlmdghost}{}}
\newcommand{\rlmdNoFFG}{\specifyExec{\NoFFGExec}{\rlmdghost}{}}
\newcommand{\hfcvoteFFG}[1]{\ensuremath{\hfcFFG^{\mathsf{vote},#1}}}
\newcommand{\hfcproposeFFG}[1]{\ensuremath{\hfcFFG^{\mathsf{propose},#1}}}
\newcommand{\rlmdvoteNoFFG}[1]{\ensuremath{\rlmdNoFFG^{\mathsf{vote},#1}}}
\newcommand{\rlmdproposeNoFFG}[1]{\ensuremath{\rlmdNoFFG{}^{\mathsf{propose},#1}}}

\newcommand{\VFFG}{\overset{\FFGExec}{\V}{}}
\newcommand{\VNoFFG}{\overset{\NoFFGExec}{\V}{}}
\newcommand{\makeNoFFG}[2][]{\specifyExec{\NoFFGExec}{\ifthenelse{\equal{#1}{}}{#2}{#2_{#1}}}{}}
\newcommand{\makeFFG}[2][]{\specifyExec{\FFGExec}{\ifthenelse{\equal{#1}{}}{#2}{#2_{#1}}}{}}
\ExplSyntaxOn
\cs_generate_variant:Nn \tl_if_eq:nnTF { ee }
\NewDocumentCommand{\removeparens}{m}
{
  \tl_if_eq:eeTF {\tl_head:n {#1}} { ( } {\tl_if_eq:eeTF {\tl_range:nnn {#1} {-1} {-1}} { ) } {\tl_range:nnn {#1} {2} {-2}} {#1}} {#1}
}

\ExplSyntaxOff
\newcommand{\proposing}[1]{{\ensuremath{\mathsf{propose}(\removeparens{#1})}}}
\newcommand{\voting}[1]{{\ensuremath{\mathsf{vote}(\removeparens{#1})}}}
\newcommand{\fastconfirming}[1]{{\ensuremath{\mathsf{fconf}(\removeparens{#1})}}}
\newcommand{\merging}[1]{{\ensuremath{\mathsf{merge}(\removeparens{#1})}}}

\newcommand{\ghost}[0]{\ensuremath{\operatorname{\textsc{GHOST}}}}

\newcommand{\C}[0]{\ensuremath{\mathcal{C}}}
\newcommand{\T}[0]{\ensuremath{\mathcal{T}}}
\newcommand{\J}[0]{\ensuremath{\mathcal{J}}}
\newcommand{\calS}[0]{\ensuremath{\mathcal{S}}}
\newcommand{\GJ}[0]{\ensuremath{\mathcal{GJ}}}
\newcommand{\GJfrozen}[1][]{\ensuremath{\GJ^{\mathrm{frozen}\ifthenelse{\equal{#1}{}}{}{,#1}}}}
\newcommand{\GF}[0]{\ensuremath{\mathcal{GF}}}

\newcommand{\slot}[0]{\ensuremath{\operatorname{\mathrm{slot}}}}

\long\def\blockcomment#1\endblockcomment{}
\algdef{SE}[DOWHILE]{Do}{doWhile}{\algorithmicdo}[1]{\algorithmicwhile\ #1}%

\algdef{SE}[UPON]{Upon}{EndUpon}[1]{\textbf{upon}\ $#1$}{\textbf{end upon}}%
\algdef{SE}[AT]{At}{EndAt}[1]{\textbf{at}\ $#1$}{\textbf{end at}}%

\newcommand{\GAT}[0]{\ensuremath{\mathsf{GAT}}}
\newcommand{\GST}[0]{\ensuremath{\mathsf{GST}}}

\newcommand{\V}[0]{\ensuremath{\mathcal{V}}}
\newcommand{\Vglobal}[0]{\ensuremath{\V_\mathsf{G}}}
\newcommand{\Vfrozen}[1][]{\ensuremath{\mathcal{V}^{\mathrm{frozen}\ifthenelse{\equal{#1}{}}{}{,#1}}}}
\newcommand{\X}[0]{\ensuremath{\mathcal{X}}}
\newcommand{\M}[0]{\ensuremath{\mathcal{M}}}

\newcommand{\FC}[0]{\ensuremath{\mathsf{FC}}}

\newcommand{\negl}[0]{\ensuremath{\operatorname{negl}}}

\newcommand{\Tconf}[0]{\ensuremath{T_\mathsf{conf}}}
\newcommand{\Tafter}[0]{\ensuremath{T_\mathsf{sec}}}
\newcommand{\Tdyn}[0]{\ensuremath{T_\mathsf{dyn}}}
\newcommand{\treorg}[0]{\ensuremath{{t_\mathsf{reorg}}}}
\newcommand{\Treorg}[0]{\ensuremath{{T_\mathsf{reorg}}}}
\newcommand{\theal}[0]{\ensuremath{{t_\mathsf{heal}}}}

\newcommand{\node}[0]{\ensuremath{\mathcal{P}}}

\newcommand{\genesis}[0]{\ensuremath{B_\mathrm{genesis}}\xspace}
\newcommand{\GHOST}[0]{\textsf{GHOST}\xspace}
\newcommand{\Goldfish}[0]{\textsf{Goldfish}\xspace}
\newcommand{\LMDGHOST}[0]{\textsf{LMD-GHOST}\xspace}
\newcommand{\RLMDGHOST}[0]{\textsf{RLMD-GHOST}\xspace}

\newcommand{\TOBSVD}[0]{\textsf{TOB-SVD}\xspace}

\newcommand{\FIL}[0]{\textsf{FIL}\xspace}
\newcommand{\true}[0]{\ensuremath{\mathit{true}}}

\newcommand{\LOGbft}[2]{%
    \ifthenelse{\equal{#1}{}}{%
        \ensuremath{\mathsf{LOG}_{\mathrm{bft}}^{#2}}%
    }{%
        \ensuremath{\mathsf{LOG}_{\mathrm{bft},#1}^{#2}}%
    }%
}

\newcommand{\ld}[1]{%
    \ifthenelse{\equal{#1}{}}{%
        \ensuremath{\mathrm{L}^{(c)}}%
    }{%
        \ensuremath{\mathrm{L}^{(#1)}}%
    }%
}

\newcommand{\bprop}[1]{%
    \ifthenelse{\equal{#1}{}}{%
        \ensuremath{\Hat{b}}%
    }{%
        \ensuremath{\Hat{b}_{#1}}%
    }%
}



\title{3-Slot-Finality Protocol for Ethereum\footnote{This work combines multiple preliminary publications which appeared at ESORICS-CBT~2023~\cite{DBLP:conf/esorics/DAmatoZ23},
  PODC~2024~\protect\cite{DBLP:conf/podc/DAmatoLZ24},
  CSF~2024~\protect\cite{rlmd}, and
  \emph{under submission}~\protect\cite{streamliningSBFT}.}
  }

\author{
  Francesco D'Amato\\
  Ethereum Foundation\\
  \and Roberto Saltini\\
  Independent\thanks{Work done while at Consensys. Supported by Ethereum Foundation grant FY24-1433.}
  \and Thanh-Hai Tran\\
  Independent\footnotemark[2] 
  \and Luca Zanolini\\
  Ethereum Foundation
}

\date{}

\begin{document}

\maketitle

\begin{abstract}
{Gasper, the consensus protocol currently employed by Ethereum, typically requires 64 to 95 slots -- {the units of time during which a new \emph{chain} extending the previous one by one block is proposed} and voted -- to finalize. This means that under ideal conditions -- where the network is synchronous, and all chain proposers, along with more than two-thirds of the validators, behave as dictated by the protocol -- proposers construct blocks on a non-finalized chain that extends at least 64 blocks. 
This exposes a significant portion of the blockchain to potential reorganizations during changes in network conditions, such as periods of asynchrony.
Specifically, this finalization delay heightens the network's exposure to Maximum Extractable Value (MEV) exploits, which could undermine the network's integrity. Furthermore, the extended finalization period forces users to balance the trade-off between economic security and transaction speed. 

To address these issues and speed up finality, we introduce a partially synchronous finality gadget, which we combine with two dynamically available consensus protocols -- synchronous protocols that ensure safety and liveness even with fluctuating validator participation levels. This integration results in secure ebb-and-flow protocols [SP 2021], achieving finality within three slots after a proposal and realizing \emph{3-slot finality}}.

\end{abstract}

\section{Introduction and Related Work}
\label{sec:introduction}
Traditional Byzantine consensus protocols, such as PBFT~\cite{DBLP:conf/osdi/CastroL99} or HotStuff, are designed for distributed systems where participants are fixed, known in advance, and cannot go \emph{offline} without being considered faulty.

Recently, dynamic participation has become a critical requirement for developing permissionless consensus protocols. This concept, initially formalized by Pass and Shi through their \emph{sleepy model}~\cite{sleepy}, encapsulates the ability of a system to handle honest participants who may go offline and come back online. A consensus protocol that maintains safety and liveness while accommodating dynamic participation is called \emph{dynamically-available}.

One problem of such dynamically-available protocols is that they do not tolerate network partitions~\cite{DBLP:journals/corr/abs-2304-14701}; no consensus protocols can satisfy both liveness (under dynamic participation) and safety (under network partitions). Simply put, a consensus protocol cannot produce a single chain\footnote{In this context, we are extending the traditional notion of consensus, typically understood as a one-shot primitive. Technically, this should be referred to as total-order broadcast or atomic broadcast. However, for the sake of a general audience, we have adopted the term ``consensus." Consequently, we consider the output of these protocols to be a sequence of transactions, batched in blocks, forming a chain. This will be formalized in Section~\ref{sec:model}.} that concurrently offers Dynamic Availability and guarantees transaction finality in case of asynchronous periods or network partitions. Because of that, dynamically-available protocols studied so far are devised in synchronous settings~\cite{goldfish, DBLP:journals/iacr/MalkhiMR22,DBLP:conf/ccs/Momose022, rlmd, streamliningSBFT}.

Working around this impossibility result, Neu, Tas, and Tse~\cite{DBLP:conf/sp/NeuTT21} introduced a family of protocols referred to as \emph{ebb-and-flow} protocols. An ebb-and-flow protocol comprises two sub-protocols, each with its own \emph{confirmation rule}, and each outputting a chain, with one serving as a prefix of the other. The first confirmation rule defines what is known as the \emph{available chain}, which provides liveness under dynamic participation (and synchrony). The second confirmation rule defines the \emph{finalized chain}, and provides safety even under network partitions, but loses liveness either under asynchrony or in case of fluctuation in the participation level. Interestingly, such family of protocols also captures the nature of the Ethereum consensus protocol, Gasper~\cite{gasper}, in which the available chain is output by (the confirmation rule of) \LMDGHOST~\cite{zamfir} protocol, and the finalized chain by the (confirmation rule of the) \emph{finality gadget} Casper~FFG~\cite{casper}. However, the original version of \LMDGHOST is not secure even in a context of full participation {and synchrony}. Potential attack vectors have been identified~\cite{DBLP:conf/sp/NeuTT21, DBLP:conf/fc/Schwarz-Schilling22} that undermine the protocol's safety and liveness.

Motivated by finding a more secure alternative to \LMDGHOST, and following the ebb-and-flow approach, D'Amato \emph{et al.}~\cite{goldfish} devise a synchronous dynamically-available consensus protocol, \Goldfish, that, combined with a generic (partially synchronous) finality gadget, implements a secure ebb-and-flow protocol. Moreover, \Goldfish is \emph{Reorg Resilient}: chains proposed by honest validators are guaranteed to not be reorganized. However, \Goldfish is brittle to temporary asynchrony~\cite{rlmd}, in the sense that even a single violation of the bound of network delay can lead to a catastrophic failure, jeopardizing the safety of \emph{any} previously confirmed chain, resulting in a protocol that is not practically viable to replace \LMDGHOST in Ethereum. In other words, \Goldfish is not \emph{Asynchrony Resilient}.

To cope with the limitation of \Goldfish with asynchrony, D'Amato and Zanolini~\cite{rlmd} propose \RLMDGHOST, a provably secure synchronous consensus protocol that does not lose safety during \emph{bounded} periods of asynchrony and which tolerates a weaker form of dynamic participation, offering a trade-off between Dynamic Availability and Asynchrony Resilience. Their protocol results appealing for practical systems, where strict synchrony assumptions might not always hold, contrary to what is generally assumed with standard synchronous dynamically-available protocols. 

The family of protocols to which \Goldfish and \RLMDGHOST belong are consensus protocols that are \emph{probabilistically} safe, guaranteeing safety with overwhelming probability. In contrast, Momose and Ren's research~\cite{DBLP:conf/ccs/Momose022} laid the foundation for \emph{deterministically} safe, dynamically-available consensus protocols, sparking a wave of subsequent research~\cite{DBLP:journals/iacr/MalkhiMR22, DBLP:conf/ccs/MalkhiM023, DBLP:conf/wdag/GafniL23, DBLP:conf/podc/DAmatoLZ24}. Unlike \Goldfish and \RLMDGHOST these protocols achieve deterministic safety by employing the notion of quorums. Traditional quorums, defined by a fixed number of actively engaging participants, are not suitable in a dynamic participation context. {By leveraging \emph{Graded Agreement}\footnote{In Graded Agreement, each decision is assigned a grade, which intuitively reflects the strength of the agreement. Notably, different formulations of graded agreement exist, each with slighly different properties. In this work, we focus on the properties defined by D'Amato \emph{et al.}~\cite{streamliningSBFT}.},} Momose and Ren~\cite{DBLP:conf/ccs/Momose022} redefined quorums dynamically, according to current participation levels, while maintaining critical properties of traditional quorums.

This development has led to the creation of various consensus protocols {based on Graded Agreement}, each with unique properties and varying levels of adversarial tolerance. The initial protocol by Momose and Ren~\cite{DBLP:conf/ccs/Momose022} accommodates up to 1/2 adversarial participants but is limited by a high latency of $16\Delta$, with $\Delta$ being the message delay bound. The subsequent study by Malkhi, Momose, and Ren~\cite{DBLP:journals/iacr/MalkhiMR22} introduced two protocols that reduce the latency to $3\Delta$ and $2\Delta$, respectively, but at the expense of lower adversarial tolerance to 1/3 and 1/4. Later enhancements~\cite{DBLP:conf/ccs/MalkhiM023} managed to revert to tolerating minority corruption while maintaining a comparable latency of $4\Delta$. Another concurrent and independent work~\cite{DBLP:conf/wdag/GafniL23} achieves 1/2 adversarial resilience, with a latency of $6\Delta$.

D'Amato \emph{et al.}~\cite{streamliningSBFT}, aiming to enhance the practicality of deterministically safe, dynamically-available consensus protocols, particularly in systems with many participants, introduced a consensus protocol, named \TOBSVD{}, that tolerates up to 1/2 adversarial participants and achieves latency comparable to \cite{DBLP:conf/ccs/MalkhiM023} -- slightly better {in expectation} and slightly worse in the best case. Crucially, it requires only a single vote round per decision {in the best case}, in contrast to the nine rounds required by \cite{DBLP:conf/ccs/MalkhiM023}.
{Moreover, D'Amato, Losa, and Zanolini~\cite{DBLP:conf/podc/DAmatoLZ24}, explored mechanisms to make dynamically-available consensus protocols based on Graded Agreement resilient to \emph{bounded} periods of asynchrony. }

\sloppy{Previous work by D'Amato and Zanolini~\cite{DBLP:conf/esorics/DAmatoZ23} proposed an ebb-and-flow protocol by combining the dynamically-available protocol \RLMDGHOST with a finality gadget similar to Casper FFG, hereafter referred to as the \emph{SSF} protocol.}
In their work, as in ours, time is divided into \emph{slots} and at the beginning of each slot a new chain is proposed by an elected proposer.
Importantly, the SSF protocol ensures that, under synchrony and at least 2/3 of the participants being honest and online, any chain proposed by an honest participant is finalized within the same slot, \ie, before the next proposer's turn.
This single-slot finality is achieved through three vote rounds within a slot.
Specifically, the SSF protocol operates in slots of duration $4\Delta$ rounds each. During the first round, a proposal is made.
In the second round, validators cast votes for what they perceive as the tip of the chain, ideally the proposal just made. In the third round, if a quorum of votes for the same head block is observed, validators cast a \emph{finality} vote for that block.
In the final round, if a validator sees a quorum of finality votes for a block, it broadcasts an acknowledgment message for external observers. If an observer sees a quorum of acknowledgments for a block by the end of the slot, it can declare the block finalized.

While theoretically sound, the efficiency and usability of this protocol for large-scale blockchain networks are questionable due to the number of vote {phases} required for each slot.
In fact, large-scale blockchain networks such as Ethereum, due to the large number of participants, to reduce the bandwidth requirements, employ an aggregation process by which votes are first sent to aggregators who then distribute the aggregated signatures.
As a consequence of this, each vote  {phase} requires double the normal network latency, increasing the slot time and decreasing the transaction throughput.
This means that in practice, in the SSF protocol a slot is $6\Delta$ as the second and third vote  {phase} need to wait for the first and second ones to complete, respectively, but the third vote  {phase} can proceed in parallel with the next slot.
Moreover, it is unclear whether the messages cast during the third  {phase} should be included on-chain or kept off-chain and potentially delaying the slot's finalization in practice.

In this work, we propose a finality gadget that can be composed with dynamically-available protocols to obtain an ebb-and-flow protocol with only one vote {phase} per slot and slot length of $5\Delta$ if we consider a vote {phase} taking $2\Delta$.
All things equals, this represents a 20\% improvement in practical network throughout compared to the protocol by D'Amato and Zanolini~\cite{DBLP:conf/esorics/DAmatoZ23}\footnote{{This assuming the same block size for both the SSF protocol and our proposed protocol. In practice, throughput can be independent of slot duration, as desired per-second throughput can be set independently, with higher per-slot throughput achievable by using longer slots.}}.

The trade-off that we make is in delaying the finalization of a chain proposed by an honest validator to occur two slots later, \ie, in our protocol a chain proposed by an honest validator in slot $t$ is finalized\footnote{Here, we consider the finalization time of a chain $\chain$ to be the time after which no chain conflicting with $\chain$ can ever by finalized, which can occur earlier than when some honest node finalizes chain $\chain$ in their view.} by the end of slot $t+2$, as long as synchrony holds, and at least 2/3 of the participants are honest and active till the vote round of slot $t+2$.
Importantly, this is ensured regardless of whether the proposers of slots $t+1$ and $t+2$ are honest.
In practice, this means that, assuming that each vote {phase} takes $2\Delta$, we require the synchronous and participation assumptions to hold for $11\Delta$ (in our protocol the vote round occurs $\Delta$ time after the beginning of a slot), rather than $5\Delta$ as in the SSF protocol.
However, {like SSF,} we offer a method by which, under these assumptions and considering that vote phases take $2\Delta$, a chain proposed by an honest proposer is \emph{confirmed} by the dynamically-available protocol of any honest validator at time $3\Delta$ of the same slot, which ensures that such chain will then be finalized in slot $t+2$.
This offers an avenue for users to know in advance that, as long as the assumptions on synchrony and participation holds, such chain will be finalized.

Also, we can integrate the third round of voting from the SSF protocol into ours to obtain a protocol that has two vote {phases} per slot, but still retain slot length of $5\Delta$ --
when assuming that each vote {phase} takes $2\Delta$ because, as explained above, the third vote {phase} from the SSF protocol can proceed in parallel with the rest of the protocol and therefore does not increase the practical slot length -- but that reduces finalization from $11\Delta$ down to $8\Delta$.

In practice, for networks where the periods of synchrony and at least 2/3 of the participants being honest and online typically last much longer than $11\Delta$, our protocol presents no real drawback as confirmation still happens within the same slot which then leads to finalization, while attaining a higher transaction throughput and, arguably, a simpler protocol.
{
Moreover, users are concerned with the time that a transaction takes to be either confirmed or finalized, which is the time that it takes for this transaction to be included in a proposed chain and for this chain to be either confirmed or finalized.
This is not the same as the time taken to confirm or finalize a block proposed by an honest validator
because transactions are not necessarily submitted just before an honest proposer proposes a chain.
They are submitted whenever the user needs to interact with the blockchain.
So, they could be submitted at any point of a slot.
Also, they could be submitted when the proposer of the next slot is Byzantine who therefore might not include them in the chain that they proposer, if any.
For this reason, from a user perspective, it makes sense to consider the expected confirmation and finalization times under the assumption that a transaction submission time is uniformly distributed.
Then, the expected confirmation and finalization times correspond to the time taken to confirm or finalize a chain proposed by an honest proposer, depending on which of the two measures we interested in, plus  $\frac{(1+\beta) \cdot \text{slot-time}}{2(1 - \beta)}$ where \(\beta\) represents the adversarial power in the network.

Let us compare the expected confirmation time of SSF and 3SF first.
For $\beta = \frac{1}{3}$\footnote{Liveness can only be guaranteed if $\beta < \frac{1}{3}$}, the expected confirmation time for SSF is $9\Delta$ whereas for 3SF is $8\Delta$ meaning an $\approx 11$\% improvement.
For $\beta = 0$, the expected confirmation time for SSF is $6\Delta$ whereas for 3SF is $5.5\Delta$ meaning an $\approx 8$\% improvement.

Moving to the expected finalization time, for $\beta = \frac{1}{3}$%
,
for SSF it is $11\Delta$, for 3SF it is $16\Delta$ and for the two-slot variant of 3SF it is $13\Delta$ meaning that the expected finalization time for 3SF  $\approx 46$\% higher than the one of SSF, but this reduces to $\approx 18$\% for the its two-slot variant.
For $\beta = 0$, the expected finalization time for SSF is $8\Delta$, for 3SF is $13.5\Delta$ and for the two-slot variant of 3SF is $10.5\Delta$ meaning that the expected finalization time for 3SF  $\approx 69$\% higher than the one of SSF, but this reduces to $\approx 31$\% for the its two-slot variant.

}

Finally, slot time has been shown\cite{milionis2024automatedmarketmakinglossversusrebalancing} to be an important parameter in determining the \emph{economic leakage} of on-chain automated market makers (AMMs) due to arbitrage. For instance, arbitrage profits (and equivalently liquidity providers (LP) losses) are proportional to the square root of {slot} time, so that a lower {slot} time is very desirable by financial applications built on top of a smart contract blockchain.

{
Overall, our protocol achieves a balance by trading a higher expected finalization time for a shorter expected confirmation time, which could be sufficient for most users. At the same time, it offers shorter slot time and improved throughput as discussed above.
}

{Additionally,} we show how to integrate our finality gadget
with two dynamically-available protocols to obtain a secure ebb-and-flow.
The first dynamically-available protocol that we consider is a probabilistically-safe and bounded-asynchrony-period-resilient variant of the deterministically-safe protocol \TOBSVD{}~\cite{streamliningSBFT}, the second is \RLMDGHOST~\cite{rlmd}.
The resulting protocols, in addition to the standard ebb-and-flow properties, ensure Safety and Reorg Resilience of the available chain even in the face of a subsequent bounded period of asynchrony.
This makes our protocols particularly attractive for blockchain networks, such as Ethereum, where safety failures in the available chain can be exploited by dishonest participants to steal honest participant's Maximum Extractable Value (MEV)~\cite{DBLP:conf/sp/DaianGKLZBBJ20} without being punished for it.
{Moreover, critical to the practical application in large-scale blockchain networks, both of our resulting protocols ensure that the dynamically-available component can \emph{heal} from any arbitrarily long period of asynchrony as long as at least 2/3 of participants are honest and online for a sufficient amount of time.
This is not a property mentioned in the original work that introduced ebb-and-flow protocol~\cite{ebbandflow} where the dynamically-available protocol is required to ensure safety only if synchrony holds from the beginning.
However, any real system is bound to experience some period of asynchrony of arbitrary length at some point.
Ensuring that the dynamically-available protocol can recover is then important from a practical point of view.}

The remainder of this work is structured as follows. In Section~\ref{sec:model}, we present our system model along with all the necessary background notions {and property definitions}. Common notions for both the protocols we introduce in this work are detailed in Section~\ref{sec:common-notions}.
Section~\ref{sec:ffg} introduces and proves the correctness of the finality gadget.
Notably, the finality gadget, or \emph{FFG-component}, is common to both presented protocols. Therefore, the properties and results discussed in Section~\ref{sec:ffg} apply to both protocols.
Our first faster finality protocol is presented in Section~\ref{sec:ga-based}, where we provide the pseudo-code of the {first} protocol, the one based upon \TOBSVD~\cite{streamliningSBFT},  and prove its properties, demonstrating that it is a secure ebb-and-flow protocol (as defined in Section~\ref{sec:model}).
The second protocol, the one based upon \RLMDGHOST~\cite{rlmd}, is introduced in Section~\ref{sec:rlmd-based}, and we conduct a similar analysis to that in Section~\ref{sec:ga-based}.
In \Cref{sec:practical-consideration}, we discuss how to resolve some of the main challenges presented in implementing either protocol, and examine their communication complexity.
Then, in \Cref{sec:healing}, we detail the conditions required for Dynamic Availability and Reorg Resilience to hold even in partially synchronous settings and provide the intuition underpinning this leaving, the detailed proof to \Cref{sec:healing-detailed}.
Also, in \Cref{sec:two-slot-finality} we explore a modification to our protocols that allow finalizing chains within two slots only by adding one vote round but without this affecting the practical length of slots when aggregation is used to reduce bandwidth requirements.
Finally, conclusions are drawn in Section~\ref{sec:conclusion}.

\section{Model and Preliminary Notions}
\label{sec:model}

\subsection{System model}

\paragraph*{Validators.} We consider a set of $n$ \emph{validators} $v_1, \dots, v_n$ that communicate with each other by exchanging messages. Every validator is identified by a unique cryptographic identity and the public keys are common knowledge. Validators are assigned a protocol to follow, consisting of a collection of programs with instructions for all validators. 
We assume the existence of a probabilistic poly-time adversary $\mathcal{A}$ that can choose up to $f$ validators to corrupt over an entire protocol execution. Any validator is \emph{honest} until, if ever, it is corrupted by the adversary at which point it stays corrupted for the remaining duration of the protocol execution, and is thereafter called \emph{adversarial}. The adversary $\mathcal{A}$ knows the the internal state of adversarial validators. The adversary is \emph{adaptive}: it chooses the corruption schedule dynamically, during the protocol execution.
Honest validators faithfully follow the assigned protocol while adversarial validators may deviate from it arbitrarily.
Each validator has a \emph{stake}, which we assume to be the same for every validator.
If a validator $v_i$ misbehaves and a proof of this misbehavior is given, it loses a part of its stake ($v_i$ gets \emph{slashed}). 

\paragraph*{Time.} Time is divided into discrete \emph{rounds}. 
We define the notion of \emph{slot} as a collection of $4\Delta$ rounds.
Given round $r$ we define $\slot(r)$ as the slot of round $r$, \ie, $\slot(r) := \lfloor \frac{r}{4\Delta} \rfloor$.

\paragraph*{View.} A \emph{view} (at a given round $r$), denoted by {$\V^r$}, is a subset of all the messages that a validator has received until $r$. The notion of view is \emph{local} for the validators. For this reason, when we want to focus the attention on a specific view of a validator $v_i$, we denote with {$\V_i^r$} the view of $v_i$ (at round $r$).

\paragraph*{Links.} We assume that a best-effort gossip primitive that will reach all validators is available. In a protocol, this primitive is accessed through the events “sending a message through gossip” and “receiving a gossiped message.” We also use the term \emph{casts} to describe when a validator $v_i$ sends a message through gossip. Moreover, we assume that messages from honest validator to honest validator are eventually received and cannot be forged. This includes messages sent by adversarial validators, once they have been received by some honest validator $v_i$ and gossiped around by $v_i$. 

\paragraph*{Network Model.} We consider a partially synchronous model in which validators have synchronized clocks but there is no a priori bound on message delays. However, there exists a time (not known by the validators), called \emph{global stabilization time} (\GST), after which message delays are bounded by $\Delta$ rounds with $\Delta$ known to validators.

{We also allow for a \emph{single short asynchronous period} after $\GST$, starting at slot $t_a+1$ and extending for at most $\pi$ slots, \ie, including at most rounds $[4\Delta (t_a+1), 4\Delta (t_a+\pi+1))$.
The value of $t_a$ is unbeknownst to the validators, but we assume that the validators know the value of $\pi$.}

\paragraph*{Sleepiness.}
The adversary $\mathcal{A}$ can decide for each round $r$ which honest validator is \emph{awake} or \emph{asleep} at round $r$. Asleep validators do not execute the protocol and {messages received in round $r$} are queued and delivered in the first round in which the validator is awake again. Honest validators that become awake at round~$r$, before starting to participate in the protocol, must first execute (and terminate) a \emph{joining protocol} (see Section~\ref{sec:joining-protocol}), after which they become \emph{active}~\cite{goldfish}.
All adversarial validators are always awake, and are not prescribed to follow any protocol. Therefore, we always use active, awake, and asleep to refer to honest validators. As for corruptions, the adversary is adaptive also for sleepiness, \ie, the sleepiness schedule is also chosen dynamically by the adversary. Note that awake and active validators coincide in the sleepy model~\cite{sleepy}. Finally, there is a time (not known by the validators), called \emph{global awake time} (\GAT), after which all validators are always~awake.

{
\paragraph*{Transaction Pool.}
We assume the existence of an ever growing set of transaction, called \emph{transaction pool} and denoted by $\txpool$, that any
 every validator has read access to.
Consistently with previous notation, we use $\txpool^r$ to refer to the content of $\txpool$ at time $r$.
}

\paragraph*{Blocks and chains.}
A \emph{block} is a pair of elements, denoted as \( B = (b,p) \). Here, \( b \) represents the \emph{block body} -- essentially, the main content of the block which contains a batch of transactions grouped together\footnote{We often use the dot notation to refer to the elements of a tuple. For instance, $B.b$ represents the block body of $B$.}.
{We assume that a block body can contain an unbounded number of transactions.}
For the sake of simplicity, we are not including details about the block's parent in this definition. In practice, each block body contains a reference pointing to another block. 
The second element of the pair, \( p \geq 0 \), indicates the \emph{slot} where the block \( B \) is proposed.
By definition, if $B_p$ is the parent of $B$, then $B_p.p < B.p$.
We denote with {$\genesis=(b_{-1},-1)$} the \emph{genesis block}, which is the only block that does not have a parent {and has a negative slot}.
Given the definition above, each different block $B$ implicitly identifies a different finite \emph{chain} of blocks starting from block $B$, down to the genesis block, by recursively moving from a block to its parent.
Hence, we make no real distinction between a block and the chain that it identifies.
So, by chain $\chain$, we mean the chain identified by the block $\chain$.
Let us consider two chains, $\chain_1$ and $\chain_2$. 
We write $\chain_1 \prec \chain_2$ if and only if $\chain_1$ is a strict ancestor of $\chain_2$.
In this case, we also say that $\chain_1$ is a strict prefix of $\chain_2$ or, conversely, that $\chain_2$ is a strict extension of $\chain_1$.
We say that $\chain_1$ \emph{conflicts} with $\chain_2$ if and only if neither $\chain_1 \preceq \chain_2$ nor $\chain_2 \preceq \chain_1$ holds. 
Given a chain \(\chain\), we define the \(\kappa\)-deep prefix of \(\chain\) at slot \(t\) (denoted by \(\chain^{\lceil\kappa,t}\)) as the longest prefix \(\chain'\) of \(\chain\) such that \(\chain'.p \leq t - \kappa\).
When slot $t$ is clear from the context, we just write $\chain^{\lceil \kappa}$ to mean $\chain^{\lceil \kappa,t}$.
We also define the following total pre-order between chains. 
{We write $\chain_1 \leq \chain_2$ if and only if $\chain_1.p \leq \chain_2.p$.
However, when we write $\chain_1 = \chain_2$, we mean that the two chains are the same, not that just $\chain_1.p =\chain_2.p$.} 
Also, we sometimes also call $\chain.p$ the \emph{height} or \emph{length} of chain $\chain$.
Finally, we let $\tx \in \chain$ to mean that transaction $\tx$ is included in a block of chain $\chain$, \ie, $\exists \chain' \preceq \chain,\, \tx \in \chain'.b$.
{Naturally, if $\mathit{TX}$ is a set of transactions, then we write $\mathit{TX} \subseteq \chain$ to mean $\forall \mathit{tx}\in \mathit{TX}, \mathit{tx} \in \chain$.}

{\paragraph*{Blockchain Protocol.} 
In this work, we only consider protocols that output one or more chains.
If $\chain$ is the chain output by a protocol $\Pi$, then we use the notation $\chain^r_i$ to indicate the chain output by validator $v_i$ at round $r$ when executing protocol $\Pi$ (we do not include $\Pi$ in the notation as that will always be clear from the context).
If $v_i$ is asleep in round $r$ and $r_a < r$ is the highest round during which $v_i$ was awake, then we assume that $\chain^r_i = \chain^{r_a}_i$. }

\paragraph*{Proposer election mechanism.}
In each slot $t$, a validator~$v_p$ is selected to be a \emph{proposer} for~$t$, \ie, to extend the chain with new blocks. Observe that, when we want to highlight the fact that $v_p$ is a proposer for a specific slot~$t$, we use the notation~$v_{p}^t$. Otherwise, when it is clear from the context, we just drop the slot $t$, to make the notation simpler. As the specification of a proposal selection mechanism is not within the goals of this work, we assume the existence of a proposer selection mechanism satisfying the requirements of \emph{uniqueness, unpredictability, and fairness}: $v_p$ is unique, the identity of~$v_p$ is only known to other validators once~$v_p$ reveals itself, and any validator has probability~$\frac{1}{n}$ of being elected to be a proposer at any slot.

\paragraph*{Proposing.} In this work, we only consider protocols that define the action of a validator \emph{proposing} a chain. The specifics of this actions are protocol dependant.
This generalizes the common notion of proposing a block.
For any of the two protocols presented in this work, proposing happens in the first round of a slot.
{Hence, for readability purpose, we often use $\proposing{t} := 4\Delta t$.}

{\paragraph*{Voting.} Either of the two protocols presented in this work include a vote round where validators vote on the proposed chain.
Again for sake of readability, we often use $\voting{t} := 4\Delta t + \Delta$.}


\subsection{Adversary resiliency}
\label{sec:adversary-resiliency}
Let $H_r$, $A_r$, and $H_{r, r'}$ be the set of active validators at round $r$, the set of adversarial validators at round $r$, and the set of  validators that are active \emph{at some point} in rounds $[r,r']$, \ie, $H_{r,r'} = \bigcup_{i=r}^{r'} H_i$ (if $i < 0$ then $H_i \coloneqq \emptyset$), respectively.
We let $A_\infty = \lim_{t \to \infty} A_{\voting{t}}$.
Note that $f = |A_\infty|$. Also, we say that validator $v_i$ is \emph{always-honest} if and only if $v_i \notin A_\infty$, \ie, it is never corrupted.

Except when explicitly mentioned, throughout this work we assume $f<\frac{n}{3}$.
Each protocol presented in this work extends this assumption in slightly different ways.

\subsubsection{Assumptions for the protocol of Section~\ref{sec:ga-based}}
For the protocol in Section~\ref{sec:ga-based} we assume that for every slot~$t$ after $\GST$:




\begin{align}
  &|H_{\voting{t}} \setminus {A_{\voting{t+1}}}| > |A_{\voting{t+1}} \cup \left(H_{\voting{(t-\eta+1)},\voting{(t-1)}}\setminus H_{\voting{t}}\right)|\label[eqconstraint]{eq:pi-sleepiness}
\end{align}
where
$
\eta := \begin{cases}
  1, &\text{if $\pi = 0$}\\
  \pi + 2, &\text{otherwise}
\end{cases}
$.

In other words, we require the number of active validators at round $\voting{(t-1)}$ which are not corrupted by round $\voting{t}$ to be greater than the number of adversarial validators at round $\voting{t}$, together with the number of validators that were active {at some point} between rounds $\voting{(t-\eta)}$ and $\voting{(t-2)}$, but not at round $\voting{(t-1)}$, with $\eta = \pi +2$ except if $\pi = 0$ (\ie, we admit no asynchronous period after $\GST$), in which case $\eta = 1$ and Constraint~\eqref{eq:pi-sleepiness} simplifies to $|H_{\voting{(t-1)}} \setminus {A_{\voting{t}}}| > |A_{\voting{t}}|$.
Intuitively, $\eta$ corresponds to \emph{the expiration period} applied to some messages that are discarded by the protocol if they were sent more than $\eta$ slots ago. 
This will become clearer later in \Cref{sec:ga-based}.

We also assume that all of the following constraints hold: 
\begin{align}
  &\forall t' \in \{t_a+1, \dots, t_a + \pi+2\},\; |H_{\voting{t_a}} \setminus A_{\voting{t}}| > | A_{\voting{t'}} \cup \left(H_{\voting{(t'-\eta)},\voting{(t'-1)}} \setminus H_{\voting{t_a}}\right)|\label[eqconstraint]{eq:async-condition}\\
  &H_{\voting{t_a}} \setminus A_{\voting{(t_a+1)}} \subseteq H_{\voting{t_a} + \Delta}\label[eqconstraint]{eq:async-condition3}\\
  &|(H_{\voting{(t_a+1,t_a+\pi+1)}} \setminus H_{\voting{(t_a)}})\cup A_\infty|<\frac{2}{3}\label[eqconstraint]{eq:async-condition2}
\end{align}

The three constraints above only deal with the short asynchronous period $\pi$.
\Cref{eq:async-condition} limits how many honest validators that were not active in round $\voting{t_a}$ may awake during the short asynchronous period or the two following slots.
This is because the adversary can take advantage of the short asynchronous period to manipulate them to send messages that would jeopardize the safety of the protocol.
Constraint~\eqref{eq:async-condition3} just requires that all the validators active in round $\voting{t_a}$ and not corrupted by round $\voting{t_a+1}$ are also active $\Delta$ rounds after $\voting{t_a}$.
{\Cref{eq:async-condition2} is required to prevent the adversary from leveraging the messages not subject to expiration to affect the safety of the protocol. We will provide a more detailed explanation of such a scenario in due course in \Cref{sec:analysis-tob-sync}.}

\subsubsection{Assumptions for the Protocol of Section~\ref{sec:rlmd-based}}
For the protocol in Section~\ref{sec:rlmd-based} we instead require that the following condition holds for any slot $t$ after $\GST$, in addition to \Cref{eq:async-condition,eq:async-condition3,eq:async-condition2} defined above:
\begin{equation}
\label[eqconstraint]{eq:sleepy-req}
  |H_{\voting{t}}| > |A_{\voting{(t+1)}} \cup (H_{\voting{(t-\eta+1)}, \voting{(t-1)}}\setminus H_{\voting{t}})|
\end{equation}

Observe that Constraint~\eqref{eq:sleepy-req} is less restrictive than Constraint~\eqref{eq:pi-sleepiness}, as it just imposes limits on $|H_{\voting{(t-1)}}|$ rather than on $|H_{\voting{(t-1)}} \setminus {A_{\voting{t}}}|$. The reason for this diminished restrictiveness will become clear when we present our protocol in Section~\ref{sec:rlmd-based}.

We refer to the adversary model just described as the \emph{generalized partially synchronous $\eta$-sleepy model} (or \wlogen, when the context is clear, as the \emph{$\eta$-sleepy model} for short). Finally, we say that an execution in the generalized partially synchronous sleepy model is \emph{$\eta$-compliant} if and only if it satisfies $\eta$-sleepiness, \ie, if it satisfies either \Cref{eq:pi-sleepiness,eq:async-condition,eq:async-condition3,eq:async-condition2} or \Cref{eq:sleepy-req,eq:async-condition,eq:async-condition3,eq:async-condition2}, depending on the protocol.

\subsection{Security}
\label{sec:security}

\paragraph*{Security Parameters.}
In this work we treat $\lambda$ and $\kappa$\footnote{{In this work, the value $\kappa$ represents the number of slots that are required in order to be certain, except for a negligible probability, that at least one of the proposers in these slots is honest.}} as the security parameters related to the cryptographic components utilized by the protocol and the protocol's own security parameter, respectively. We also account for a finite time horizon, represented as $\Tconf$, which is polynomial in relation to $\kappa$. An event is said to occur with \emph{overwhelming probability} if it happens except with probability which is $\negl(\kappa) + \negl(\lambda)$. The properties of cryptographic primitives hold true with a probability of $\negl(\lambda)$, signifying an overwhelming probability, although we will not explicitly mention this in the subsequent sections of this work.
{We also assume that $\kappa > 1$.} 

\begin{definition}[Safe protocol]
  \label{def:safety}
 We say that a protocol outputting a confirmed chain $\Chain$
 \tht{When we rewrite the introduction, it is good to provide the intuition of a confirmed chain and an available chain.}
 is \emph{safe} after time $\Tafter$ in executions $\mathcal{E}$, if and only if for any execution in $\mathcal{E}$, two rounds $r, r' \geq \Tafter$, and any two honest validators $v_i$ and $v_j$ (possibly $i=j$) at rounds $r$ and $r'$ respectively, either $\Chain_i^r \preceq \Chain_{j}^{r'}$ or $\Chain_j^{r'} \preceq \Chain_i^r$.

 A protocol satisfies \emph{$\eta$ Safety} after time $\Tafter$ if it is safe after time $\Tafter$ in any $\eta$-compliant execution.
 We say that a protocol always satisfies $\eta$ Safety if it satisfies $\eta$ Safety after time $0$.
 If $\Tafter = 0$ and $\mathcal{E}$ includes all partially synchronous executions, then we say that a protocol is \emph{always safe}.
\end{definition}

\begin{definition}[Live protocol]
  \label{def:liveness}
 {We say that a protocol outputting a confirmed chain $\Chain$ is \emph{live} after time $\Tafter$ in executions $\mathcal{E}$, and has confirmation time $\Tconf$, if and only if for any execution in $\mathcal{E}$, any rounds $r \geq \Tafter$ and $r_i \geq r+\Tconf$, any transaction $\tx$ in the transaction pool at time $r$, and any validator $v_i$ active in round $r_i$,  $\tx \in \Chain^{r_i}_i$.}

  A protocol satisfies \emph{$\eta$ Liveness} after time $\Tafter$ with confirmation time $\Tconf$ if it is live after time $\Tafter$ in any $\eta$-compliant execution and has confirmation time $\Tconf$.
  We say that a protocol always satisfies $\eta$ Liveness with confirmation time $\Tconf$ if it satisfies $\eta$ Safety after time $0$ and has confirmation time $\Tconf$.
  If $\Tafter = 0$ and $\mathcal{E}$ includes all partially synchronous executions, then we say that a protocol is \emph{always live}.
\end{definition}

\begin{definition}[Secure protocol~\cite{goldfish}]
 \label{def:security}
We say that a protocol outputting a confirmed chain $\Chain$ is \emph{secure} after time $\Tafter$, and has confirmation time $\Tconf$, if it is
  \begin{itemize}
    \item safe after time $\Tafter$ and
    \item live after time $\Tafter$ with confirmation time $\Tconf$.
  \end{itemize}
A protocol (always) satisfies $\eta$ Security with confirmation time $\Tconf$ if it (always) satisfies both $\eta$ Safety and $\eta$ Liveness with confirmation time $\Tconf$.
A protocol is always secure if it always both safe and live
\end{definition}

We now recall the definitions of \emph{Dynamic Availability} and \emph{Reorg Resilience} from~\cite{rlmd}.

\begin{definition}[Dynamic Availability]
 \label{def:dyn-ava}
We say that a protocol is $\eta$-\emph{dynamically-available} after time $\Tdyn$ if and only if it satisfies $\eta$ Security after time $\Tdyn$ with confirmation time $\Tconf = O(\kappa)$.

We simply say that a protocol is $\eta$-\emph{dynamically-available} if and only if it always satisfies $\eta$ Security when $\GST = 0$.

Moreover, we say that a protocol is \emph{dynamically-available} if it is $1$-dynamically-available, as this corresponds to the usual notion of Dynamic Availability.
\end{definition}

Note that we define the concept of Dynamic Availability with the underlying assumption of a certain level of adversarial resilience. In the first protocol, as discussed in Section~\ref{sec:ga-based}, we attain Dynamic Availability provided that Constraint~\eqref{eq:pi-sleepiness} is satisfied. Following this, in Section~\ref{sec:rlmd-based}, Dynamic Availability is achieved on the satisfaction of Constraint~\eqref{eq:sleepy-req}.

\begin{definition}[Reorg Resilience]
  \label{def:reorg-resilience}
  We say that a protocol is \emph{reorg-resilient} after slot $t_{\mathsf{reorg}}$ and time $T_{\mathsf{reorg}}$\footnote{{In this work, we always have $\Treorg \geq 4\Delta\treorg$.}} in executions $\mathcal{E}$ if and only if, 
  for any execution in $\mathcal{E}$, round $r \geq T_{\mathsf{reorg}}$ and validator $v_i$ honest in $r$, {any chain} proposed in any slot $t\geq t_{\mathsf{reorg}}$
  by a validator honest in round $\proposing{t}$ does not conflict with $\Chain^r_i$.

  We simply say that a protocol is \emph{reorg-resilient} if and only if it is reorg-resilient after slot 0 and time 0.
  
  A protocol is \emph{$\eta$-reorg-resilient} if it is reorg-resilient in any $\eta$-compliant execution.
\end{definition}



Together with Dynamic Availability, we want our protocol to be \emph{accountably safe}.

\begin{definition}[Accountable Safety]
    \label{def:acc-safety}
We say that a protocol has \emph{Accountable Safety} with resilience $f^\mathrm{acc} > 0$, or that it is $f^\mathrm{acc}$-accountable, if, upon a safety violation, {by having access to all messages sent,} it is possible to identify at least $f$ responsible participants. In particular, {by having access to all messages sent,} it is possible to collect evidence from sufficiently many honest participants and generate a cryptographic proof that identifies $f^\mathrm{acc}$ adversarial participants as protocol violators. Such proof cannot falsely accuse any honest participant that followed the protocol correctly.
\end{definition}

As a consequence of the CAP theorem~\cite{DBLP:journals/corr/abs-2304-14701}, no consensus protocols can satisfy both liveness (under dynamic participation) and safety (under temporary network partitions). Simply put, a consensus protocol (for state-machine replication) cannot produce a single chain that concurrently offers Dynamic Availability and guarantees transaction finality in case of asynchronous periods or network partitions. To overcome this impossibility result, Neu, Tas, and Tse~\cite{DBLP:conf/sp/NeuTT21} introduce a family of protocols, referred to as \emph{ebb-and-flow} protocols. Neu \emph{et al.}~\cite{DBLP:conf/sp/NeuTT21} propose a protocol that outputs two chains, one that provides liveness under dynamic participation (and synchrony), and one that provides Accountable Safety even under network partitions. This protocol is called \emph{ebb-and-flow} protocol. We present a generalization of it, in the $\eta$-sleepy model.

\rs{Should we parametrize by the Accountable Safety threshold as well? \eg, $(f, \eta)$-secure}
\begin{definition}[$\eta$-secure ebb-and-flow protocol]
A $\eta$-secure \emph{ebb-and-flow protocol} outputs an available chain $\chainava$ that is $\eta$-dynamically-available\footnote{When we say that a chain produced by a protocol $\Pi$ satisfies property $P$ (or \emph{the chain is $P$}), we mean that the protocol $\Pi$ itself satisfies property $P$ (or \emph{the protocol is $P$}).}, and a finalized (and accountable) chain $\chainfin$ that is always safe and is live after $\max(\GST,\GAT){+O(\Delta)}$ with $\Tconf = O(\kappa)$\footnote{The reason for this is that we need $\kappa$ slots to find a slot with an honest proposal, after which the proposal will be finalized in $12\Delta$ rounds.}, therefore is secure after $\max(\GST,\GAT){+O(\Delta)}$ with $\Tconf = O(\kappa)$. Moreover, for each honest validator $v_i$ and for every round $r$, $\chainfin_i^r$ is a prefix of $\chainava_i^r$.
\end{definition}

Both our faster finality protocols adopt the ebb-and-flow methodology~\cite{ebbandflow}: the available chain $\chainava$ is output by the $\eta$-dynamically-available protocol (or \emph{component}), while the finalized chain $\chainfin$, is output by a PBFT-style protocol whose behavior is similar to that of FFG Casper~\cite{casper}.

\paragraph{Asynchrony Resilience.}
We define the set $W_r$ of to be the set of active honest validator in round $r$, but restricted to also be in $H_{\voting{t_a}}$ if $\slot(r)$ in $[t_a, t_a+\pi+1]$, \ie,
\[
W_r := \begin{cases}
  H_r \cap H_{\voting{t_a}},&\text{if $\slot(r)$ in $[t_a, t_a+\pi+1]$}\\
  H_r,&\text{otherwise}.
\end{cases}
\]
{Informally, we say that any validator in such a set is an \emph{aware} validator.}
Then, we define $\eta$ Asynchrony Reorg Resilience and $\eta$ Asynchrony Safety Resilience as follows~\cite{rlmd}.
\begin{definition}[$\eta$ Asynchrony Reorg Resilience]
  \label{def:async-resilience}
  If $\pi > 0$,
  we say that a protocol is \emph{$\eta$-asynchrony-reorg-resilient} after slot $t_{\mathsf{reorg}}$ and time $T_{\mathsf{reorg}}$ if and only if the following holds for any $\eta$-compliant execution.
  For any slot $r_i\geq T_{\mathsf{reorg}}$ and validator $v_i \in W_{r_i}$, a {chain} proposed in a slot $t \in [t_{\mathsf{reorg}},t_a]$ by a validator honest in round $\proposing{t}$ does not conflict with $\Chain^r_i$.

  We simply say that a protocol is \emph{$\eta$-asynchrony-reorg-resilient} if and only if it is $\eta$-asynchrony-reorg-resilient after slot $0$ and time $0$.
\end{definition}

\begin{definition}[$\eta$ Asynchrony Safety Resilience]
  \label{def:async-safety-resilience}
  If $\pi > 0$,
  we say that a protocol is \emph{$\eta$-asynchrony-safety-resilient} after time $\Tafter$ if and only if the following holds for any $\eta$-compliant execution.
  Let $r_i$ be any round such that $r_i \in [\Tafter,4 \Delta t_a + \Delta]$ and $v_i$ any validator honest in $r_i$.
  For any slot $r_j$ and validator $v_j$ in $W_{r_j}$, $\Chain^{r_j}_j$ does not conflict with $\Chain^{r_i}_i$.

  We simply say that a protocol is \emph{$\eta$-asynchrony-safety-resilient} if and only if it is $\eta$-asynchrony-safety-resilient after time 0.
\end{definition}


\section{Common Notions}
\label{sec:common-notions}

\paragraph*{Checkpoints.} 

A \emph{checkpoint} is a kind of tuple, represented as \( \C = ({\chain}, c) \).
In this tuple, \(\chain \) is a chain, \( c \) signifies the slot where \( \chain \) is proposed for justification (this concept is introduced and explained below).
Observe that, as a consequence of our protocol design, the slot \( c \) for the checkpoint will always occur after the slot \( \chain.p \) where the chain was proposed, i.e., \( c {\ge} \chain.p \)
.
We refer to \( c \) as the \emph{checkpoint slot} of~\( \C \).

\paragraph*{Votes.}
Validators cast two main types of votes: \textsc{ffg-vote}s and \textsc{vote}s. 
Each \textsc{vote} include an \textsc{ffg-vote}.
Specifically, an \textsc{ffg-vote} is represented as a tuple \([\textsc{ffg-vote}, \C_1, \C_2, v_i]\), where {$v_i$ is the validator sending the \textsc{ffg-vote}, while} \(\C_1\) and \(\C_2\) are checkpoints.
These checkpoints are respectively referred to as the \emph{source} (\(\C_1\)) and the \emph{target} (\(\C_2\)) of the \textsc{ffg-vote}.
{Such an \textsc{ffg-vote} is \emph{valid} (\ie, $\mathrm{valid}(\C_1 \to \C_2$)) if and only if}
\(\C_1.c < \C_2.c\) and  \(\C_1.\chain \preceq \C_2.\chain\). 
\textsc{ffg-vote}s effectively act as \emph{links} connecting the source and target checkpoints. We sometimes denote the whole \textsc{ffg-vote} simply as \(\C_1 \to \C_2\). For instance, we might say that a validator \(v_i\) casts a \(\C_1 \to \C_2\) vote. On the other hand, a \textsc{vote} cast by a validator~$v_i$ is a tuple \([\textsc{vote}, \chain, \C_1 \to \C_2, t, v_i]\), where \(\chain\) represents a chain, \(\C_1 \to \C_2\) encapsulates the associated \textsc{ffg-vote} \([\textsc{ffg-vote}, \C_1, \C_2, v_i]\), and \(t\) is the slot in which such vote is cast.
{In this case, we might say that $v_i$ casts a \textsc{vote} message for chain $\chain$.
We say that a validator $v_i$ \emph{equivocates} if and only if it sends two \textsc{vote} messages $[\textsc{vote}, \chain, \cdot, t, v_i]$ and $[\textsc{vote}, \chain', \cdot, t, v_i]$ with $\chain \neq \chain'$, \ie, it casts two \textsc{vote} messages for the same slot but different chains.

\paragraph*{Proposals.}
This work contains protocols using two types of \textsc{proposal} messages. For the first protocol (Section~\ref{sec:ga-based}), a proposal is a tuple {$[\textsc{propose}, \chain_p, \chain^C, Q^C, \C, t,v_k]$ where $\chain_p$ is the \textsc{propose}d chain,} $Q^C$ is a \emph{quorum certificate} for the \emph{fast confirmed} $\chain^C$, $\mathcal{C}$ is a checkpoint, and $t = \chain_p.p$\footnote{\label{fn:redund}Note that in this context, the parameter $t$ is redundant in the \textsc{propose} message since it is already incorporated within $\chain_p$. However, for the sake of clarity, we have chosen to include it explicitly in the \textsc{propose} message.} is the slot in which this proposal is cast.
In this case, we say that validator $v_k$ \textsc{propose}s chain $\chain_p$ in slot $t$.
The notion of fast confirmed chain will be defined in Section~\ref{sec:revisiting-tob}.

Subsequently, in Section~\ref{sec:rlmd-based}, a proposal is a tuple [\textsc{propose}, $\chain_p$, $\V$, $t$,$v_k$] where $\chain_p$ is a chain (as above), $\V$ a view, and $t=B.p$. We refer to $\V$ as the \emph{proposed view}.

\paragraph*{Gossip behavior.}
Votes and blocks are gossiped at any time, regardless of whether they are received directly or as part of another message. For example, a validator receiving a vote also gossips the block that it contains, and a validator receiving a proposal also gossips the blocks and votes contained in the proposed view. Finally, a proposal from slot $t$ is gossiped only during the first $\Delta$ rounds of slot~$t$.

\paragraph*{Joining protocol.}
\label{sec:joining-protocol}

Honest validators that become awake at round $r$, before starting to participate in the protocol, must first execute (and terminate) a \emph{joining protocol}, after which they become \emph{active}. 
Given slot~$t$, when an honest validator $v_i$ wakes up at some round $r \in (\voting{(t-2)}+\Delta,\voting{(t-1)}+\Delta]$, as per our system model, all the messages that it received while being asleep are immediately delivered to it.
Then, validator $v_i$ executes the protocol, but without sending any message, up until round $\voting{t}$ at which point it becomes active, until either corrupted or put to sleep by the adversary.
This also implies that if a validator $v_i$ is elected to be the leader in a given slot, but by the propose time of that slot it is not active yet, then $v_i$ will not send any \textsc{propose} message in that slot.

\paragraph{Fork-choice function.} A \emph{fork-choice function} is a deterministic function, denoted as $\FC$. This function, when given a set of views\footnote{The number of views used as input varies depending on the protocol. This distinction will become evident when discussing the two different fork-choice functions in the following sections.}, a chain, and a slot $t$ as inputs, produces a chain $\chain$. 
In this work we will focus our attention on two types of fork-choice functions. For the first protocol  (Section~\ref{sec:ga-based}), we consider a \emph{majority} fork-choice function, i.e., the outputted chain is the highest chain supported by a majority of the voting weight (\Cref{alg:mfc}).  
Subsequently,  in Section~\ref{sec:rlmd-based}, we consider a fork-choice function based on $\ghost$~\cite{ghost} (\Cref{alg:rlmd-fc}).

\paragraph*{Confirmation rule.}
A confirmation rule allows validators to identify a \emph{confirmed prefix} of the chain outputted by the fork-choice function, for which safety properties hold, and which is therefore used to define the output of the protocol. Since the protocol we are going to present outputs two chains, the {available chain} $\chainava$ and the {finalized chain} $\chainfin$, output by the finality component, we have two confirmation rules. One is \emph{finality}, which we introduced in Section~\ref{sec:ffg}, and defines $\chainfin$. The other confirmation rule, defining $\chainava$, is itself essentially split in two rules, a \emph{slow} $\kappa$-deep confirmation rule, which is live also under dynamic participation, and a \emph{fast (optimistic) confirmation rule}, requiring $\frac{2}{3}n$ honest validators to be awake, \ie, a stronger assumption than just $\eta-$compliance. $\chainava$ is updated to the chain confirmed by either one, so that liveness of $\chainava$ only necessitates liveness of one of the two rules. In particular, $\eta$-compliance is sufficient for liveness. On the other end, safety of $\chainava$ requires both rules to be safe.

\section{FFG component}
\label{sec:ffg}

In \Cref{sec:common-notions}, we described two kind of votes that a validator $v_i$ casts in the protocol. In particular, the \textsc{ffg-vote}, encapsulated within the generic \textsc{vote}, is used by the \emph{FFG-component} of our protocol\footnote{The component of our protocol that outputs $\chainfin$ is almost identical to Casper~\cite{casper}, the \emph{friendly} finality gadget (FFG) adopted by the Ethereum consensus protocol Gasper~\cite{gasper}. This is the reason why we decided to use the \emph{FFG} terminology already accepted within the Ethereum ecosystem.}. The FFG component of our protocol aims at finalizing{, in each slot, a chain that extends the one finalized in the previous slot}.

\paragraph*{Justification.}
We say that a set of \textsc{ffg-vote}s is a \emph{supermajority set} if it contains valid \textsc{ffg-vote}s from at least \(\frac{2}{3}n\) distinct validators.
A checkpoint \(\C\) is considered \emph{justified} if it either corresponds to the genesis block, i.e., \(\C = (B_\text{genesis}, 0)\), or if there exists a supermajority set of links \(\{\C_i \to \C_j\}\) satisfying the following conditions. First,  for each link $\C_i \to \C_j$ in the set, {\(\C_i \to \C_j\) is valid and} \(\C_i.\chain \preceq \C.\chain \preceq \C_j.\chain\). Moreover, all source checkpoints \(\C_i\) in these links need to be already justified, and the checkpoint slot of \(\C_j\) needs to be the same as that of \(\C\) (\(\C_j.c=\C.c\)), for every \(j\). It is important to note that the source and target chain may vary across different votes.
This justification rule is formalized by the binary function $\mathsf{J}(\V,\C)$ in \Cref{alg:justification-finalization} which outputs \true{} if and only if checkpoint $\C$ is justified according to the set of messages $\V$.
Lastly, we say that a {chain \(\chain\)} \emph{is justified} if and only if there exists a justified checkpoint \(\C\) for which {\(\C.\chain = \chain\)}.

\begin{algo}[t!]
  \caption{Justification and Finalization}
  \label{alg:justification-finalization}
  \vbox{
  \small
  \begin{numbertabbing}\reset
    xxxx\=xxxx\=xxxx\=xxxx\=xxxx\=xxxx\=MMMMMMMMMMMMMMMMMMM\=\kill
  {\textbf{function} $\mathrm{valid}(\C_1 \to \C_2)$ }\label{}\\
  \> {\textbf{return} }\label{}\\
  \>\>{$\land\; \C_1.\chain \preceq \C_2.\chain$}\label{}\\
  \>\>{$\land\; \C_1.c < \C_2.c$}\label{}\\
  \\
  \textbf{function} $\mathsf{J}(\C,\V)$ \label{}\\

    \> \textbf{return}\label{}\\
    \>\> $\lor\; \C = (\genesis,0)$\label{}\\
    \>\> $\lor\; \exists \M \subseteq \V: \land\; |\{v_k : [\textsc{vote},\cdot,\cdot,\cdot,v_k] \in \M\}|\geq \frac{2}{3}n$
    \label{}\\
    \>\> $\hphantom{\lor\; \exists \M \in \V:}\land\; \forall\; [\textsc{vote},\cdot,\calS \to \T,\cdot,\cdot] \in \M: \land\;\mathrm{valid}(\calS\to\T)$\label{}\\
    \>\> $\hphantom{\lor\; \exists \M \in \V:\land\; \forall\; [\textsc{vote},\cdot,\calS \to \T,\cdot,\cdot] \in \M:}\, \land\;\mathsf{J}(\calS,\V)$\label{}\\
    \>\> $\hphantom{\lor\; \exists \M \in \V:\land\; \forall\; [\textsc{vote},\cdot,\calS \to \T,\cdot,\cdot] \in \M:}\, \land\;\calS.\chain \preceq \C.\chain\preceq \T.\chain$\label{}\\
    \>\> $\hphantom{\lor\; \exists \M \in \V:\land\; \forall\; [\textsc{vote},\cdot,\calS \to \T,\cdot,\cdot] \in \M:}\, \land\;\T.c = \C.c$\label{}\\
    \\

  \textbf{function} $\mathsf{F}(\C,\V)$ \label{}\\
    \> \textbf{return}\label{}\\
    \>\> $\lor\; \C = (\genesis,0)$\label{}\\
    \>\> $\lor\; \land\; \mathsf{J}(\C,\V)$\label{}\\
    \>\> $\hphantom{\lor\;}\land\; \exists \M \subseteq \V: \land\; |\{v_k : [\textsc{vote},\cdot,\cdot,\cdot,v_k] \in \M\}|\geq \frac{2}{3}n$\label{}\\
    \>\> $\hphantom{\lor\;\land\; \exists \M \in \V:}\land\; \forall\; [\textsc{vote},\cdot,\C \to \T,\cdot,\cdot] \in \M: \land\;\mathrm{valid}(\C\to\T)$\label{}\\
    \>\> $\hphantom{\lor\;\land\; \exists \M \in \V:\land\; \forall\; [\textsc{vote},\cdot,\C \to \T,\cdot,\cdot] \in \M:}\, \land\;\T.c = \C.c+1$\label{}
    \\[-5ex]        
  \end{numbertabbing}
  }
\end{algo}

\begin{figure}[h]
\centering
\includegraphics[width=0.7\textwidth]{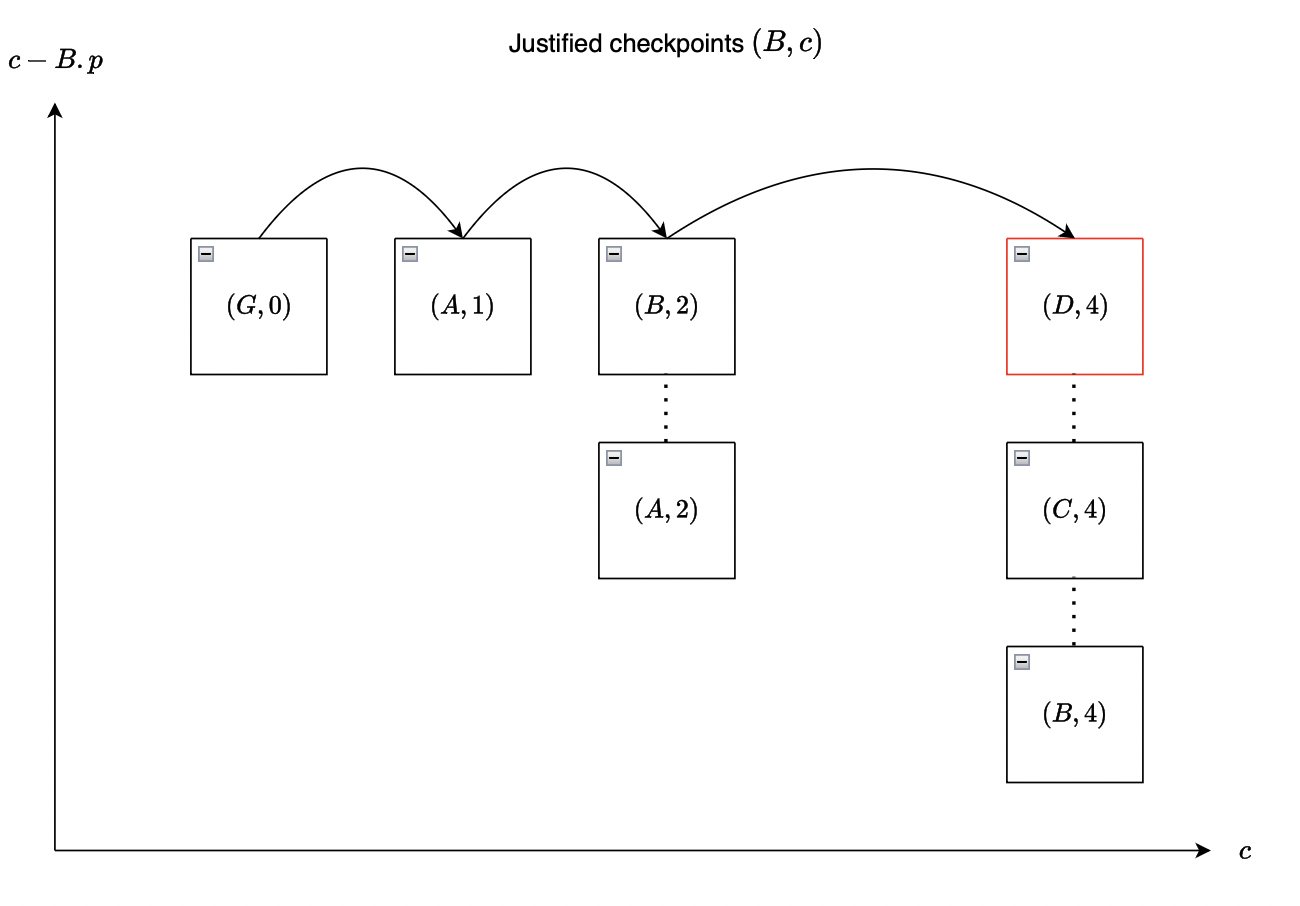}
\caption{\rsnotfn{Improve pictures.} Example of justification. The figure illustrates the progression of justified checkpoints over four slots. The checkpoints are presented from left to right, indicating an increase in the checkpoint slot~$c$. Within each checkpoint slot, the proposal slot~$p$ decreases from top to bottom. This visualization is based on the assumption that all justifications occur through votes with consistent source and target, as shown by the arrows. {Chains} labeled $A, B, C,$ and $D$ are \textsc{propose}d in slots $0, 1, 2,$ and $3$, respectively, {with each chain extending the previous one}. Specifically, at slot 2, votes originating from the source checkpoint $(A,1)$ and targeting checkpoint $(B,2)$ result in the justification of both checkpoints $(B,2)$ and $(A,2)$. Slot 3 does not see any checkpoint justification. However, in slot 4, checkpoints containing chains $B, C,$ and $D$ are justified. This justification is achieved through votes with the source $(B,2)$ and the target $(D,4)$.}
\label{fig:your_label}
\end{figure}

\paragraph*{Ordering among checkpoints.}
We define a {total pre-order} among checkpoints \(\C\) and \(\C'\) by lexicographic ordering on checkpoint and proposal slots. Specifically,  {$\C \leq \C'$ if and only if either \(\C.c < \C'.c\) or, in the case where \(\C.c = \C'.c\), then \(\C.\chain.p \leq \C'.\chain.p\). 
However, when we write $\C = \C'$ we mean that the two checkpoints are the same, not just that $\C.c = \C'.c \land \C.\chain.p=\C'.\chain.p$.}

\paragraph*{Greatest justified checkpoint and chain.}
A checkpoint is considered justified in a view \(\V\) if \(\V\) contains a supermajority set of links justifying it. A \emph{justified} checkpoint which is maximal in \(\V\) with respect to the previously defined lexicographic ordering is referred to as the \emph{greatest justified checkpoint} in \(\V\), denoted as \(\GJ(\V)\). In the event of ties, they are broken arbitrarily. Chain \(\GJ(\V).\chain\) is referred to as the \emph{greatest justified chain}.

\paragraph*{Finality.}
A checkpoint \(\C\) is \emph{finalized} if it is justified and there exists a supermajority link with source \(\C\) and potentially different targets \(\C_i\) where \(\C_i.c = \C.c + 1\). A chain \(\chain\) is finalized if there exists a finalized checkpoint \(\C\) with \(\chain = \C.\chain\). The checkpoint \(\C = (B_\text{genesis}, 0)\) is finalized by definition.
This finalization rule is formalized by the binary function $\mathsf{F}(\V,\C)$ in \Cref{alg:justification-finalization} which outputs \true{} if and only if checkpoint $\C$ is finalized according to the set of messages $\V$.
{Given a view $\V$, a finalized checkpoint which is maximal in \(\V\) with respect to the previously defined lexicographic ordering is referred to as the \emph{greatest finalized checkpoint} in \(\V\), denoted as \(\GF(\V)\).
In the event of ties, they are broken arbitrarily.
We say that a chain $\chain$ is finalized according to a view $\V$ if and only if $\chain \preceq \GF(\V).\chain$.

\paragraph*{Slashing.}
A validator \(v_i\) is subject to slashing (as introduced in Section~\ref{sec:model}) for two \emph{distinct} \textsc{ffg-vote}s \(\C_1 \to \C_2\) and \(\C_3 \to \C_4\) if either of the following conditions holds: {\(\mathbf{E_1}\) (Double voting)} if \(\C_2.c = \C_4.c\), implying that a validator must not cast distinct \textsc{ffg-vote}s for the same checkpoint slot; or {\(\mathbf{E_2}\) (Surround voting)} if \(\C_3 < \C_1\)
according to the lexicographic ordering on checkpoint and proposal slots, and \(\C_2.c < \C_4.c\), indicating that a validator must not vote using a lower checkpoint as source and must avoid voting within the span of its other votes\footnote{In the implementation, for efficiency purposes, \textsc{ffg-vote}s do not include  a chain, but only its hash.
Then, given that we want to be able to determine whether a  validator committed a slashable offense just by looking at the \textsc{ffg-vote}s that it sent, in the implementation, checkpoints are represented by triples $\mathcal{C}=(H, c, p)$ where $H$ is the hash of a chain $\chain$ and $p=\chain.p$.
A checkpoint is valid only if $\chain.p = p$, and an \textsc{ffg-vote} is valid only if both checkpoints are valid.}.

\begin{figure}[h]
\centering
\includegraphics[width=0.7\textwidth]{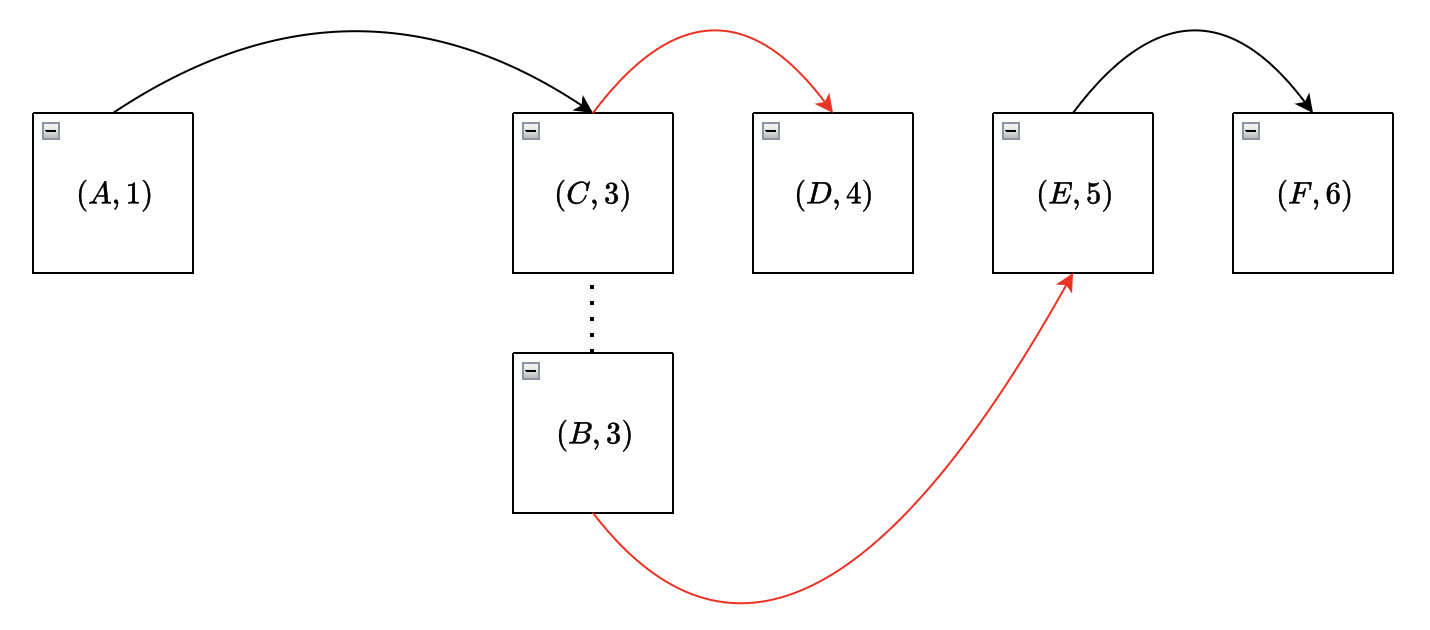}
\caption{The figure illustrates an example of surround voting, where the slashable votes are highlighted in red. In this scenario, from slot 4 to slot 3, the source checkpoint shifts ``backwards" from $(C,3)$ to a lower checkpoint $(B,3)$. This instance of surround voting is unique to this protocol and does not occur in Casper-FFG~\cite{casper}, due to the two source checkpoints sharing the same checkpoint slot.}
\label{fig:your_label}
\end{figure}

\subsection{Analysis}
\label{sec:ffg-analysis}
In this section, first we show that the finalized chain is accountably safe, exactly as done in Casper~\cite{casper}.
Then, we provide a set of conditions that, if satisfied, ensure that chains \textsc{propose}d by honest validators are finalized within two slots.
In \Cref{sec:rlmd-based,sec:ga-based}, we will show that the algorithms presented in each of these sections ensure such properties after $\max(\GST, \GAT) + 4\Delta$.


\begin{lemma}
\label{lem:accountable-jutification}
Assume $f \in [0,n]$\lz{is it misleading to use $f$ here as adversarial resilience?} and let
$\C_\mathsf{f}=(\chain_\mathsf{f},c_\mathsf{f})$ be a finalized checkpoint and $\C_\mathsf{j}=(\chain_\mathsf{j},c_\mathsf{j})$ be a justified checkpoint with $c_\mathsf{j} \geq c_\mathsf{f}$ {according to any two views}.
Either $\chain_\mathsf{j} \succeq \chain_\mathsf{f}$ or at least $\frac{n}{3}$ validators can be detected to have violated either $\mathbf{E_1}$ or $\mathbf{E_2}$.
\end{lemma}

\begin{proof}
{First, we show} that no checkpoint $\C'=(\chain',c_\mathsf{f})$, with $\chain'$ conflicting with $\chain_\mathsf{f}$, {can ever be justified.}
If that is not the case, clearly $\geq \frac{n}{3}$ validators are slashable for violating $\mathbf{E_1}$: For the justification of \((\chain_\mathsf{f},c_\mathsf{f})\), {it is required to have a supermajority of \textsc{ffg-vote}s with the chain of the target checkpoint being an extension of $\chain_\mathsf{f}$}. Similarly, the justification of \((\chain',c_\mathsf{f})\) {requires to have a supermajority of \textsc{ffg-vote}s with the chain of the target checkpoint being an extension of $\chain'$}. Given that {any descendant of $\chain_\mathsf{f}$ cannot be a descendant of $\chain'$, and vice versa,} and that we need a supermajority set of links for justification, the intersection of the sets of voters contributing to the justification of $\C$ and $\C'$ include at least $\frac{n}{3}$ validators {which have sent two \textsc{ffg-vote}s $\calS_\mathsf{f} \to \T_\mathsf{f}$ and $\calS' \to \T'$ with $\T_\mathsf{f}.c = \T'.c = c_\mathsf{f}$ and $\T_\mathsf{f}.\chain \neq \T'.\chain$ thereby violating condition $\mathbf{E_1}$}.

Now, by contradiction, assume that $\chain_\mathsf{j} \not\succeq \chain_\mathsf{f}$ and that there does not exist a set of at least $\frac{n}{3}$ validators that can be detected to have violated either $\mathbf{E_1}$ or $\mathbf{E_2}$.
Let \(c'_\mathsf{j} > c_\mathsf{f}\) be the smallest slot for which a checkpoint \(\C'_\mathsf{j} = (\chain'_\mathsf{j}, c'_\mathsf{j})\) is justified with $\chain'_\mathsf{j} \not\succeq \chain_\mathsf{f}$, \ie, either $\chain'_\mathsf{j}$ conflicts with $\chain_\mathsf{f}$ or $\chain'_\mathsf{j}$ is a strict prefix of $\chain_\mathsf{f}$. Given our assumptions we know that one such a checkpoint exists.

Let $(A_i,c_i) \to (B_i,c'_\mathsf{j})$ and $(\chain_\mathsf{f},c_\mathsf{f}) \to (C,c_\mathsf{f}+1)$ be the \textsc{ffg-vote}s involved in the justification of $(\chain'_\mathsf{j},c'_\mathsf{j})$ and in the finalization of $\chain_\mathsf{f}$, respectively, cast by a validator~$v_i$.
By definition of justification we have that~$A_i \preceq \chain'_\mathsf{j} \preceq B_i$. We observe two cases:

\begin{description}
    \item[Case 1: $c'_\mathsf{j}=c_\mathsf{f}+1$.] If $\chain'_\mathsf{j}$ conflicts with $\chain_\mathsf{f}$, then $A_i \preceq \chain'_\mathsf{j}$ implies $A_i \neq \chain_\mathsf{f}$, and thus the two votes are different. Conversely, if $\chain'_\mathsf{j}$ is a strict prefix of $\chain_\mathsf{f}$, then $A_i \preceq \chain'_\mathsf{j} {\prec} \chain_\mathsf{f}$, and thus $A_i \neq \chain_\mathsf{f}${, which implies that the two votes are different in this case as well}. {Hence, in} both cases, validator~$v_i$ {violates condition $\mathbf{E_1}$ and therefore} is slashable.

    \item[Case 2: $c'_\mathsf{j}>c_\mathsf{f}+1$.] We have to consider three cases:

\begin{description}
    \item[Case 2.1: $c_i > c_\mathsf{f}$.] 
    {By the definition of justification, if the checkpoint $(\chain'_\mathsf{j},c'_\mathsf{j})$ is justified, then the checkpoint $(A_i, c_i)$ must also be justified.} However, {given that $c_i < c'_\mathsf{j}$,} under the assumption of the minimality of $c'_\mathsf{j}$, we have $\chain_\mathsf{f} \preceq A_i$. Given that $A_i \preceq \chain'_\mathsf{j}$, this leads to the contradiction $\chain_\mathsf{f} \preceq \chain'_\mathsf{j}$.

    \item[Case 2.2: $c_i = c_\mathsf{f}$.] 
    {As argued in the proof of the case above, the checkpoint $(A_i, c_i)$ is justified.} 
    {As proved at the beginning of this proof, there cannot exist any justified checkpoint $\C'=(\chain',c_\mathsf{f})$, with $\chain'$ conflicting with $\chain_\mathsf{f}$.}
    {Consequently, either $\chain_\mathsf{f} \preceq A_i$ holds, which implies either $\chain_\mathsf{f}  \preceq \chain'_\mathsf{j}$,  contradicting our assumption that $\chain_\mathsf{f} \npreceq \chain'_\mathsf{j}$, or $A_i \prec \chain_\mathsf{f}$ which implies  $A_i.p < \chain_\mathsf{f}.p$ from which we have that $(c_i,A_i.p)=(c_\mathsf{f},A_i.p) < (c_\mathsf{f},\chain_\mathsf{f}.p)$ and $c_\mathsf{f}+1 < c'_\mathsf{j}$, which violates condition $\mathbf{E_2}$.}

    \item[Case 2.3: $c_i < c_\mathsf{f}$.] 
    Considering $c_i < c_\mathsf{f} < c_\mathsf{f}+1 < c'_\mathsf{j}$, this situation constitutes a violation of $\mathbf{E_2}$ due to the existence of surrounding voting.
\end{description}
\end{description}
{Given that justifications require supermajority link, the intersection of the set of voters involved in the finalization of $\C_\mathsf{f}$ and justification of $\C'_\mathsf{j}$ includes at least $\frac{n}{3}$ validators which, by the reasoning outlined above, would violated either $\mathbf{E_1}$ or $\mathbf{E_2}$.}
This leads to contradicting our assumptions which concludes the proof.
\end{proof}

\begin{lemma}
\label{lem:accountable-safety}
Assume $f \in [0,n]$.
If two conflicting chains are finalized according to any two respective views, then at least $\frac{n}{3}$ validators can be detected to have violated either $\mathbf{E_1}$ or $\mathbf{E_2}$.
\end{lemma}

\begin{proof}
Let $\chain_1$ and $\chain_2$ be two conflicting chains according to view $\V_1$ and $\V_2$ respectively.
By the definition of finalized chains, this implies that there exists two checkpoints $\C_1 = (\chain'_1, c_1)$ and $\C_2 = (\chain'_2, c_2)$ that are finalized according to view $\V_1$ and $\V_2$ respectively, such that $\chain'_1 \succeq \chain_1$ conflicts with $\chain'_2 \succeq \chain_2$.
Assume without loss of generality that $c_2 \geq c_1$.
Given that finalization implies justification, we can apply \Cref{lem:accountable-jutification} to conclude that if $\chain'_2 \not\succeq \chain'_1$, then at least $\frac{n}{3}$ validators can be detected to have violated either $\mathbf{E_1}$ or $\mathbf{E_2}$.
\end{proof}

To complete the proof of Accountable Safety, we also need to show that up until a validator follows the protocol, it never gets slashed.
This property cannot be concluded by the finalization and justifications rule alone as it depends on the algorithm employed to cast \textsc{ffg-vote}s.
Then, to keep the treatment of Accountable Safety as decoupled as possible from this, below we detail a constraint 
to be guaranteed algorithm employed to cast \textsc{ffg-vote}s  that then we prove to be sufficient to ensure that honest validators never get slashed.

\begin{constraint} \label{prop:never-slashed}
  The algorithm employed to cast \textsc{ffg-vote}s must guarantee the following conditions.
  \begin{enumerate}[label=\ref*{prop:never-slashed}.\arabic*.,ref=\ref*{prop:never-slashed}.\arabic*]
  \item At most one \textsc{ffg-vote} is sent in a slot. 
  \customlabel[constraint]{prop:never-slashed-1}{\theenumi}
  \item If $\T$ is the target checkpoint in the \textsc{ffg-vote} sent in slot $t$, then $\T.c = t$.
  \customlabel[constraint]{prop:never-slashed-2}{\theenumi}
  \item Let $\calS$ and $\mathcal{S'}$ be the two source checkpoints  in the \textsc{ffg-vote}s sent in slot $t$ and $t'$. If $t \leq t'$, then $\calS \leq \mathcal{S'}$.\customlabel[indhypothesis]{ind:1}{\theindhypothesis.\theenumi}
  \customlabel[constraint]{prop:never-slashed-3}{\theenumi}
  \end{enumerate}
\end{constraint}
  
\begin{lemma} \label{lem:never-slashed-generalized}
  If \Cref{prop:never-slashed} holds, then honest validators are never slashed. 
\end{lemma}

\begin{proof}
  Take any round $r$ and any validator $v_i$ honest in round $r$.
  Take also any two different \textsc{ffg-vote}s $\C_1 \rightarrow \C_2$ and $\C_3 \rightarrow \C_4$ that $v_i$ has sent by round $r$.
  If no such \textsc{ffg-vote}s exists, then clearly $v_i$ cannot be slashed.
  By Constraints~\ref{prop:never-slashed-1} and \ref{prop:never-slashed-2}, we have that $\C_2.c \neq \C_4.c$.
  Hence, $\textbf{E}_\textbf{1}$ is not violated.
  Then, without loss of generality, assume $\C_2.c < \C_4.c$.
  By \Cref{prop:never-slashed-3}, we know that $\C_1  \leq  \C_3$.
  Given that $\C_2.c < \C_4.c$, $\textbf{E}_\textbf{2}$ cannot be violated.
  Therefore, until $v_i$ is honest, it cannot be slashed.
\end{proof}

Again with the aim to keep this analysis as general as possible, we define a second Constraint that is sufficient to ensure Accountable Safety.
\begin{constraint}\label{prop:chfin}
  {Let $\Vglobal^r$ be the global view at time $r$, \ie, the set of all messages sent up to time $r$,
  The chain $\chainfin^r_i$ output by a validator $v_i$ honest in round $r$ is any finalized chain according to view $\Vglobal^r$.}
\end{constraint}

\begin{theorem}[Accountable Safety]
  \label{thm:accountable-safety}
  Let $f \in [0,n]$.
  If \Cref{prop:never-slashed,prop:chfin} hold, then the finalized chain $\chainfin$ is {$\frac{n}{3}$-accountable}.
\end{theorem}
\begin{proof}
  Follows from \Cref{lem:accountable-safety,lem:never-slashed-generalized}.
\end{proof}

Let us now move to discuss Liveness.
Given that also Liveness depends on the protocol employed to propose chains and cast \textsc{ffg-vote}s, like we did above, we provide a constraint to be guaranteed algorithm employed to cast \textsc{ffg-vote}s  that then we prove to be sufficient to guarantee liveness of the FFG component.

\begin{constraint}\label{prop:succ-for-ffg-liveness}
  The algorithm employed to propose chains and cast \textsc{ffg-vote}s must guarantee the following conditions (assuming $f<\frac{n}{3}$).
  \begin{enumerate}[label=\ref*{prop:succ-for-ffg-liveness}.\arabic*.,ref=\ref*{prop:succ-for-ffg-liveness}.\arabic*]
    \item At any slot $t$, if an always
    honest validator $v_i$ casts an \textsc{ffg-vote} $\calS_i \to \T_i$ during round $r$, then:
    \begin{enumerate}[label*=\arabic*.,ref=\theenumi.\arabic*]
      \item\customlabel[constraint]{prop:succ-for-ffg-liveness-1-1}{\theenumii} the vote is valid;
      \item\customlabel[constraint]{prop:succ-for-ffg-liveness-1-2}{\theenumii} There exists a set of messages $\V^{\mathsf{FFGvote},t}_i$ such that $\calS_i = \GJ(\V^{\mathsf{FFGvote},t}_i)$, i.e., the source checkpoint corresponds to the greatest justified checkpoint according to some set of messages $V^{\mathsf{vote},t}_i \subseteq \V^r_i$ in the view of validator $v_i$ at round $r$; and
      \item\customlabel[constraint]{prop:succ-for-ffg-liveness-1-3}{\theenumii} $\T_i.c = t$.
    \end{enumerate}
    \item If a \textsc{propose} message for chain $\chain_p$ is sent in a round $\proposing{t}{\geq \max(\GST,\GAT)+\Delta}$ by a proposer honest in that round, then:
    \begin{enumerate}[label*=\arabic*.,ref=\theenumi.\arabic*]
      \item $chain_p$ includes all the transactions in $\txpool^\proposing{t}$;\customlabel[constraint]{prop:succ-for-ffg-liveness-2-1}{\theenumii}
      \item in slot $t$, any always-honest validator $v_i$ casts a valid \textsc{ffg-vote} $\calS \to \T_i$ with $\T_i.\chain \preceq \chain_p$, \ie, all always-honest validators send \textsc{ffg-vote}s with the same source checkpoint and with a target checkpoint (potentially different for each validator) prefix of chain $\chain_p$;\customlabel[constraint]{prop:succ-for-ffg-liveness-2-2}{\theenumii}
      \item in slot $t+1$:\customlabel[constraint]{prop:succ-for-ffg-liveness-2-3}{\theenumii}
      \begin{enumerate}[label*=\arabic*.,ref=\theenumii.\arabic*]
        \item\customlabel[constraint]{prop:succ-for-ffg-liveness-2-3-2}{\theenumiii} for any always-honest validator $v_i$, $\V^{\mathsf{FFGvote},t+1}_i$ includes $\V^{\mathsf{FFGvote},t}_j$ of any always-honest validator $v_j$, and all the \textsc{ffg-vote} messages sent by always-honest validators in slot $t$; and
        \item\customlabel[constraint]{prop:succ-for-ffg-liveness-2-3-1}{\theenumiii}  any always-honest validator $v_i$ sends an \textsc{ffg-vote} $\calS_i \to \T_i$ with $\T_i.\chain = \chain_p$.
      \end{enumerate}
      \item in slot $t+2$,
      \begin{enumerate}[label*=\arabic*.,ref=\theenumii.\arabic*]
        \item\customlabel[constraint]{prop:succ-for-ffg-liveness-2-4-1}{\theenumiii} 
        for any always-honest validator $v_i$, $\V^{\mathsf{FFGvote},t+2}_i$ includes $\V^{\mathsf{FFGvote},t+1}_j$ of any always-honest validator $v_j$, and all the \textsc{ffg-vote} messages sent by always-honest validators in slot $t+1$; and
        \item\customlabel[constraint]{prop:succ-for-ffg-liveness-2-4-2}{\theenumiii} for any validator $v_i$ honest in round $4\Delta (t+2) + 2 \Delta$, (i) $\V^{4\Delta (t+2) + 2 \Delta}_i$ includes all the \textsc{ffg-vote}s cast by always-honest validators in slot $t+2$ and (ii) $\chainfin^{4\Delta (t+2) + 2 \Delta}_i = \max(\{\chain \colon \chain \preceq \GF(\V^{4\Delta (t+2) + 2 \Delta}_i).\chain \land \chain.p \preceq \chain_p.p\})$.
      \end{enumerate}
    \end{enumerate}
  \end{enumerate}
\end{constraint}

\begin{theorem}[Liveness]
  \label{thm:liveness-ffg-general}
  Assume that \Cref{prop:succ-for-ffg-liveness} holds (and $f< \frac{n}{3}$).
  Let $v_p$ be any validator honest by the time it \textsc{propose}s chain $\chain_p$ in a slot $t$ such that $\proposing{t} \geq \max(\GST, \GAT) + 4\Delta$.
  Then, chain $\chain_p$ is justified at slot~$t+1$, and finalized at slot~$t+2$.
  In particular, $\chain_p \preceq \chainfin^{4\Delta (t+2) + 2 \Delta}_i${ and $\txpool^\proposing{t} \subseteq \chainfin^{4\Delta (t+2) + 2 \Delta}_i$} for any validator $v_i \in H_{4\Delta (t+2) + 2 \Delta}$.
\end{theorem}

\begin{proof}
  Let $t := \slot(r)$ and
  consider any honest validator $v_i$ and the set of messages $\V^{\mathsf{FFGvote},t+1}_i$.
  By $f<\frac{n}{3}$, \Cref{prop:succ-for-ffg-liveness-2-2} and \Cref{prop:succ-for-ffg-liveness-2-3-2}, we know that $\V^{\mathsf{FFGvote},t+1}_i$ includes a set of \textsc{ffg-vote}s from at least $\frac{2}{3}n$ of the validators such that each vote in this set has the same source checkpoint $\calS$ and a target checkpoint prefix of chain $\chain_p$.
  Also, by \Cref{prop:succ-for-ffg-liveness-1-2} and \Cref{prop:succ-for-ffg-liveness-2-3-2}, we know that $\calS$ is justified according to $\V^{\mathsf{FFGvote},t+1}_i$.
  Then, $f<\frac{n}{3}$ and \Cref{prop:succ-for-ffg-liveness-1-3} imply that $\GJ(\V^{\mathsf{FFGvote},t+1}_i).c = t \land \GJ(\V^{\mathsf{FFGvote},t+1}_i).\chain \preceq \chain_p$. 
  Also, by \Cref{prop:succ-for-ffg-liveness-1-2}, $\GJ(\V^{\mathsf{FFGvote},t+1}_i)$ is the source checkpoint of the \textsc{ffg-vote} sent by $v_i$ in slot $t+1$.

  Then, given that by \Cref{prop:succ-for-ffg-liveness-1-3} and \Cref{prop:succ-for-ffg-liveness-2-3} all honest validators in slot $t+1$ send \textsc{ffg-vote}s with target $(\chain_p,t+1)$ and that, by \Cref{prop:succ-for-ffg-liveness-2-4-1}, all of these votes are included in $\V^{\mathsf{FFGvote},t+2}_i$, we have that $\GJ(\V^{\mathsf{FFGvote},t+2}_i) = (\chain_p, t+1)$.

  Again by \Cref{prop:succ-for-ffg-liveness-1-2}, in slot $t+2$, all honest validators then send an \textsc{ffg-vote} with source $(\chain_p, t+1)$.
  By \Cref{prop:succ-for-ffg-liveness-2-4-2}, by round ${4 \Delta (t+2) + 2 \Delta}$, all of these \textsc{ffg-vote}s are in the view of any validator $v_i$ honest at such round.
  Because $f<\frac{n}{3}$, then $\chain_p \preceq \GF(\V^{4 \Delta (t+2) + 2 \Delta}_i).\chain$.
  Again by \Cref{prop:succ-for-ffg-liveness-2-4-2}, $\chain_p \preceq \chainfin^{4 \Delta (t+2) + 2 \Delta}_i$.
  {\Cref{prop:succ-for-ffg-liveness-2-1} then implies that $\txpool^\proposing{t} \subseteq \chainfin^{4\Delta (t+2) + 2 \Delta}_i$.}
\end{proof}

{By the Theorem above, we know that, as long as all the stated assumptions hold, then at time $2\Delta$ into slot $t+2$ every honest validator has finalized the chain $\chain_p$ proposed by an honest proposer in slot $t$.
However, already starting $\Delta$ time earlier, \ie, from $\Delta$ time into slot $t+2$, 
no honest validator will ever finalize a chain conflicting with the chain $\chain_p$ proposed in slot $t$.
This is formalized in the following Corollary.

\begin{corollary}\label{cor:liveness-ffg-general}
  Assume that \Cref{prop:succ-for-ffg-liveness} holds (and $f< \frac{n}{3}$).
  Let
  $\chainfin^r_\mathsf{G}$ be the longest finalized chain according to view $\V^r_\mathsf{G}$ and
  $v_p$ be any validator honest by the time it \textsc{propose}s chain $\chain_p$ in a slot $t$ such that $\proposing{t} \geq \max(\GST, \GAT) + 4\Delta$.
  Then, $\chain_p \preceq \chainfin^{4\Delta (t+2)+\Delta}_\mathsf{G}$.
\end{corollary}
\begin{proof}
  Given \Cref{prop:succ-for-ffg-liveness-2-4-2}, any honest validator send its \textsc{ffg-vote} for slot $t+2$ no later than time $4\Delta (t+2) + \Delta$.
  Then, the proof follows from the proof of \Cref{thm:liveness-ffg-general}.
\end{proof}}

\section{{\TOBSVD}-Based Faster Finality Protocol}
\label{sec:ga-based}

In this section, our goal is to introduce a faster finality protocol for Ethereum that builds upon a dynamically-available protocol based on a modified version of the protocol {\TOBSVD{} by} D'Amato \emph{et al.}~\cite{streamliningSBFT}. 
We start by recalling the deterministically safe, dynamically-available consensus protocol {\TOBSVD} in \Cref{sec:revisiting-tob}. 
We refer the reader to the original paper for technical details. 
Our focus here is to briefly show how the protocol works, the guarantees it provides, and modifications to achieve a probabilistically safe, dynamically-available consensus protocol.
Next, we show how to introduce fast confirmations in the modified consensus protocol and prove its $\eta$ Dynamic Availability which enables its use in a majority-based faster finality protocol in \Cref{sec:prob-tob}. 
Finally, our majority-based faster finality protocol is presented in \Cref{sec:tob-execution}.

\subsection{Revisiting {\TOBSVD}}
\label{sec:revisiting-tob}


{\TOBSVD}~\cite{streamliningSBFT} proceeds as follows.
During each slot~$t$, a proposal (\(B\)) is made in the first round through a [\textsc{propose}, $B$, $t$, $v_i$] message, and a decision is taken in the third round. During the second round, every active validator \(v_i\) casts a \([\textsc{vote}, \chain, \cdot, t, v_i]\) message for a chain \(\chain\)\footnote{Note that we use $\cdot$ for the third component of a \textsc{vote} message as this component is not necessary for this part of the paper. Specifically, this refers to the FFG component in \Cref{alg:3sf-tob-noga}, similar to \Cref{alg:rlmd-ffg}. Since this section only describes a variant of the dynamically-available protocol \TOBSVD by D'Amato \emph{et al.}, we omit the FFG component here and will reintroduce it later when presenting the final algorithm in Section~\ref{sec:ga-based}.}. Let $\V$ be any view and \(t\) be the current slot.
We define:
{
\[
\V^{\chain,t} := \{[\textsc{vote}, \chain', \cdot, \cdot, v_k]: [\textsc{vote}, \chain', \cdot, \cdot, v_k] \in  \V' \land \V' = \mathsf{FIL}_{\text{lmd}}(\mathsf{FIL}_{\text{$\eta$-exp}}(\mathsf{FIL}_{\text{eq}}(\V),t)) \land \chain \preceq \chain'\}
,
\]
}
with 
 \(\mathsf{FIL}_{\text{$\eta$-exp}}\), \(\mathsf{FIL}_{\text{lmd}}\), and \(\mathsf{FIL}_{\text{eq}}\) defined as in \Cref{alg:mfc}. 
Specifically, $\FIL_{\text{eq}}$ removes any  \textsc{vote} message sent by a validator that has equivocated, $\FIL_{\text{$\eta$-exp}}$ retains only the unexpired \textsc{vote} messages, and $\FIL_{\text{lmd}}$ selects only the latest \textsc{vote} messages cast by each validator.
So, intuitively, $\V^{\chain,t}$ corresponds to the set of all latest \textsc{vote} messages in $\V$ that (i) are for a chain extending $\chain$, (ii) are non-expired with respect to slot $t$, and (iii) are sent by validators that have never equivocated.

From here on, we adopt a \emph{majority fork-choice function}, denoted as $\mfc$ and presented in \Cref{alg:mfc}. 
The function \(\mfc(\V, \V', \chain^C, t)\) starts from a chain \(\chain^C\) and returns the longest chain $\chain \succeq \chain^C$ such that \(\left|\V^{\chain,t} \cap (\V')^{\chain,t} \right|> \frac{\left|\mathsf{S}(\V',t)\right|}{2}\), where $\mathsf{S}(\V',t)$
corresponds to the set of validators that, according to view $\V'$, have sent a non-expired \textsc{vote} message in slot $t$, 
\ie, 
the majority of the validators in $\mathsf{S}(\V',t)$ are validators that have never equivocated according to either view $\V'$, and that, according to both views, their latest non-expired \textsc{vote} message is for a chain extending $\chain$.
If such chain does not exist, then it just returns $\chain^C$.

\begin{algo}[h!]
  \caption{$\mfc$, the majority fork-choice function}
  \label{alg:mfc}
  \vbox{
  \small
  \begin{numbertabbing}\reset
    xxxx\=xxxx\=xxxx\=xxxx\=xxxx\=xxxx\=MMMMMMMMMMMMMMMMMMM\=\kill     
  \textbf{function} $\mfc(\V,\V', \chain^C, t)$ \label{}\\
  \> \textbf{return} the longest chain $\chain \succeq \chain^C$ such that $\chain = \chain^C \lor \left|\V^{\chain,t}\cap (\V')^{\chain,t} \right|> \frac{\left|\mathsf{S}(\V',t)\right|}{2}$\label{}\\
  \\
  
  \textbf{function} $\mathsf{S}(\V,t)$\label{}\\
      \> \textbf{return} $ \{v_k: [\textsc{vote}, \cdot, \cdot, \cdot, v_k] \in \V' \land {\V'} = \mathsf{FIL}_{\eta\text{-exp}}(\V,t)\}$ \label{}\\
  \\
  \textbf{function} $\FIL_{\text{lmd}}(\V)$ \label{}\\
  \> \textbf{let} $ \V' := \V \setminus \{[\textsc{vote},\cdot,\cdot,t',v_k] \in \V : \exists [\textsc{vote},\cdot,\cdot,t'',v_k] \in \V: t'' > t'\}$\\
  \> // $\V' := \V$ without all but the most recent (\emph{i.e., latest}) votes of each validator\\
  \> \textbf{return} $\V'$ \label{}\\
  \\
      \textbf{function} $\FIL_{\eta\text{-exp}}(\V, t)$ \label{}\\
    \> \textbf{let} $ \V' := \V \setminus \{[\textsc{vote},\cdot,\cdot,t',\cdot] \in \V: t' < t - \eta - 1\}$\label{}\\
  \> // $\V' := \V$ without all votes from slots $ < t-\eta - 1$\\
  \> \textbf{return} $\V'$ \label{}\\
  \\
  \textbf{function} $\FIL_{\text{eq}}(\V)$ \label{}\\
  \> \textbf{let} $ \V' := \V \setminus \{[\textsc{vote},\cdot,\cdot,\cdot,v_k] \in \V:\exists t_\text{eq},[\textsc{vote},B',\cdot,t_\text{eq},v_k]\in\V,[\textsc{vote},B'',\cdot,t_\text{eq},v_k]\in\V: B' \neq B''\}$\label{}\\
  \> // $\V' := \V$ without all votes by equivocators in $\V$\\
  \> \textbf{return} $\V'$ \label{}\\[-5ex]     
  \end{numbertabbing}
  }
  \end{algo}

The protocol executed by validator $v_i$ works as follows:
\begin{enumerate}
    \item \textbf{Propose, $4\Delta t$:} If $v_p$ is the proposer for slot $t$, send message [\textsc{propose}, $\chain_p$, $t$, $v_p$] through gossip, with $\chain_p \succeq \mfc(\V_p,\V_p,\genesis,t)$ { and $\chain_p.p = t$}\footnote{Note that, in practice, $\mfc(\V_i,\V_i,\genesis,t)$ is the parent of $\chain_p$. However, in our treatment we consider the more general case where $\chain_p$ can be any extension of $\mfc(\V_i,\V_i,\genesis,t)$, as long as $\chain_p.p =t$.}.
    \item \textbf{Vote, $4\Delta t + \Delta$:} Let $\chaincanmfc := \mfc(\V_i^\mathrm{frozen},\V_i,\genesis,t)$ and
    send message [\textsc{vote}, $\chain'$, $\cdot$, $t$, $v_i$] through gossip, with $\chain'$ a chain \textsc{propose}d in slot $t$ extending $ \chaincanmfc$ and $\chain'.p = t$, or  $\chain' = \chaincanmfc$ if no such proposal exists. 
    \item \textbf{Decide, $4\Delta t + 2\Delta$:} {Set $\Chain_i = \mfc(\V''_i,\V_i,\genesis,t)$.} Set $\V''_i = \V_i$, and store $\V''_i$.
    \item \textbf{$4\Delta t + 3\Delta$:} Set $\Vfrozen_i = \V_i$, and store $\Vfrozen_i$.
\end{enumerate}

D'Amato \emph{et al.}~\cite{streamliningSBFT} have demonstrated that {\TOBSVD} satisfies both Safety and Liveness as defined in Definition~\ref{def:security} assuming $|H_{\voting{(t-1)}} \setminus {A_{4\Delta t+2\Delta}}| > |A_{4\Delta t+2\Delta}|$
\footnote{Note that we consider the adversary at round $4\Delta t + 2\Delta$ due to the decision phase. This constraint differs from the one used in Constraint~\eqref{eq:pi-sleepiness}, which was intended for the modified version of the protocol that we will present shortly, which does not include a decision phase.}. While the original proofs by D'Amato \emph{et al.} assume $\eta = 1$, the underlying arguments can readily be extended to scenarios where $\eta > 1$. This extension is feasible assuming a variant of Constraint~\eqref{eq:pi-sleepiness} that considers the set $A_{\voting{t}}$ as the set of adversarial validators active at round $4\Delta t + 2\Delta$, and by considering only the latest messages cast by each validator. We do not directly prove this claim in the context of the current protocol. Instead, we demonstrate it in a modified version of the protocol, which we will present in the following.

The construction of a probabilistically safe total-order broadcast can be achieved by leveraging the protocol just presented, with a critical understanding that the validators' behavior is unaffected by decisions. Decisions primarily serve as a confirmation rule for the protocol and to ensure deterministic safety. In particular, the decide phase of slot $t$ encapsulates both the decision-making and the storage of $\V''$, which is then used in the subsequent slot for further decisions. This indicates that all activities within the decide phase are pertinent only to making decisions, and nothing else. By omitting the decide phase, the protocol can be simplified, retaining only the necessary components to maintain probabilistic safety. We have then the following modified protocol.

\begin{enumerate}
    \item \textbf{Propose, $3\Delta t$:} If $v_i$ is the proposer for slot $t$, send message [\textsc{propose}, $\chain_p$, $t$, $v_i$] through gossip, with $\chain_p \succeq \mfc(\V_i,\V_i,\genesis,t)$ { and $\chain_p.p = t$}.
    \item \textbf{Vote, $3\Delta t + \Delta$:} Let $\chaincanmfc := \mfc(\V_i^\mathrm{frozen},\V_i,\genesis,t)$ and
    send message [\textsc{vote}, $\chain'$, $\cdot,$ $t$, $v_i$] through gossip, with $\chain'$ a \textsc{propose}d chain {in slot $t$} extending {$ \chaincanmfc$ and $\chain'.p = t$}, or $\chain' = \chaincanmfc$
    if no such proposal exists. {Set $\Chain_i = \left(\chaincanmfc \right)^{\lceil\kappa}$.}
    \item \textbf{$3\Delta t + 2\Delta$:} Set $\Vfrozen_i = \V$, and store $\Vfrozen_i$. Set $t = t+1$.
\end{enumerate}

It is possible to show that when there is synchrony, Constraint~\eqref{eq:pi-sleepiness} (with slots of duration $3\Delta$) holds, and assuming an honest proposer for slot~$t$, it is guaranteed that the $B$ will extend the lock (\ie, $\chain^{3\Delta t + \Delta}_i$). As a result, all honest validators of slot $t$ will vote for the honest proposal $B$ of slot $t$. This property ensures that the validators' votes are consistently aligned with the honest proposal. This property allows us to obtain $\eta$ Reorg Resilience. In turn, $\eta$ Reorg Resilience gives us probabilistic safety for the $\kappa$-deep confirmation rule. Due to the similarity in the arguments presented in the proofs, we will demonstrate these statements on a final modification of the initially discussed protocol.

{\subsection{Probabilistically safe total-order broadcast with fast confirmations} \label{sec:prob-tob}}

This ultimate version of the protocol, presented in \Cref{algo:prob-ga-fast}, incorporates \emph{fast confirmations}. The integration of these fast confirmations will serve as a building block for our majority-based faster finality protocol. 

To this extent, \Cref{algo:prob-ga-fast} defines the function $\texttt{fastconfirmsimple}(\V,t)$ to returns the tuple $(\chain^C, Q^C)$ where $\chain^C$ is a chain that has garnered support from more than $\frac{2}{3}n$ of the validators that have \textsc{vote}d in slot~$t$ according to the view $\V$, and $Q^C${, called \emph{quorum certificate,}} is a proof of this {consisting of \textsc{vote} messages for chains extending $\chain^C$}.
If such chain does not exist, then it just returns  the tuple $(\genesis,\emptyset)$.

{
Also, \Cref{algo:prob-ga-fast} relies on the  following external function.
\begin{itemize}
  \item \textbf{$\texttt{Extend}(\chain,t)$:} If $\chain$'s length is less than $t$, then $\texttt{Extend}(\chain,t)$ produces a chain $\chain'$ of length $t$ that extends $\chain$ and that includes all of the transactions in the transaction pool at time $\proposing{t}$.
  The behavior is left unspecified for the case that $\chain$'s length is $t$ or higher.
  This will not cause us any problem because, as we will show, this case can never happen.
  Formally,
  \begin{property}\label{prop:extend}
    \begin{align*}
      \chain.p < t \implies &\land \texttt{Extend}(\chain,t) \succ \chain\\
      &\land \forall \tx \in \txpool^\proposing{t}, \tx \in \texttt{Extend}(\chain,t) 
    \end{align*}
  \end{property}
\end{itemize}
}

Moreover, as per \Cref{algo:prob-ga-fast}, each honest validator $v_i$ maintains the following state variables:

\begin{itemize}
    \item \textbf{Message Set \(\Vfrozen_i\)}: Each validator stores in \(\Vfrozen_i\) a snapshot of its view $\V_i$ at time \(3\Delta\) in each slot.
    
    \item \textbf{Chain variable \(\chainfrozen_i\)}: Each validator also stores in \(\chain^{\text{frozen}}_i\) the greatest fast confirmed chain at time $3\Delta$ in each slot $t$.
    The variable $\chainfrozen_i$ is then updated between time $0$ and $\Delta$ of slot $t+1$ incorporating the information received via the \textsc{propose} message for slot $t+1$.
    This can be seen as effectively executing the view-merge logic (largely employed in the protocol presented in \Cref{sec:rlmd-based}) on this variable.
\end{itemize}

{
Additionally, \Cref{algo:prob-ga-fast} outputs the following chain.
\begin{itemize}
  \item \textbf{Confirmed Chain $\Chain_i$:} The confirmed chain $\Chain_i$ is updated at time $\Delta$ in each slot and also potentially at time $2\Delta$.
\end{itemize}
}


{In the remainder of this work, when we want to refer to the value of a variable $\mathsf{var}$ \emph{after} the execution of the code for round $r$, we write $\mathsf{var}^r$.
If the variable already has a superscript, like $\Vfrozen[]_i$, we simply write $\Vfrozen[r]$.}

Then, \Cref{algo:prob-ga-fast} proceeds as follows.

\begin{description}
    \item[Propose:] At round $4\Delta t$, there is the propose phase. 
    Here, the proposer $v_p$ of slot $t$ broadcasts the proposal $[\textsc{propose}, \chain_p, \chain^C, Q^C, \cdot, t, v_p]$.
    {Here, $\chain_p$ is the chain returned by the function $\texttt{Extend}(\chaincanmfc,t)$ {where $\chaincanmfc$ corresponds to the output of the fork-choice $\mfc(\V_i,\V_i,\genesis,t)$.}
    Whereas,} $\chain^C$ and $Q^C$ are the fast confirmed chain and relative quorum certificate, respectively, according to the current view of $v_p$ as per output of $\texttt{fastconfirmsimple}$. Note that the proposer inputs in the fork-choice function $\mfc$ the current view $\V_p$ and fast confirmed chain $\chain^C$.
    
    \item[Vote:] \sloppy{During the interval \( [4\Delta t, 4\Delta t + \Delta] \), a validator \( v_i \), upon receiving a proposal  $[\textsc{propose}, \chain_p, \chain^C, Q^C, \cdot, t', v_p]$ from the proposer \( v_p \) for slot~\( t'=t \), first verifies whether such proposal is \emph{valid}.}
    This corresponds to checking whether \( Q^C \) is a valid certificate for \( \chain^C \) and $\chain_p.p = t$. If the \textsc{propose} message is valid, then validator~\( v_i \) updates its local variable $\chain^{\text{frozen}}_i$ according to the received proposal.

    {During the voting phase, each validator \(v_i\) computes the fork-choice function \(\mfc\) to continue building upon its output.
    The inputs to $\mfc$ are $v_i$'s view at time $3\Delta$ in the previous slot ($\V^{4\Delta(t-1)+3\Delta}_i = \Vfrozen[\voting{t}]_i$), the current view ($\V^\voting{t}_i$), and the chain $\chainfrozen_i$ updated after receiving the \textsc{propose} message for slot $t$, if any.

    Observe that by the definition of $\mfc$ (\Cref{alg:mfc}), when computing the fork-choice during the voting phase, validator $v_i$ only considers chains for which it received \textsc{vote} messages by round $4\Delta (t-1) + 3\Delta$, extending the chain from validators who have not equivocated up to the current round ($\voting{t}$), according to $v_i$'s view.
    In contrast, when computing the fork-choice during the propose phase, the proposer considers any chain for which it received \textsc{vote} messages by the proposing time ($\proposing{t}$), extending the chain from validators that have never equivocated according to the proposer’s view at that time ($\proposing{t}$).
    This approach ensures that the chain output by the fork-choice function in the voting phase is the prefix of the chain output by the fork-choice function in the propose phase of the same slot.
    This property is also known as Graded Delivery~\cite{streamliningSBFT}.
 
    If validator $v_i$ receives a valid \textsc{propose} message for a chain that extends the output of the fork-choice function, it casts a \textsc{vote} for that chain. Otherwise, it casts a \textsc{vote} for the current output of the fork-choice function. 

    In the voting phase, validator $v_i$ also sets the confirmed chain $\Chain_i$ to the highest chain among the $\kappa$-deep prefix of the chain output by $\mfc$ and the previous confirmed chain, filtering out this last chain if it is not a prefix of the chain output by $\mfc$.}

    \item[Fast Confirm:] Validator \(v_i\) checks if a chain $\chain$ can be fast confirmed, i.e., if at least $2/3n$ of the validators cast a \textsc{vote} message for a chain $\chain^\mathrm{cand}$. If that is the case, it sets $\Chain_i$ to this chain.

    \item[Merge:] At round $4\Delta t + 3\Delta$,  validator $v_i$ updates its variables $\Vfrozen_i$ and $\chainfrozen_i$ to be used in the following slot.
\end{description}

For readability purpose, we often use $\fastconfirming{t} := 4\Delta t + 2 \Delta$ and $\merging{t} := 4\Delta t + 3 \Delta$.

\begin{algo}[th!]
  \fontsize{8}{10}\selectfont
  \caption{Probabilistically safe variant of the total-order broadcast protocol of D'Amato \emph{et al.}~\cite{streamliningSBFT} with fast confirmations - code for $v_i$}
  \label{algo:prob-ga-fast}
  \begin{numbertabbing}\reset
  xxxx\=xxxx\=xxxx\=xxxx\=xxxx\=xxxx\=MMMMMMMMMMMMMMMMMMM\=\kill
    {\textbf{Output}} \label{}\\
    \>{$\Chain_i \gets \genesis$: confirmed chain of validator $v_i$}\label{}\\
    \textbf{State} \label{}\\
    \> $\V_i^\text{frozen}  \gets \{\genesis\}$: snapshot of $\V$ at time $4\Delta t + 3\Delta$ \label{}\\
    \> $\chain^\text{frozen}_i \gets \genesis $: snapshot of the fast confirmed chain at time $4\Delta t + 3\Delta$ \label{} \\
    \textbf{function} $\texttt{fastconfirmsimple}(\V,t)$\label{}\\
    \> \textbf{let} $\fastcands :=\{\chain \colon |\{v_j\colon \exists \chain' \succeq \chain : \ [\textsc{vote}, \chain', \cdot, t,\cdot] \in \V \}| \geq \frac{2}{3}n\}$\label{}\\
    \> \textbf{if} $\fastcands \neq \emptyset$ \textbf{ then}\label{}\\
    \>\>\textbf{let} $\fastcand := \max\left(\fastcands\right)$\label{}\\
    \>\> \textbf{let} $Q := \{[\textsc{vote},\chain',\cdot,t,\cdot]\in \V : \chain' \succeq \fastcand \}$\label{}\\
    \>\> \textbf{return} $(\fastcand,Q)$ \label{}\\
    \> \textbf{else}\label{}\\
    \>\> \textbf{return} $(\genesis,\emptyset)$\label{}\\    
    \textsc{Propose}\\
    \textbf{at round} $4\Delta t$ \textbf{do} \label{}\\
    \> \textbf{if} $v_i = v_p^t$ \textbf{then} \label{}\\
    \>\> \textbf{let} $ (\chain^C,Q^C) := \texttt{fastconfirmsimple}(\V_i,t-1)$\label{}\\
    \>\> \textbf{let} $\chaincanmfc := \mfc(\V_i, \V_i, \chain^C, t)$ \label{}\\
    \>\> \textbf{let} $ \chain_p := {\mathsf{Extend}(\chaincanmfc,t)}$ \label{line:algga-no-ffg-new-block}\\
    \>\> send message $[\textsc{propose}, \chain_p, \chain^C, Q^C, \cdot, t, v_i]$ through gossip \label{}\\
    \textsc{Vote}\\
    \textbf{at round} $4\Delta t + \Delta$ \textbf{do} \label{}\\
    \> \textbf{let} $\chaincanmfc := \mfc(\Vfrozen_i, \V_i, \chain^\text{frozen}_i, t)$ \label{}\\
    \> $\Chain_i \gets \max(\{\chain \in \{\Chain_i,{(\chaincanmfc)^{\lceil\kappa}}\}: \chain \preceq {\chaincanmfc}\})$\label{line:algga-no-ffg-vote-chainava}\\
    \> {\textbf{let} $ \chain := $ \textsc{propose}d chain {from slot $t$} extending ${\chaincanmfc}$ and with $\chain.p =t$, if there is one, or ${\chaincanmfc}$ otherwise}\label{line:algga-no-ffg-vote-comm}\\
    \>  send message $[\textsc{vote}, \chain, \cdot, t, v_i]$ through gossip \label{line:algga-no-ffg-vote}\\
    \textsc{Fast Confirm}\\
    \textbf{at round} $4\Delta t + 2\Delta$ \textbf{do} \label{line:algga-no-ffg-on-confirm}\\
    \> \textbf{let} $ (\fastcand,Q) := \texttt{fastconfirmsimple}(\V_i,t)$\label{}\\
    \> \textbf{if} $Q \neq \emptyset$ \textbf{then}\label{line:algga-no-ffg-if-set-chaava-to-bcand}\\
    \>\> $\Chain_i \gets \fastcand$\label{line:algga-no-ffg-set-chaava-to-bcand}\\
    \textsc{Merge}\\
    \textbf{at round} $4\Delta t + 3\Delta$ \textbf{do} \label{}\\
    \>  $\V^\text{frozen}_i \gets \V_i$\label{}\\ 
    \> $(\chainfrozen_i,\cdot) \gets \texttt{fastconfirmsimple}(\V_i,t)$\label{line:algga-no-ffg-merge-ch-frozen}\\
    \\
    \textbf{upon} receiving a gossiped message $[\textsc{propose}, {\chain_p}, \chain^C_p, Q^C_p, \cdot, t, v_p]$ \textbf{at any round in} $[4\Delta t, 4\Delta t + \Delta]$ \textbf{do}\label{}\\
      \> \textbf{if} $\mathrm{valid}([\textsc{propose}, {\chain_p}, \chain^C_p, Q^C_p, \cdot, t, v_p] ) \land \chain^\text{frozen}_i \preceq \chain^C_p$ \textbf{then} \label{line:algga-no-ffg-upon}\\  
    \>\> $\chain^\text{frozen}_i \gets \chain^C_p$ \label{algo:3sf-ga-setchainfrozen-to-chainc}\\[-5ex]
  \end{numbertabbing}
\end{algo}

\subsubsection{Analysis}\label{sec:ga-prob-analysis}

Now, we proceed with proving that \Cref{algo:prob-ga-fast} satisfied $\eta$ Reorg Resilience {(\Cref{thm:reorg-res-prop-tob})}, 
$\eta$ Dynamic Availability (\Cref{thm:dyn-avail-fast-conf-tob}), 
$\eta$ Asynchronous Reorg Resilience (\Cref{thm:async-resilience-tob}) 
and $\eta$ Asynchronous Safety {Resilience} (\Cref{{thm:async-safety-resilience-tob}}).
Remember that throughout this work we assume that less than one-third of the entire validator set {is ever controlled by the adversary (\ie, $f<\frac{n}{3}$)}.
In all the following results, we also assume $\GST=0$, , and that Constraint~\eqref{eq:pi-sleepiness} holds.

\sloppy{In the remaining part of this section, given any slot $t$, we let 
$\mfcpropose{t}_i := \mfc(\V^{\proposing{t}}_i,\V^{\proposing{t}}_i,\chain^C,t)$ with 
$(\chain^C,\cdot) = \texttt{fastconfirmsimple}(\V^\proposing{t}_i,t-1)$, and  $\mfcvote{t}_i := \mfc(\Vfrozen[\voting{t}]_i,\V^{\voting{t}}_i,\chainfrozen[\voting{t}]_i,t)$}.

{At high level, our analysis proceeds as follows.
\Cref{lem:keep-voting-tob-fast-conf,lem:one-fast-confirm-all-vote-fast-conf,lem:vote-proposal-fast-conf} set the base argument.
Specifically, in \Cref{lem:keep-voting-tob-fast-conf}, we show that if all active validators in a given slot $t$ \textsc{vote} for a chain extending $\chain$, then, for any validator active at the following slot $t+1$, the output of the fork-choice function $\mfc$ will be an extension of $\chain$ which, importantly, also implies that any validator active in slot $t+1$ \textsc{vote} for an extension of $\chain$. 
Then, \Cref{lem:one-fast-confirm-all-vote-fast-conf} leverages  \Cref{lem:keep-voting-tob-fast-conf} to show that if any active validator fast confirms a chain $\chain$ in a given slot $t$, then for any validator active at any subsequent slot, the output of the fork-choice function $\mfc$ will be an extension of $\chain$.
Finally, in \Cref{lem:vote-proposal-fast-conf} we show similar results for any chain proposed by an honest validator.
Thereafter, \Cref{lem:ga-confirmed-always-canonical} makes use of these initial results to prove that the chain confirmed by any validator honest in a given round $r_i$, regardless of whether they are active in that round, is a prefix of the chain output by the fork-choice $\mfc$ by any validator active at subsequent rounds.
This allows us to then prove $\eta$ Reorg Resilience in \Cref{thm:reorg-res-prop-tob}.
Moving on, \Cref{lem:ga-mfc-proposer-shorter-than-t} shows that the premise of \Cref{prop:extend} is always satisfied which implies that honest proposers are always able to include all the transactions in the transaction pool in the chain that they propose.
This, \Cref{lem:vote-proposal-fast-conf} and \Cref{lem:ga-confirmed-always-canonical} are then put to use in \Cref{thm:dyn-avail-fast-conf-tob} to show that \Cref{algo:prob-ga-fast} is an $\eta$-dynamically-available protocol.
Thereafter, \Cref{thm:fast-liveness-tob} proves that chains can be fast confirmed which is not listed among the security properties but very relevant to \Cref{algo:prob-ga-fast}.
Finally, the remaining Lemmas and Theorems deal with asynchrony resilience and we leave their high-level overview up to that point.
}

\begin{lemma}
    \label{lem:keep-voting-tob-fast-conf}
    If, in slot $t$, all validators in $H_{\voting{t}}$ cast \textsc{vote} messages for chains extending chain \(\chain\),
    then, for any validator $v_i\in H_{\voting{(t+1)}}$, $\mfcpropose{(t+1)}_i \succeq \chain$ and $\mfcvote{(t+1)}_i \succeq \chain$, which implies that, in slot $t+1$, all validators in $H_{\voting{(t+1)}}$ cast \textsc{vote} messages for chains extending \(\chain\).
\end{lemma}

\begin{proof}
Let $v_i$ be any honest validator in \(H_{\voting{(t+1)}}\).
Due to the joining protocol, this implies that $v_i$ is awake at round $\merging{t} = 4\Delta (t+1)-\Delta$.
Clearly the same holds at time $\proposing{(t+1)}$: $\chain^{\text{frozen},\proposing{(t+1)}}_i = \genesis \lor \chain^{\text{frozen},\proposing{(t+1)}}_i\succeq \chain$.
Now we consider each case separately and for each of them prove that $\mfcpropose{(t+1)}_i \succeq \chain$.
\begin{description}
  \item[Case 1: {$\chainfrozen[\proposing{t+1}]_i = \genesis$}.]
  Note that due to the Lemma's and synchrony assumptions, all the \textsc{vote} messages sent by validators in \(H_{\voting{t}} \setminus A_{\voting{(t+1)}}\) are included in $\V_i^{\proposing{(t+1)}}$.
   Hence, at time $\proposing{(t+1)}$, \(\left|\V_i^{\chain,t+1} \right| \geq \left|H_{\voting{t}}    \setminus A_{\voting{(t+1)}}\right|\).
  Similarly, $H_{\voting{t}} \setminus A_{\voting{(t+1)}} \subseteq \mathsf{S}(\V_i,t+1)$.
  Also, $\mathsf{S}(\V_i,t+1)\setminus \left(H_{\voting{t}} \setminus A_{\voting{(t+1)}}\right) \subseteq A_{\voting{(t+1)}} \cup \left(H_{\voting{(t-\eta+1)},\voting{(t-1)}}\setminus H_{\voting{t}}\right)$ as any \textsc{vote} in $\mathsf{S}(\V_i,t+1)$ that is not from a validator in $H_{\voting{t}} \setminus A_{\voting{(t+1)}}$ must be either from a validator Byzantine at time $\voting{(t+1)}$ or from a validator active at some point between $\voting{(t-\eta+1)}$ and $\voting{(t-1)}$ but not active at time $\voting{t}$.
  Then, we have 
  $ \left|\mathsf{S}(\V_i,t+1)\right| =
  \left|H_{\voting{t}} \setminus A_{\voting{(t+1)}} \right| + \left|\left(\mathsf{S}(\V_i,t+1) \setminus \left(H_{\voting{t}} \setminus A_{\voting{(t+1)}}\right)\right)\right|
  \leq
  \left|H_{\voting{t}} \setminus A_{\voting{(t+1)}} \right| 
  +
  \left|A_{\voting{(t+1)}} \cup \left(H_{\voting{(t-\eta+1)},\voting{(t-1)}}\setminus H_{\voting{t}}\right)\right|
  <
  2\left|H_{\voting{t}} \setminus A_{\voting{(t+1)}} \right| 
  $ where the last inequality comes from \Cref{eq:pi-sleepiness}.

  {Hence, at time $\proposing{(t+1)}$, \(\left|\V_i^{\chain,t+1} \right|\)\(\geq \left|H_{\voting{t}} \setminus A_{\voting{(t+1)}}\right|  > \frac{\left|\mathsf{S}(\V_i,t+1)\right|}{2}\).}
  
  Therefore, $\mfcpropose{(t+1)}_i \succeq \chain$.  
  \item[Case 2: {$\chainfrozen[\proposing{t+1}]_i \succeq \chain$}.] 
  \sloppy{{Given that by definition, $\mfcpropose{(t+1)}_i \succeq \chainfrozen[\proposing{t+1}]_i $, we have that $\mfcpropose{(t+1)}_i \succeq \chain$.}}
\end{description} 

Now let us move to proving that $\mfcvote{(t+1)}_i  \succeq \chain$.
Note that if a $[\textsc{propose}, B, \chain_p^C, Q^C_p, t+1, v_p]$ message is valid, then $\chain_p^C = \genesis \lor \chain_p^C \succeq \chain$.
\Cref{algo:3sf-ga-setchainfrozen-to-chainc} implies that $\chainfrozen[\voting{(t+1)}]_i = \chainfrozen[\merging{t}]_i \lor \chainfrozen[\voting{(t+1)}]_i = \chain_p^C$.
Then, from $\chain^{\text{frozen},\merging{t}}_i = \genesis \lor \chain^{\text{frozen},\merging{t}}_i\succeq \chain$ and \Cref{algo:3sf-ga-setchainfrozen-to-chainc}, it follows that, 
at time $\voting{(t+1)}$, $\chain^{\text{frozen},\voting{(t+1)}}_i = \genesis \lor \chain^{\text{frozen},\voting{(t+1)}}_i \succeq \chain$.
Now we consider each case separately. 

\begin{description}
  \item[Case 1: {$\chainfrozen[\voting{t+1}]_i = \genesis$}.]
  Note that due to the Lemma's and synchrony assumptions, all the \textsc{vote} messages sent by validators in \(H_{\voting{t}} \setminus A_{\voting{(t+1)}}\) are included in $\V^{\text{frozen},\voting{(t+1)}}_i\cap \V_i^{\voting{(t+1)}}$.
  From here the proof is effectively the same as the one for Case 1 for $\mfcpropose{(t+1)}_i$, just with $\V_i^{\proposing{(t+1)}}$ replaced by $\V^{\text{frozen},\voting{(t+1)}}_i\cap \V_i^{\voting{(t+1)}}$.
  \item[Case 2: {$\chainfrozen[\voting{t+1}]_i \succeq \chain$}.] 
  Similarly to proof of Case 2 for $\mfcpropose{(t+1)}_i$, given that by definition, $\mfcvote{(t+1)}_i \succeq \chainfrozen[\voting{t+1}]_i $, we have that $\mfcvote{(t+1)}_i \succeq \chain$.l
\end{description} 

From $\mfcvote{(t+1)}_i  \succeq \chain$ follows that, in slot $t+1$,  $v_i$ casts a \textsc{vote} message for a chain extending~$\chain$.
\end{proof}

\begin{lemma}
    \label{lem:one-fast-confirm-all-vote-fast-conf}
    \sloppy{If an honest validator fast confirms a chain $\chain$ in slot $t$ ($\ie$, there exists $Q$ such that $(\chain, Q) = \texttt{fastconfirmsimple}(\V^{\fastconfirming{t}}_i,t) \land Q \neq \emptyset$), then, for any slot $t'>t$ and validator $v_i \in H_{\voting{(t')}}$, $\mfcpropose{t'}_i \succeq \chain$ and $\mfcvote{t'}_i \succeq \chain$, which implies that, all validators in $H_{\voting{(t')}}$ cast a \textsc{vote} message for a chain extending $\chain$.}
\end{lemma}

\begin{proof}
The proof is by induction on $t'$.

\begin{description}
  \item[Base Case: $t'=t+1$.] 
  Assume that at round $\fastconfirming{t}$, an honest validator in \(H_{\voting{t}}\) fast confirms a chain $\chain$.
  Given that we assume $f < \frac{n}{3}$, conflicting quorum certificates cannot form in the same slot.
  Then, due to the joining protocol and the synchrony assumption, at time $\merging{t} = 4\Delta (t+1) - \Delta$, $\chain^\text{frozen, \merging{t}}_i = \chain$ for any validator $v_i \in H_{\voting{(t+1)}}$.

  Similarly, if a $[\textsc{propose}, \chain_p, \chain_p^C, Q^C_p, t+1, v_p]$ message is valid, then $\chain_p^C = \genesis \lor \chain_p^C = \chain$.
  This and \Crefrange{line:algga-no-ffg-upon}{algo:3sf-ga-setchainfrozen-to-chainc} imply that, at time $\voting{(t+1)}$, $\chain^{\text{frozen},\voting{(t+1)}}_i =  \chain$.

  Hence, because $\mfcvote{(t+1)}_i \succeq \chainfrozen[\voting{t+1}]_i $, $\mfcvote{(t+1)}_i \succeq \chain$, implying that all validators in $H_{\voting{(t+1)}}$ cast \textsc{vote} messages extending chain $\chain$.
  \item[Inductive Step: $t' > t+1.$] 
  Here we can just apply \Cref{lem:keep-voting-tob-fast-conf} to conclude the proof.\qedhere
\end{description}
\end{proof}

\begin{lemma}
\label{lem:vote-proposal-fast-conf}
Let $t$ be a slot with an honest proposer $v_p$ and assume that~$v_p$ casts a $[\textsc{propose}, \chain_p, \chain^C_p, Q^C_p, t, v_p]$ message. 
Then, for any slot $t'\geq t$, all validators in $H_{\voting{(t')}}$ cast a \textsc{vote} message for a chain extending $\chain_p$.
Additionally, for any slot $t'' > t$ and any validator $v_i\in H_{\voting{(t'')}}$, $\mfcpropose{t''}_i\succeq \chain_p$ and $\mfcvote{t''}_i\succeq \chain_p$.
\end{lemma}

\begin{proof}
The proof is by induction on $t'$.

\begin{description}
  \item[Base Case: $t' = t$.] 
  Suppose that in slot \(t\), an honest proposer $v_p$ sends a $[\textsc{propose}, \chain_p, \chain^C_p, Q^C_p t, v_p]$ message.
  Consider an honest validator \(v_i \in H_{\voting{t}}\).
  Note that due to the synchrony assumption, $\V^{\merging{(t-1)}}_i \subseteq \V^{\proposing{t}}_p$.
  Given that we assume $f < \frac{n}{3}$, this further implies that $\chain^{\text{frozen},\merging{(t-1)}}_i \neq \genesis \implies \chain^C_p  = \chain^{\text{frozen},\merging{(t-1)}}_i $.
  Hence, clearly $\chain^C_p \succeq \chain^{\text{frozen},\merging{(t-1)}}_i$.
  Therefore, due to \Crefrange{line:algga-no-ffg-upon}{algo:3sf-ga-setchainfrozen-to-chainc}, $\chain^{\text{frozen},\voting{t}}_i = \chain^C_p$.
  
  We know that either \(\left|\left(\Vfrozen_i\right)^{\mfcvote{t}_i,t} \cap \V_i^{\mfcvote{t}_i,t} \right| > \frac{\left|\mathsf{S}(\V_i,t)\right|}{2}\) or $\mfcvote{t}_i = \chain^\text{frozen}_i$.
  Let us consider each case separately.

  \begin{description}
    \item[Case 1: \(\left|\left(\Vfrozen\right)^{\mfcvote{t}_i,t} \cap \V_i^{\mfcvote{t}_i,t} \right| > \frac{\left|\mathsf{S}(\V_i,t)\right|}{2}\).]  
    By the Graded Delivery property~\cite{streamliningSBFT}, this implies that at time $\proposing{t}$, $\left|\V_p^{\mfcvote{t}_i,t}\right|>\frac{\left|\mathsf{S}(\V_p,t)\right|}{2}$ meaning that, due to \Cref{line:algga-no-ffg-new-block}, $\chain_p\succeq \mfcvote{t}_i$ and hence, due to \Cref{line:algga-no-ffg-vote,line:algga-no-ffg-vote-comm}, in slot $t$, $v_i$ casts a \textsc{vote} message for $\chain_p$.
  
    \item[Case 2: $\mfcvote{t}_i = \chain^\text{frozen}_i$.]
    Due to \Cref{line:algga-no-ffg-vote,line:algga-no-ffg-vote-comm}, $v_i$ still casts a \textsc{vote} for $\chain_p$ as $\chain_p \succeq \chain^C_p = \chain^\text{frozen}_i = \mfcvote{t}_i$.
  \end{description}
  \item[Inductive Step: $t' > t$.]
  Here we can just apply \Cref{lem:keep-voting-tob-fast-conf} to conclude the proof.\qedhere
\end{description}
\end{proof}

\begin{lemma}
  \label{lem:ga-confirmed-always-canonical}
  Let $r_i$ be any round and $r_j$ be any round such that $r_j\ge r_i$ and $r_j \in \{\proposing{\slot(r_j)}, \voting{\slot(r_j)}\}$. Then, for any validator~$v_i$ honest in round $r_i$ and any validator $v_j \in H_{r_j}$, $\Chain^{r_i}_i \preceq \mfc^{r_j}_j$.
\end{lemma}

\begin{proof}

We proceed by contradiction.
Let $r_i$ be the smallest round such that there exist two honest validators $v_i$ and $v_j$, and round $r_j$ such that $r_j\ge r_i$ and $r_j \in \{\proposing{\slot(r_j)}, \voting{\slot(r_j)}\}$ and 
$\Chain^{r_i}_i \npreceq \mfc^{r_j}_j$, 
that is, $r_i$ is the first round where the chain confirmed by an honest validator conflicts with the output of $\mfc$ of (another) honest validator at a propose or vote round $r_j \geq r_i$.
Given the minimality of $r_i$, $\Chain^{r_i-1}_i \neq \Chain^{r_i}_i$ which then implies that $v_i \in H_{r_i}$.
This can only happen if $r_i$ is either a voting or a fast confirmation round.
Let $t_i= \mathrm{slot}(r_i)$ and proceed by analyzing each case separately. 
\begin{description}
\item[Case 1: $r_i$ is a vote round.] 
Due to \Cref{line:algga-no-ffg-vote-chainava}, $\Chain^{r_i}_i \succeq \left(\mfcvote{t_i}_i\right)^{\lceil \kappa}$.
Let us now consider two sub cases.
  \begin{description}
    \item[Case 1.1: $\Chain^{r_i}_i = \left(\mfcvote{t_i}_i\right)^{\lceil \kappa}$.] 
    With overwhelming probability (Lemma 2~\cite{rlmd}), there exists at least one slot \(t_p\) in the interval \([t_i - \kappa, t_i)\) with an honest proposer $v_p$.
    Let $\chain_p$ be the chain \textsc{propose}d by $v_p$ in slot $t_p$.
    Given that $t_p < t_i$, \Cref{lem:vote-proposal-fast-conf} implies that 
    $\mfcvote{t_i}_i \succeq \chain_p$.
    Then, because $t_p\geq t_i-\kappa$, we have that $\chain_p \succeq  \left(\mfcvote{t_i}_i\right)^{\lceil \kappa} = \Chain^{r_i}_i$.
    Because $t_p < \slot(r_j)$, \Cref{lem:vote-proposal-fast-conf} also implies that $\chain^{r_j}_j \succeq \chain_p \succeq \Chain^{r_i}_i$ leading to a contradiction.
    \item[Case 1.2: $\Chain^{r_i}_i \succ \left(\mfcvote{t_i}_i\right)^{\lceil \kappa}$.] 
    This case implies that $\Chain^{r_i}_i  = \Chain^{r_i-1}_i$.
    From the minimality of $r_i$ we reach a contradiction.
  \end{description}
  \item[Case 2: $r_i$ is a fast confirmation round.]
  Note that this implies that $t_i < \slot(r_j)$.
  Because of the minimality of $r_i$, we only need to consider the case that $(\Chain^{r_i}_i, Q) = \texttt{fastconfirmsimple}(\V^{r_i}_i,t_i) \land Q \neq \emptyset$.
  Therefore, we can apply \Cref{lem:one-fast-confirm-all-vote-fast-conf} to conclude that $\chain^{r_j}_j \succeq \Chain^{r_i}_i$ reaching a contradiction.\qedhere
\end{description}
\end{proof}

\begin{theorem}[Reorg Resilience]\label{thm:reorg-res-prop-tob}
  \Cref{algo:prob-ga-fast} is $\eta$-reorg-resilient. 
\end{theorem}

\begin{proof}
  Take a slot $t_p$ with an honest proposer $v_p$ who sends a \textsc{propose} message for chain $\chain_p$.
  Take also any round $r_i$ and validator $v_i$ honest in round $r_i$.
  Now let $t_j$ be any slot such that $t_j > \max(t_p,\slot(r_i))$.
  By Constraint~\eqref{eq:async-condition} we know that $H_{\voting{(t_j)}}$ is not empty.
  So, pick any validator $v_j \in H_{\voting{(t_j)}}$.
  \Cref{lem:vote-proposal-fast-conf} implies that $\mfcvote{t_j}_j \succeq \chain_p$.
  \Cref{lem:ga-confirmed-always-canonical} implies that $\mfcvote{t_j}_j \succeq \Chain^{r_i}_i$.
  Hence, $\chain_p$ does not conflict with $\Chain^{r_i}_i$.
\end{proof}


{
\begin{lemma}\label{lem:ga-mfc-proposer-shorter-than-t}
  For any slot $t$ and  validator $v_i \in H_\proposing{t}$, $\mfcpropose{t}_i.p < t$.
\end{lemma}
\begin{proof}
  Due to the Validity property of Graded Agreement~\cite{streamliningSBFT} and the fact that $\genesis.p <0$,
  it is sufficient to prove that no validator honest in round $\proposing{t}$ has ever sent a \textsc{vote} message for a chain $\chain$ with $\chain.p \geq t$.
  The proof is by induction.
  \begin{description}
    \item[Base Case: $t=0$.] Obvious as no validator in $H_\proposing{t}$ has ever sent any message.
    \item[Inductive Step: $t >0$.] 
    Due to the inductive hypothesis and the fact that honest nodes do not send any \textsc{vote} message in round $\proposing{(t-1)}$, by the Validity property of Graded Agreement, we know that $\mfcvote{t-1}.p < t-1$.
    Then, the proof follows from \Cref{line:algga-no-ffg-vote-comm}.\qedhere
  \end{description}
\end{proof}
}

\begin{theorem}
    \label{thm:dyn-avail-fast-conf-tob}
\Cref{algo:prob-ga-fast} is $\eta$-dynamically-available.
\end{theorem}

\begin{proof}
We begin by proving the $\eta$ Liveness of the protocol with a confirmation time of {$T_{\text{conf}} = 8\kappa\Delta + \Delta$}.
We prove liveness by first considering the $\kappa$-deep rule only.
Take a round \( r \) at slot \( t = \text{slot}(r) \), another round {\( r_i \geq r + 8\kappa\Delta + \Delta \geq  4\Delta(t+2\kappa)+\Delta = \voting{(t+2\kappa)}\)}, and an honest validator \( v_i \in H_{r_i} \).
Let $t_i = \mathrm{slot}(r_i)$.
Due to the joining protocol, we know that the first active round for $v_i$, at or after $\voting{(t+2\kappa)}$, is a vote round.
There is a high probability of finding a slot \( t_p  \in [t + 1, t + \kappa]\) hosted by an honest proposer (Lemma 2~\cite{rlmd}).
Let $\chain_p$ be the chain \textsc{propose}d by the honest proposer $v_p$ in slot $t_p$.
{Due to \Cref{lem:ga-mfc-proposer-shorter-than-t} and \Cref{prop:extend}, we know that $\chain_p$ includes all of the transaction in $\txpool^\proposing{t_p}$.
Given that $\proposing{t_p} \geq r$, $\txpool^r \subseteq \txpool^\proposing{t_p}$, which implies that $\chain_p$ includes all of the transactions in $\txpool^r$.
}
Given that $\mathrm{slot}(r_i)>t_p$, as a consequence of \Cref{lem:vote-proposal-fast-conf}, $\mfcvote{t_i}_i \succeq \chain_p$.
Note that \Cref{line:algga-no-ffg-vote-chainava} implies that $\Chain^{r_i}_i \succeq \left(\mfcvote{t_i}_i\right)^{\lceil \kappa,t_i}$.
Then, because \( t_p \leq t + \kappa \leq t_i - \kappa \), \( \chain_p \preceq \left(\mfcvote{t_i}_i\right)^{\lceil \kappa,t_i}\) and hence $\chain_p \preceq \Chain_i^{r_i}${, which implies that $\Chain_i^{r_i}$ includes any transaction in $\txpool^r$}.

We now want to show that fast confirmation does not interfere.
Given that $t_i \geq t_j$, from \Cref{lem:vote-proposal-fast-conf}, we know that any \textsc{vote} cast in slot $t_i$ by honest validators are for chains extending $\chain_p$, given that we assume $f< \frac{n}{3}$, if $r_i$ is a fast confirmation round and  $v_i$  sets $\Chain^{r_i}_i$, then $\Chain^{r_i}_i\succeq\chain_p$ still holds.

We now show $\eta-$safety.
Take any two rounds $r_i$ and $r_j$ and validators $v_i$ and $v_j$ honest in round $r_i$ and $r_j$ respectively.
Now let $t_k$ be any slot such that $t_k > \max(\slot(r_j),\slot(r_i))$.
By Constraint~\eqref{eq:pi-sleepiness} we know that $H_{\voting{(t_k)}}$ is not empty.
So, pick any validator $v_k \in H_{\voting{(t_k)}}$.
\Cref{lem:ga-confirmed-always-canonical} implies that $\mfcvote{t_k}_k \succeq \Chain^{r_i}_i$ and $\mfcvote{t_k}_k \succeq \Chain^{r_j}_j$.
Hence, $\Chain^{r_i}_i$ does not conflict with $\Chain^{r_j}_j$.
\end{proof}

\begin{lemma}[Liveness of fast confirmations]
\label{thm:fast-liveness-tob}
Take a slot $t$ in which \(|H_{\voting{t}}| \geq \frac{2}{3}n\).
If in slot $t$ an honest validator sends a \textsc{propose} message for chain $\chain_p$, then, for any validator $ v_i \in H_{\fastconfirming{t}}$, $\Chain^{\fastconfirming{t}}_i \succeq \chain_p$.
\end{lemma}

\begin{proof}
By assumption we have that \(|H_{\voting{t}}| \geq \frac{2}{3}n\), implying from \Cref{lem:vote-proposal-fast-conf} that every honest validator in \(H_{\voting{t}} \) casts a \textsc{vote} message for $\chain_p$. It follows that any validator $ v_i \in H_{\fastconfirming{t}}$ receives a quorum of \textsc{vote} messages for $\chain_p$. 
Hence, due to \Cref{line:algga-no-ffg-if-set-chaava-to-bcand,line:algga-no-ffg-set-chaava-to-bcand},$\Chain^{\fastconfirming{t}}_i \succeq \chain_p$.
\end{proof}

\medskip

Finally, we move to 
demonstrating that \Cref{algo:prob-ga-fast} also provides $\eta$ Asynchrony Reorg Resilience and $\eta$ Asynchrony Safety Resilience, as defined in Section~\ref{sec:security}.
{
\Cref{lem:asyn-induction} shows that if from slot $t_a$ till a slot $t'$ falling within the short period of asynchrony, all aware validators \textsc{vote} for chains extending $\chain$, then so will do all aware validators at the next slot $t'+1$.
Next, \Cref{lem:asyn-induction3} leverages this result proving that, if in a slot before the short period of asynchrony starts, all active validators \textsc{vote} for chains extending $\chain$, then the chain confirmed by any aware validator at any subsequent round does not conflict with $\chain$.
This and \Cref{thm:reorg-res-prop-tob} allow then concluding in \Cref{thm:async-resilience-tob} that \Cref{algo:prob-ga-fast} is $\eta$-asynchrony-reorg-resilient.
Finally, \Cref{thm:async-safety-resilience-tob} makes use of \Cref{lem:ga-confirmed-always-canonical,lem:asyn-induction3}, and \Cref{thm:dyn-avail-fast-conf-tob} to prove that \Cref{algo:prob-ga-fast} is also $\eta$-asynchrony-safety-resilient.
}

\begin{lemma}
  \label{lem:asyn-induction}
    Let $t$ be any slot in $[t_a, t_a + \pi]$ with $\pi > 0$.
    Assume that in any slot $t' \in [t_a,t]$, all validators in $W_{\voting{t'}}$ cast \textsc{vote} messages for chains extending chain $\chain$.
    Then, in slot $t+1$, all validators in $W_\voting{(t+1)}$ also cast \textsc{vote} messages for chains extending chain $\chain$.
\end{lemma}

\begin{proof}
  The lemma makes use of two main ingredients: the expiration period is $\eta = \pi + 2$, and Constraint~\eqref{eq:async-condition}. Firstly, the votes which influence the voting phase in slot $t + 1$ are those from slots $[t - \eta, t]$.
  These include votes from slot $t_a$ since $t+1-\eta \le t_a + \pi  + 1 - \eta < t_a$. Since by \Cref{eq:async-condition3} 
  all validators in $H_{\voting{(t_a)}}\setminus A_{\voting{(t_a+1)}}$ were awake at $4\Delta t_{a} + 2\Delta$, they have all received each other's votes from slot~$t_a$.
  Then, from the Lemma's assumption, for any validator $v_i \in W_{\voting{(t+1)}}$ and any validator $v_a \in H_{\voting{(t_a)}} \setminus A_{\voting{(t+1)}}$, the latest message from $v_a$ in $\V^\voting{(t+1)}_i$ is for a chain extending $\chain$.
  By Constraint~\eqref{eq:async-condition}, validators in $H_{\voting{(t_a)}}\setminus A_{\voting{(t+1)}}$ constitute a majority of all unexpired votes in slot $t+1$. This implies that any such validator $v_i$ would again have $\left|\left(\Vfrozen_i\right)^{\chain,t+1}\cap \V_i^{\chain,t+1} \right|> \frac{\left|\mathsf{S}(\V_i,t+1)\right|}{2}$.
  {Note also that because we assume $f<\frac{n}{3}$, for any validator $ v_i \in H_{\fastconfirming{t}}$, if $(\chain', Q ) = \texttt{fastconfirmsimple}(\V^\fastconfirming{t}_i,t)$ and $Q \neq \emptyset$, then $\chain' \succeq \chain$.}
  Therefore, any such validator would in slot $t+1$ cast a  \textsc{vote} for an extension of~$\chain$. 
\end{proof}

\begin{lemma}\label{lem:asyn-induction2}
  Assume $\pi > 0$ and
  take a slot $t \leq t_a$ such that, in slot $t$, any validator in $H_{\voting{t}}$ casts a \textsc{vote} message for a chain extending $\chain$.
  Then, for any slot $t_i \geq t$, any validator in $W_{\voting{t_i}}$ casts a \textsc{vote} message for a chain extending~$\chain$.
\end{lemma}
\begin{proof}
  The proof is by induction on $t_i$.
  \begin{description}
    \item[Base Case: {$t_i \in [t,t_a]$}.] From \Cref{lem:keep-voting-tob-fast-conf}.
    \item[Inductive Step: $t_i > t_a$.]
    We assume that the Lemma holds for slot $t_i-1$ and prove that it holds also for slot $t_i$.
    Let us proceed by cases.
    \begin{description}
      \item[Case 1: {$t_i \in [t_a+1, t_a+\pi+1]$}.] 
      From \Cref{lem:asyn-induction}.
      \item[Case 2: $t_i = t_a+\pi+2$.]
      From the induction hypothesis, we know that in slot $t_a + \pi + 1$, any validator in $W_\voting{(t_a + \pi + 1)}$ casts a \textsc{vote} message for a chain extending $\chain$. 
      Since slot $t_a + \pi + 1$ is synchronous, such \textsc{vote} messages are all in the view $\V^{\voting{(t_a + \pi + 2)}}_i$ of any validator $v_i \in W_\voting{(t_a + \pi + 2)} = H_{\voting{(t_a + \pi + 2)}}$. 
      Also, given that $\eta = \pi+2$, such \textsc{vote} messages are not expired in slot $t_a + \pi + 2$ as $t_a + \pi + 2 -\eta  \leq t_a$.
      Given that Constraint~\eqref{eq:async-condition} holds for $t_a+\pi+2$, we can use the same reasoning applied in \Cref{lem:asyn-induction} to conclude that validator $v_i \in H_{\voting{(t_a+\pi+2)}}$ casts a \textsc{vote} message for a chain extending $\chain$.
      \item[Case 3: $t_a \geq t_a+\pi+3$.]
      Given that for any slot $t_j \geq t_a + \pi + 2$, $W_{\voting{t_j}} = H_{\voting{(t_j)}}$, here we can just apply \Cref{lem:keep-voting-tob-fast-conf}.\qedhere
      \end{description}
  \end{description}
\end{proof}

\begin{lemma}\label{lem:asyn-induction3}
  Assume $\pi > 0$ and
  take a slot $t \leq t_a$ such that, in slot $t$, any validator in $H_{\voting{t}}$ casts a \textsc{vote} message for a chain extending $\chain$.
  Then, for any round $r_i \geq \voting{t}$ and validator $v_i \in W_{r_i}$, $\Chain^{r_i}_i$ does not conflict with $\chain$.
\end{lemma}
\begin{proof}
  Take any round $r_i \geq \voting{t}$ and validator $v_i \in W_{r_i}$.
  Given that the confirmed chain is updated only either in a voting or fast confirmation round, let us consider the two following cases.
  \begin{description}
    \item[Case 1: $r_i = \voting{\slot(r_i)}$.] From \Cref{lem:asyn-induction2} we know that any validator $v_i \in W_\voting{\slot(r_i)}$ casts a \textsc{vote} messages for some chain $\chain' \succeq \chain$.
    Due to \Crefrange{line:algga-no-ffg-vote-chainava}{line:algga-no-ffg-vote-comm}, $\chain' \succeq \mfcvote{\slot(r_i)} \succeq \Chain^{r_i}_i$.
    Hence, $\Chain^{r_i}_i$ does not conflict with $\chain$.
    \item[Case 2: $r_i = \fastconfirming{\slot(r_i)}$.] Let us consider two sub cases.
    \begin{description}
      \item[Case 2.1: $\Chain^{r_i}_i = \Chain^{r_i-1}_i$.] This implies that $\Chain^{r_i}_i = \Chain^{\voting{\slot(r_i)}}_i$.
      Due to the joining protocol, $v_i \in W_{\fastconfirming{\slot(r_i)}}$ implies $v_i \in W_{\voting{\slot(r_i)}}$.
      Therefore, this case is proved by Case 1 above.
      \item[Case 2.2: $\Chain^{r_i}_i \neq \Chain^{r_i-1}_i$.]
      This implies $(\Chain^{r_i}_i, Q) = \texttt{fastconfirmsimple}(\V^{r_i}_i,\slot(r_i)) \land Q \neq \emptyset$.
      Given that, as established above, all validators in $W_\voting{\slot(r_i)}$ cast \textsc{vote} messages for chains extending $\chain$, Constraint~\eqref{eq:async-condition} implies that the number of validators that might cast a \textsc{vote} message for a chain non extending $\chain$ is less than half the total number of validators.
      This implies that any chain non extending $\chain$ cannot receive a quorum of \textsc{vote} message for it.
      Therefore, $\Chain^{r_i}_i \succeq \chain$ meaning that $\Chain^{r_i}_i$ does not conflict with $\chain$.\qedhere
    \end{description}
  \end{description}
\end{proof}
  
\begin{theorem}[Asynchrony Reorg Resilience]
  \label{thm:async-resilience-tob}
  \Cref{algo:prob-ga-fast} is $\eta$-asynchrony-reorg-resilient.
\end{theorem}
\begin{proof}
  Assume $\pi > 0$ and
  take a slot $t_p \leq t_a$ with an honest proposer, any round $r_i$ and any validator $v_i \in W_{r_i}$.
  If $r_i < \voting{t_p}$, we can apply \Cref{thm:reorg-res-prop-tob}.
  Otherwise, we can apply \Cref{lem:vote-proposal-fast-conf,lem:asyn-induction3}.
\end{proof}

\begin{theorem}[Asynchrony Safety Resilience]
  \label{thm:async-safety-resilience-tob}
  \Cref{algo:prob-ga-fast} $\eta$-asynchrony-safety-resilient.
\end{theorem}
\begin{proof}
  Assume $\pi > 0$ and
  pick any round $r_i\leq 4 \Delta t_a + \Delta$, any round $r_j$, any validator $v_i$ honest in $r_i$ and any validator $v_j  \in W_{r_j}$.
  Let us now proceed by cases.
  \begin{description}
    \item[Case 1: $r_j < \voting{t_a}$.]  The proof for this case follows from $\eta$ Safety (\Cref{thm:dyn-avail-fast-conf-tob}).
    \item[Case 2: $r_j \geq \voting{t_a}$.] \Cref{lem:ga-confirmed-always-canonical} implies that for any validator $v_k \in H_{\voting{(t_a)}}$, $\mfcvote{t_a}_k \succeq  \Chain^{r_i}_i$.
    This implies that, in slot $t_a$, $v_k$ casts a \textsc{vote} message for a chain extending $\Chain^{r_i}_i$.
    Hence, we can apply \Cref{lem:asyn-induction3} to conclude the proof for this case.\qedhere
  \end{description}
\end{proof}



We are now prepared to develop our faster-finality protocol, which integrates the dynamically-available protocol \Cref{algo:prob-ga-fast} with the FFG components introduced in \Cref{sec:ffg} to obtain an $\eta$-secure ebb-and-flow protocol.

\subsection{Faster finality protocol execution}
\label{sec:tob-execution}

\begin{algo}[htb!]
\fontsize{8}{10}\selectfont
\caption{Faster finality protocol -- code for validator $v_i$}
\label{alg:3sf-tob-noga}
\begin{numbertabbing}\reset
xxxx\=xxxx\=xxxx\=xxxx\=xxxx\=xxxx\=MMMMMMMMMMMMMMMMMMM\=\kill
  {\textbf{Output}} \label{}\\
  \> {$\chainava_i \gets \genesis$: available chain}\label{}\\
  \> {$\chainfin_i \gets \genesis$: finalized chain}\label{line:algotb-set-chfin-init}\\
  \textbf{State} \label{}\\
  \> $\V_i^\text{frozen}  \gets \{\genesis\}$: snapshot of $\V$ at time $4\Delta t + 3\Delta$ \label{}\\
  \> $\chain^\text{frozen}_i \gets \genesis $: snapshot of the fast confirmed chain at time $4\Delta t + 3\Delta$ \label{} \\
  \> $\GJ_i^\text{frozen} \gets (\genesis, 0)$: latest frozen greatest justified checkpoint \label{}\\
  \textbf{function} $\texttt{fastconfirm}(\V,t)$\label{}\\
  \> \textbf{let} $ (\chain^C, Q) := \texttt{fastconfirmsimple}(\V, t)$\label{}\\
  \> \textbf{if} $\chain^C \succeq \GJ(\V).\chain$ \textbf{then}\label{}\\
  \>\> \textbf{return} $(\chain^C, Q)$\label{}\\
  \> \textbf{else}\label{}\\
  \>\> \textbf{return} $(\GJ(\V).\chain,\emptyset)$\label{}\\
  \textsc{Propose}\\
  \textbf{at round} $4\Delta t$ \textbf{do} \label{}\\
  \> \textbf{if} $v_i = v_p^t$ \textbf{then} \label{}\\
  \>\> \textbf{let} $ (\chain^C,Q^C) := \texttt{fastconfirm}(\V_i,t-1)$\label{}\\
  \>\> \textbf{let} $\chaincanmfc := \mfc(\V_i, \V_i, \chain^C, t)$ \label{}\\
  \>\> \textbf{let} $ \chain_p := {\mathsf{Extend}(\chaincanmfc, t)}$ \label{}\\
  \>\> send message [\textsc{propose}, $\chain_p$, $\chain^C$, $Q^C$, $\GJ(\V_i)$,  $t$, $v_i$] through gossip \label{}\\
  \textsc{Vote}\\
  \textbf{at round} $4\Delta t + \Delta$ \textbf{do} \label{}\\
  \> \textbf{let} $\chaincanmfc := \mfc(\Vfrozen_i, \V_i, \chain^\text{frozen}_i, t)$ \label{line:algtob-set-mfc}\\
  \> $\chainava_i \gets \max(\{\chain \in \{\chainava_i,{(\chaincanmfc)^{\lceil\kappa}},\GJ_i^\text{frozen}.\chain\}: \chain \preceq {\chaincanmfc}\})$\label{line:algtob-vote-chainava}\\
  \> {$\chainfin_i \gets \max(\{\chain \colon \chain \preceq \chainava_i \land \chain \preceq \GF(\V_i).\chain\})$}\label{line:algotb-set-chfin-vote} \\
  \> \textbf{let} $ \T := (\chainava_i,t)$\label{line:algtob-set-target-checkpoint} \\
  \> \textbf{let} $ \chain := $ \textsc{propose}d chain {from slot $t$} extending ${\chaincanmfc}$ and with $\chain.p =t$, if there is one, or ${\chaincanmfc}$ otherwise\label{line:algtob-vote-comm}\\  
  \>  send message [\textsc{vote}, $\chain$, $\GJ_i^\text{frozen} \to \T$, $t$, $v_i$] through gossip \label{line:algtob-vote}\\
  \textsc{Fast Confirm}\\
  \textbf{at round} $4\Delta t + 2\Delta$ \textbf{do} \label{line:algotb-at-confirm}\\
  \> \textbf{let} $ (\fastcand,\cdot) := \texttt{fastconfirm}(\V_i,t)$\label{line:algtob-set-fastcand-fconf}\\
  \> \textbf{if} $\chainava_i \nsucceq \fastcand$ \textbf{then}\label{}\\
  \>\> $\chainava_i \gets \fastcand$\label{line:algtob-set-chaava-to-bcand}\\
  \> {$\chainfin_i \gets \GF(\V_i).\chain$}\label{line:algotb-set-chfin-fast}\\  
  \textsc{Merge}\\
  \textbf{at round} $4\Delta t + 3\Delta$ \textbf{do} \label{}\\
  \>  $\V^\text{frozen}_i \gets \V_i$\label{}\\ 
  \> $(\chainfrozen_i,\cdot) \gets \texttt{fastconfirm}(\V_i,t)$\label{line:algtob-merge-ch-frozen}\\
  \> $\GJfrozen_i \gets \GJ(\V_i)$\label{line:algtob-merge-gj-frozen} \\
  \\
  \textbf{upon} receiving a gossiped message
    $[\textsc{propose}, \chain_p, \chain^C_p, Q^C_p, \GJ_p, t, v_p]$   \textbf{at any round in} $[4\Delta t, 4\Delta t + \Delta]$ \textbf{do} \label{line:algtob-upon}\\  
  \> \textbf{when round} $4\Delta t + \Delta$ \textbf{do}\label{}\\
  \>\> \textbf{if} $\text{valid}([\textsc{propose},\chain_p, \chain^C_p, Q^C_p, \GJ_p, t, v_p]) \land \mathsf{J}(\GJ_p,\V_i)  \land \GJ_p \ge \GJ_i^\text{frozen}$ \textbf{then} \label{line:algtob-prop-if}\\
  \>\>\> $\GJ_i^\text{frozen} \gets \GJ_p$ \label{line:algtob-prop-merge-gj}\\
  \>\>\> \textbf{if} $\chain^\text{frozen}_i \not \succeq \GJ_p.\chain$ \textbf{then}\label{line:algtob-prop-check-chfrozen-gj}\\
  \>\>\>\> $\chain^\text{frozen}_i \gets \GJ_p.\chain$ \label{line:algtob-prop-set-ch-frozen-to-gj}\\
  \>\>\> \textbf{if} $\chain^\text{frozen}_i \preceq \chain^C_p$ \textbf{then}
    \label{line:algtob-if-setchainfrozen-to-chainc}\\
  \>\>\>\> $\chain^\text{frozen}_i \gets \chain^C_p$ \label{algo:3sf-noga-setchainfrozen-to-chainc}\\[-5ex]
\end{numbertabbing}
\end{algo}

In this section, we present \Cref{alg:3sf-tob-noga} which integrates the $\eta$-dynamically-available and reorg-resilient protocol of \Cref{algo:prob-ga-fast} with the FFG component introduced in \Cref{sec:ffg} to obtain a $\eta$-secure ebb-and-flow protocol.
We describe \Cref{alg:3sf-tob-noga} by discussing the main differences compared to \Cref{algo:prob-ga-fast}.

\medskip

The first difference is in the logic employed by \Cref{alg:3sf-tob-noga} to determine fast confirmed chains.
This is encoded in the function $\texttt{fastconfirm}(\V, t)$ which returns the same output of $\texttt{fastconfirmsimple}(\V,t)$ from \Cref{alg:3sf-tob-noga} as long as the returned fast confirmed chain extends $\GJ(\V).\chain$.
Otherwise, it just returns $(\GJ(\V).\chain,\emptyset)$.
The reason underpinning this change is to allow $\eta$-dynamical-availability and Reorg Resilience to be ensured again after long periods of asynchrony once the network becomes synchronous (see \Cref{sec:healing} for more details).

\medskip

Second, compared to \Cref{algo:prob-ga-fast}, \Cref{alg:3sf-tob-noga} maintains the following additional state variable:
\begin{itemize}    
    \item \textbf{Checkpoint Variable \(\GJfrozen_i\)}: Each validator  stores in \(\GJfrozen_i\) the greatest justified checkpoint according to their view at time $3\Delta$ in slot $t-1$.
    Like $\chainfrozen_i$, this variable is then updated between time $0$ and $\Delta$ in slot $t$ incorporating the information received via the \textsc{propose} message for slot $t$.
    Also like $\chainfrozen_i$, this can be seen as effectively executing a view-merge on this variable.
\end{itemize}

{Next, \Cref{alg:3sf-tob-noga} outputs two chains.
\begin{itemize}
  \item \textbf{Available Chain $\chainava_i$:} This roughly corresponds to the confirmed chain $\Chain_i$ of \Cref{algo:prob-ga-fast}.
  \item \textbf{Finalized Chain $\chainfin_i$:}
  {
  At any fast confirmation round, $\chainfin$ is set to $\GF(\V_i).\chain$, \ie, the chain of the greatest justified checkpoint.
    
  Also, at any vote round, the finalized chain $\chainfin_i$ is set to 
  the longest chain such that of $\chainfin_i \preceq \GF(\V_i).\chain \land \chainfin_i \preceq \chainava_i$%
  \footnote{{This is to ensure that $\chainfin_i$ is always a prefix of $\chainava_i$, as required by $\eta$-secure ebb-and-flow protocols.
  Simply, setting $\chainfin_i$ to $\GF(\V_i)$ in vote rounds would not ensure such properties during asynchrony.}}%
  \footnote{
    {Setting $\chainfin_i$ in vote rounds is not strictly required to ensure any of the theorems or lemmas of this work, even though, if one did remove setting $\chainfin$ in vote rounds, then one would need to tweak some of the proofs.
    However, we decided to also set $\chainfin$ in vote round as, there are scenarios under asynchrony whereby setting $\chainfin$ in the vote round of slot $t$ results in a chain longer than the chain $\chainfin$ that was set in the fast confirmation round of slot $t-1$.
    Also, we believe, but have no formal proof for it yet, that updating $\chainfin$ at vote and fast confirmation rounds provides the highest frequency at which $\chainfin$ can be updated without breaking any of the theorems and lemmas of this work.
    }
  }
  }%
  .
\end{itemize}}

\medskip

Finally, compared to \Cref{algo:prob-ga-fast}, \Cref{alg:3sf-tob-noga} proceeds as follows.
\begin{description}
    \item[Propose:] In the propose phase, the proposer for slot $t$ includes in the \textsc{propose} message that it sends also the greatest justified checkpoint according to its view at the time of proposing.

    \item[Vote:] The validity condition of a $[\textsc{propose}, \chain_p, \chain^C, Q^C, \GJ_p, t, v_p]$ message received by an honest validator $v_i$ is extended to also check that $\GJ_p$ is indeed a justified checkpoint according to the $v_i$'s current view.
    To give enough time for \textsc{ffg-vote}s that justify $\GJ_p$ to be received by $v_i$, $v_i$ postpones to time $\voting{t}$ acting upon  any \textsc{propose} messages received in the interval \( [4\Delta t, 4\Delta t + \Delta] \).
    Then, only if at that time the validity condition passes and $\GJ_p$ is greater than $v_i$'s frozen latest greatest justified checkpoint $\GJfrozen_i$, then $v_i$  updates its local variables $\GJ^{\text{frozen}}$ and $\chain^{\text{frozen}}$ according with the received proposal always ensuring that $\chainfrozen_i \succeq \GJfrozen_i$.

    In \Cref{alg:3sf-tob-noga}, honest validators also include an \textsc{ffg-vote} in the \textsc{vote} messages that they send. 
    Specifically, the source checkpoint of such \textsc{ffg-vote} sent corresponds to $\GJfrozen_i$.
    The target checkpoint is determined by selecting the highest chain among 
    the previous available chain, the $\kappa$-deep prefix of the chain output by the fork-choice function $\mfc$, and $\GJfrozen_i.\chain$, filtering out any chain if it is not a prefix of the chain output by the fork-choice function $\mfc$.

    Validator $v_i$ also sets the available chain output, $\chainava_i$ to such chain.
    So, compared to the confirmed chain in \Cref{algo:prob-ga-fast}, when setting $\chainava_i$, \Cref{alg:3sf-tob-noga} weighs the previous available chain also against $\GJfrozen_i.\chain$\footnote{While, this change is not strictly required for correctness, under some conditions it ends up producing a longer available chain which is better from a responsiveness point of view.
    For example, this is the case if we have a slot $t'$ where the proposer is honest, $\frac{2}{3}$ of the validators are active and honest  up to the voting phase of slot $t'+1$, but then some of these are corrupted immediately after, and $v_i$ is asleep in both slots and it becomes active in slot $t'+2$ only.
    In such a scenario, with the modification discussed in this paragraph, in the voting phase of slot $t'+2$, $v_i$ sets its available chain to to the chain \textsc{propose}d by the proposer of slot $t'$.
    Without such modification, $v_i$'s would not update its available chain during the voting phase of slot $t'+2$.}.

    \item[Fast Confirm:] {Validator $v_i$ checks whether $\chainava_i$ is a prefix of the chain $\mathrm{fast}^{\mathrm{cand}}$ output by $\texttt{fastconfirm}(\V_i, t)$ or it conflicts with it.
    In either case, $v_i$ updates $\chainava_i$ to $\mathrm{fast}^{\mathrm{cand}}$\footnote{It is possible to design an algorithm where, if $\chainava_i$  conflicts with $\fastcand$, then it is left to the specific implementation decide whether to update $\chainava_i$  to $\fastcand$ or leave it unchanged. However, to keep the algorithm straightforward, we chose not to incorporate this additional decision layer.}}.

    \item[Merge:] At round $4\Delta t + 3\Delta$, every validator $v_i$ also updates $\GJfrozen_i$ to the greatest justified checkpoint according to its view at that time.
\end{description}

\subsection{Faster finality protocol analysis}
\label{sec:tob-analysis}
\Cref{alg:3sf-tob-noga} works in the generalized partially synchronous sleepy model, and is in particular a $\eta$-secure ebb-and-flow protocol.
In \Cref{sec:ffg-analysis} we have already shown that the the finalized chain $\chainfin$ is $\frac{n}{3}$-accountable, and thus always safe if $f < \frac{n}{3}$.
For $\GST = 0$, we show in Section~\ref{sec:analysis-tob-sync} that, if the execution is $\eta$-compliant in this stronger sense, then all the properties of \Cref{algo:prob-ga-fast}, i.e., $\eta$ Dynamic Availability and $\eta$ Reorg Resilience, keep holding. 
Finally, in \Cref{sec:analysis-tob-psync} we show that, after $\max(\GST, \GAT) + 4\Delta$, \Cref{alg:3sf-tob-noga} guarantees the required conditions listed in Property~\ref{prop:succ-for-ffg-liveness} to ensure that the finalized chain is live.

\subsubsection{Synchrony}\label{sec:analysis-tob-sync}

\def\FFGExec{\ensuremath{{e_{\textsc{ffg}}}}}
\def\NoFFGExec{\ensuremath{{e_{\textsc{no-ffg}}}}}
In this section, we prove that chain 
$\chainava$ of \Cref{alg:3sf-tob-noga} is $\eta$-dynamically-available, $\eta$-reorg-resilient, $\eta$-asynchrony-reorg-resilient and $\eta$-asynchrony-safety-resilient.
Throughout this part of the analysis, we assume $\GST=0$, that less than one-third of the entire validator set is ever controlled by the adversary (\ie, $f<\frac{n}{3}$), and that Constraint~\eqref{eq:pi-sleepiness} holds.
At a high level, our proof strategy is to show that, under these assumptions, the integration of the FFG component into \Cref{algo:prob-ga-fast} does not affect the protocol behavior in any way that could compromise any of the properties already proven for \Cref{algo:prob-ga-fast}.
To do so, we take any execution of \Cref{alg:3sf-tob-noga} and show that there exists an adversary that induces an execution of \Cref{algo:prob-ga-fast} where (i) the messages sent by any honest validator in the two executions match except only for the FFG component of messages, (ii) the \mfc{} fork-choice outputs of any honest validator in the two executions match and (iii) any available chain output by \Cref{alg:3sf-tob-noga} is also a confirmed chain of an honest validator in the execution of \Cref{algo:prob-ga-fast}.
This then allows us to show that the result of the various Lemmas and Theorems presented in \Cref{sec:ga-prob-analysis} for \Cref{algo:prob-ga-fast} also hold for \Cref{alg:3sf-tob-noga}.

We start by formalizing the concept of validators sending the same messages except only for their FFG component.
Moreover, we define the equivalence between an execution in \Cref{algo:prob-ga-fast} and one in \Cref{alg:3sf-tob-noga}.
These definitions will be leveraged upon in the subsequent Lemmas and Theorems.

\begin{definition}\label{def:dyn-equiv}
  We say that two messages are \emph{dynamically-equivalent} if and only if they differ in at most their FFG-related components, \ie, they either are \textsc{propose} messages and differ at most in the greatest justified checkpoint component or they are \textsc{vote} messages and differ at most in the \textsc{ffg-vote} component.
  
  We say that two sets of messages are dynamically-equivalent if and only if they have the same set of equivalence classes under type-equivalence, \ie,
  for any message $m$ in any of the two sets, there exists a message $m'$ in the other set such that $m$ and $m'$ are dynamically-equivalent.

  Given two executions $e$ and $e'$ and round $r$, we say that the two   executions are \emph{honest-output-dynamically-equivalent up to round $r$} if and only if for any round $r_j<r$ and validator $v_j$ honest in round $r_j$ in both executions, the set of messages sent by $v_j$ in round $r_j$ in execution $e$ is dynamically-equivalent to the set of messages sent by $v_j$ in round $r_j$ in execution $e'$.
\end{definition}

In the remainder of this section, we use the notation $\specifyExec{e}{\X}$, where $\X$ is any variable or definition, 
to explicitly indicate the value of $\X$ in execution $e$.
Due the difference in the fast confirmation function employed by the two algorithms, we do need to explicitly specify the meaning of $\mfcproposeExec{e}{t}_i$ with $t$ being any slot.
If $e$ is an execution of \Cref{algo:prob-ga-fast}, then $\mfcproposeExec{e}{t}_i$ corresponds the definition provided in \Cref{sec:ga-prob-analysis}.
If $e$ is an execution of \Cref{alg:3sf-tob-noga}, then $\mfcproposeExec{e}{t}_i := \specifyExec{e}{\mfc}(\V^{\proposing{t}}_i,\V^{\proposing{t}}_i,\chain^C,t)$ with $(\chain^C,\cdot) = \texttt{fastconfirm}(\V^\proposing{t}_i,t-1)$.

{
\begin{definition}\label{def:exec-equiv}
    Let $\FFGExec$ by an $\eta$-compliant execution of \Cref{alg:3sf-tob-noga} and $\NoFFGExec$ be an $\eta$-compliant execution of \Cref{algo:prob-ga-fast}.
    We say that $\FFGExec$ and $\NoFFGExec$ are \emph{\mfc-equivalent} if and only if the following constraints hold:
  \begin{enumerate}

    \item\label[condition]{cond:0b-2-t} $\FFGExec$ and $\NoFFGExec$ are honest-output-dynamically-equivalent up to any round
    \item for any slot $t_i$,
    \begin{enumerate}[label*=\arabic*]
      \item \label[condition]{cond:3-2} $\mfcproposeNoFFG{t_i}_i = \mfcproposeFFG{t_i}_i$
      \item\label[condition]{cond:4-2} $\mfcvoteNoFFG{t_i}_i = \mfcvoteFFG{t_i}_i$
      \item\label[condition]{cond:5-2} for any slot $t_j \leq t_i$ and validator $v_j \in H_{\voting{(t_j)}}$, there exists a slot $t_k \leq t_j$ and a validator $v_k \in H_{\voting{(t_k)}}$ such that $\Chain^{\voting{t_k}}_k \succeq \chainava^{\voting{t_j}}_j$.
      \item\label[condition]{cond:6-2} for any slot $t_j \leq t_i$ and validator $v_j \in H_{\voting{(t_j)}}$, there exists a slot $t_k \leq t_j$ and a validator $v_k \in H_{\voting{(t_k)}}$ such that $\Chain^{\voting{t_k}}_k \succeq \chainava^{\fastconfirming{t_j}}_j$.
    \end{enumerate}
    where chain $\chainava_j$ is in \FFGExec{}, and chain $\Chain_i$ is in \NoFFGExec{}.
  \end{enumerate}
\end{definition}
}

{
\begin{lemma}\label{lem:equiv-ga2}
    Let $\FFGExec$ by any $\eta$-compliant execution of \Cref{alg:3sf-tob-noga}.  
    There exists an $\eta$-compliant execution $\NoFFGExec$ of \Cref{algo:prob-ga-fast} 
    such that $\FFGExec$ and $\NoFFGExec$ are \mfc-equivalent.
\end{lemma}
}
\def\ANoFFG{\ensuremath{A^{\NoFFGExec}}}
\def\AFFG{\ensuremath{A^{\FFGExec}}}
\begin{proof}
  Let $\FFGExec$ be any $\eta$-compliant execution of \Cref{alg:3sf-tob-noga} and let $\AFFG$ be the adversary in such an execution.
  Below, we specify the decisions made by an adversary $\ANoFFG$ in an execution $\NoFFGExec$ of \Cref{algo:prob-ga-fast}.
  Later we show that this leads to $\NoFFGExec$ satisfying all the conditions in the Lemma.
  \begin{enumerate}[label=(\roman*)]
    \item For any round $r$, the set of validators corrupted by $\ANoFFG$ in round $r$ corresponds to set of validators corrupted by $\AFFG$ in the same round $r$. 
    \item For any round $r$, the set of validators put to sleep by $\ANoFFG$ in round $r$ corresponds to set of validators put to sleep by $\AFFG$ in the same round $r$. 
    \item For any slot $t$, if an honest validator is the proposer for slot $t$ in $\FFGExec$, then it is also the proposer for slot $t$ in $\NoFFGExec$.
    \item\label{cond:adv-decision-4} For any round $r$, if $\FFGExec$ and $\NoFFGExec$ are honest-output-dynamically-equivalent up to round $r$, then $\ANoFFG$ schedules the  delivery of messages to the validators honest in round $r$ such that, for any validator $v_i$ honest in round $r_i$, the set of messages received by $v_i$ in execution $\FFGExec$ is dynamically-equivalent to the set of messages received by $v_i$ in execution $\NoFFGExec$.
  \end{enumerate}  

  Now, we prove by induction on $t_i$ that the execution $\NoFFGExec$ induced by $\ANoFFG$ satisfies all \Cref{def:exec-equiv}'s conditions.
  To do so, we add the following conditions to the inductively hypothesis.
  \begin{enumerate}[start=3]
    \item[]%
    \begin{enumerate}[start=5,label*=\arabic*]
      \item \label[condition]{cond:1-2} For any $\J$ such that $\mathsf{J}(\J,\VFFG^{\proposing{t_i}}_i)$, $\mfcproposeNoFFG{t_i}_i \succeq \J.\chain$.
      \item\label[condition]{cond:2-2} For any $\J$ such that $\mathsf{J}(\J,\VFFG^{\voting{t_i}}_i)$, $\mfcvoteNoFFG{t_i}_i \succeq \J.\chain$.
    \end{enumerate}
  \end{enumerate} 
  and rephrase \Cref{cond:0b-2-t} as follows
  \begin{enumerate}[start=3]
    \item[]%
    \begin{enumerate}[start=7,label*=\arabic*]
      \item\label[condition]{cond:0bt-2}  $\FFGExec$ and $\NoFFGExec$ are honest-output-dynamically-equivalent up to round $4\Delta (t_i+1)$.
    \end{enumerate}
  \end{enumerate}

  Note that \Cref{cond:0bt-2} holding for any $t_i$ implies \Cref{cond:0b-2-t}.

  To reduce repetitions in the proof, we treat the base case of the induction within the inductive step.

  Then, take any $t_i \geq 0$ and assume that, if $t_i > 0$, then the Lemma and the additional conditions \Cref{cond:1-2,cond:2-2,cond:0bt-2} hold for slot $t_i-1$.
  We prove that they also hold for slot $t_i$.
  We start by proving \Cref{cond:1-2,cond:2-2} for slot $t_i$, 
  then we move to \Cref{cond:3-2,cond:4-2,cond:5-2,cond:6-2} in this order and conclude with proving \Cref{cond:0bt-2}.

 \begin{description}
  \item[\Cref{cond:1-2,cond:2-2}.]
  Let $\J$ be any checkpoint such that  $\mathsf{J}(\J,\VFFG^{\voting{t_i}}_i)$.
  Because the chain of the target checkpoint of an \textsc{ffg-vote} cast by an honest validator $v_\ell$ in round $r_{\ell}$ corresponds to $(\makeFFG{\chainava}^{r_{\ell}}_\ell.\chain,\cdot)$ and we assume $f<\frac{n}{3}$,
  there exists a slot $t_k \in [0,t_i)$ and a validator $v_k \in H_{\voting{(t_k)}}$ such that $\chainava^{\voting{t_k}}_k \succeq \J.\chain$.
  By the inductive hypothesis, \Cref{cond:5-2} implies that there exists a slot $t_m \in [0,t_k]$ and validator $v_m \in H_{\voting{(t_m)}}$ such that $\Chain^{\voting{t_m}}_m \succeq \chainava^{\voting{t_k}}_k$.
  Given that $t_m \leq t_k < t_i$, from \Cref{lem:ga-confirmed-always-canonical}, we know that $\mfcproposeNoFFG{t_i} \succeq \Chain^{\voting{t_m}}_m \succeq \chainava^{\voting{t_k}}_k \succeq \J.\chain$ and that $\mfcvoteNoFFG{t_i} \succeq \Chain^{\voting{t_m}}_m \succeq \chainava^{\voting{t_k}}_k \succeq \J.\chain$.
  Given that the set of justified checkpoints in $\VFFG^{\voting{t_i}}_i$ is a superset of the set of justified checkpoints in $\VFFG^{\proposing{t_i}}_i$, both \Cref{cond:1-2} and \Cref{cond:2-2} are proven.
  \item[\Cref{cond:3-2,cond:4-2}.]
  We know that $\FFGExec$ and $\NoFFGExec$ are honest-output-dynamically-equivalent up to round $\proposing{t_i}$.
  If $t_i > 0$, then, by the inductive hypothesis, \Cref{cond:0bt-2} implies this, if $t=0$, then this is vacuously true.
  Hence, due to \cref{cond:adv-decision-4} of $\ANoFFG$'s set of decisions, $\VFFG^{\proposing{t_i}}_i$ and $\VNoFFG^{\proposing{t_i}}_i$ are dynamically-equivalent.
  From this follows that $\FFGExec$ and $\NoFFGExec$ are also honest-output-dynamically-equivalent up to round $\voting{t_i}$ and, therefore, $\VFFG^{\voting{t_i}}_i$ and $\VNoFFG^{\voting{t_i}}_i$ are dynamically-equivalent as well.
  This and \Cref{cond:1-2,cond:2-2} imply \Cref{cond:3-2,cond:4-2}.
  \item[\Cref{cond:5-2}.]
  By \Cref{line:algtob-vote-chainava}, 
  $\chainava^{\voting{t_i}}_i \in \{\chainava^{\fastconfirming{(t_i-1)}}_i,\GJfrozen[\voting{t_i}]_i.\chain,(\mfcvoteFFG{t_i}_i)^{\lceil \kappa}\}$.
  Let us consider each case.
  \begin{description}
    \item[Case 1: $\chainava^{\voting{t_i}}_i = \chainava^{\fastconfirming{(t_i-1)}}_i$.] 
    This case implies $t_i > 0$.
    Hence, we can apply the inductive hypothesis.
    Specifically, \Cref{cond:6-2} for slot $t_i-1$ implies 
    \Cref{cond:5-2} for slot $t_i$.
    \item[Case 2: {$\chainava^{\voting{t_i}}_i = \GJfrozen[t_i]_i.\chain$}]
    Due to \Cref{line:algtob-merge-ch-frozen,line:algtob-prop-if,line:algtob-prop-merge-gj}, $\mathsf{J}(\GJfrozen[t_i]_i, \V^{\voting{t_i}})$. 
    Hence, by following the reasoning applied when discussing \Cref{cond:1-2,cond:2-2},
    there exists a slot $t_k \in [0,t_i)$ and validator $v_k \in H_{\voting{(t_k)}}$ such that  $\chainava^{\voting{t_k}}_k\succeq \chainava^{\voting{t_i}}_i$.
    Therefore, we can apply the inductive hypothesis.
    Specifically, \Cref{cond:5-2} for slot $t_k$ implies 
    \Cref{cond:5-2} for slot $t_i$.
    \item[Case 3: {\normalfont $\chainava^{\voting{t_i}}_i = (\mfcvoteFFG{t_i}_i)^{\lceil \kappa}$}.] 
    \Cref{line:algga-no-ffg-vote-chainava} of \Cref{algo:prob-ga-fast} implies that $\Chain^{\voting{t_i}}_i \succeq \mfcvoteNoFFG{t_i}_i$.
    Then, 
    \Cref{cond:4-2} implies that $\Chain^{\voting{t_i}}_i \succeq \chainava^{\voting{t_i}}_i$.
  \end{description}
  \item[\Cref{cond:6-2}.]
  By \Crefrange{line:algotb-at-confirm}{line:algtob-set-chaava-to-bcand}, $\chainava^{\fastconfirming{t_i}}_i \in \{\chainava^{\voting{t_i}}_i, \GJ(\VFFG^{\fastconfirming{t_i}}).\chain, \chain^C\}$ with $(\chain^C, Q) = \texttt{fastconfirm}(\VFFG^{\fastconfirming{t_i}}_i,t_i) \land Q\neq \emptyset$.
  Let us consider each case.
  \begin{description}
    \item[Case 1: $\chainava^{\fastconfirming{t_i}}_i = \chainava^{\voting{t_i}}_i$.] In this case, \Cref{cond:5-2} implies \Cref{cond:6-2}.
    \item[Case 2: \normalfont$\chainava^{\fastconfirming{t_i}}_i = \GJ(\VFFG^{\fastconfirming{t_i}}_i).\chain$.] 
    By following the reasoning applied when discussing \Cref{cond:1-2,cond:2-2}, this means that there exists a slot $t_k \in [0,t_i]$ and validator $v_k \in H_{\voting{(t_ k)}}$ such that $\chainava^{\voting{t_k}}_k \succeq \chainava^{\fastconfirming{t_i}}_i$.
    Then, \Cref{cond:5-2} for slot $t_k$ implies 
    \Cref{cond:6-2} for slot $t_i$.
    \item[Case 3: $\chainava^{\fastconfirming{t_i}}_i = \chain^C$ with $(\chain^C, Q) = \texttt{fastconfirm}(\VFFG^{\fastconfirming{t_i}}_i,t_i) \land Q\neq \emptyset$.]
    This implies that $\VFFG^{\fastconfirming{t_i}}_i$ includes a quorum of \textsc{vote} message for chain $\chain^C$.
    Given \Cref{cond:4-2},  $\VNoFFG^{\fastconfirming{t_i}}_i$ also includes a quorum of \textsc{vote} message for chain $\chain^C$.
    Hence, $\Chain^{\fastconfirming{t_i}}_i = \chainava^{\fastconfirming{t_i}}_i$. 
    \end{description}
  \item[\Cref{cond:0bt-2}.]
  As argued in the proof of \Cref{cond:3-2,cond:4-2}, we know that $\FFGExec$ and $\NoFFGExec$ are honest-output-dynamically-equivalent up to round $\proposing{t_i}$.
  This, \Cref{cond:3-2,cond:4-2} and \cref{cond:adv-decision-4} of $\ANoFFG$'s set of decisions clearly imply \Cref{cond:0bt-2} up to round $\proposing{(t_i+1)}$.\qedhere
 \end{description}
\end{proof}

In the following Lemmas and Theorems, unless specified, we refer to executions of \Cref{alg:3sf-tob-noga}.

\begin{lemma}[Analogous of \Cref{lem:keep-voting-tob-fast-conf}]
  \label{lem:keep-voting-tob-fast-conf-ffg}
  If, in slot $t$, all validators in $H_{\voting{t}}$ cast \textsc{vote} messages for chains extending chain \(\chain\),
  then, for any validator $v_i\in H_{\voting{(t+1)}}$, $\mfcpropose{(t+1)}_i \succeq \chain$ and $\mfcvote{(t+1)}_i \succeq \chain$, which implies that, in slot $t+1$, all validators in $H_{\voting{(t+1)}}$ cast \textsc{vote} messages for chains extending \(\chain\).
\end{lemma}
\begin{proof}
  Take any $\eta$-compliant execution $\FFGExec$ of \Cref{alg:3sf-tob-noga} where in slot $t$, all validators in $H_{\voting{t}}$ cast \textsc{vote} messages for chains extending chain \(\chain\).
  Let $\NoFFGExec$ be an \mfc-equivalent execution of \Cref{algo:prob-ga-fast} which \Cref{lem:equiv-ga2} proves to exist.
  It is clear from the definition of $\mfc$-equivalent honest-output-dynamically-equivalent executions that in slot $t$ of execution $\NoFFGExec$ all validators in $H_{\voting{t}}$ cast \textsc{vote} messages for chains extending chain \(\chain\) as well.
  Hence, by \Cref{lem:keep-voting-tob-fast-conf}, in execution $\NoFFGExec$, for any validator $v_i\in H_{\voting{(t+1)}}$, $\mfcpropose{(t+1)}_i \succeq \chain$ and $\mfcvote{(t+1)}_i \succeq \chain$.
  From \Cref{cond:3-2,cond:4-2} of \Cref{lem:equiv-ga2} then it follows that this holds for execution $\FFGExec$ as well, which implies that, in slot $t+1$, all validators in $H_{\voting{(t+1)}}$ cast \textsc{vote} messages for chains extending \(\chain\).
\end{proof}

\begin{lemma}[Analogous of \Cref{lem:vote-proposal-fast-conf}]
  \label{lem:vote-proposal-fast-conf-ffg}
  Let $t$ be a slot with an honest proposer $v_p$ and assume that~$v_p$ casts a $[\textsc{propose}, \chain_p, \chain^C_p, Q^C_p, t, v_p]$ message. 
  Then, for any slot $t'\geq t$, all validators in $H_{\voting{(t')}}$ cast a \textsc{vote} message for a chain extending $\chain_p$.
  Additionally, for any slot $t'' > t$ and any validator $v_i\in H_{\voting{(t'')}}$, $\mfcpropose{t''}_i\succeq \chain_p$ and $\mfcvote{t''}_i\succeq \chain_p$.
\end{lemma}
\begin{proof}
  Take any $\eta$-compliant execution $\FFGExec$ of \Cref{alg:3sf-tob-noga} where $v_p$, the proposer of slot $t$, is honest and casts a $[\textsc{propose}, \chain_p, \chain^C_p, Q^C_p t, v_p]$ message.
  Let $\NoFFGExec$ be an $\mfc$-equivalent execution of \Cref{algo:prob-ga-fast} which \Cref{lem:equiv-ga2} proves to exist.
  It is clear from the definition of $\mfc$-equivalent honest-output-dynamically-equivalent executions that $v_p$ sends the same \textsc{propose} message in slot $t$ of execution $\NoFFGExec$.
  Hence, by \Cref{lem:vote-proposal-fast-conf}, in execution $\NoFFGExec$, for any slot $t'\geq t$, all validators in $H_{\voting{(t')}}$ cast a \textsc{vote} message for a chain extending $\chain_p$, and, for any slot $t'' > t$ and any validator $v_i\in H_{\voting{(t'')}}$, $\mfcpropose{t''}_i\succeq \chain_p$ and $\mfcvote{t''}_i\succeq \chain_p$.
  From \Cref{cond:0b-2-t,cond:3-2,cond:4-2} of \Cref{lem:equiv-ga2} then it follows that this holds for execution $\FFGExec$ as well.
\end{proof}

\begin{lemma}[Analogous of \Cref{lem:ga-confirmed-always-canonical}]\label{lem:ga-confirmed-always-canonical-ffg}
  Let $r_i$ be any round and $r_j$ be any round such that $r_j\ge r_i$ and $r_j \in \{\proposing{r_j},\voting{r_j}\}$. Then, for any validator~$v_i$ honest in round $r_i$ and any validator $v_j \in H_{r_j}$, 
  $\chainava^{r_i}_i \preceq \chain^{r_j}_j$.
\end{lemma}
\begin{proof}
  Take any $\eta$-compliant execution $\FFGExec$ of \Cref{alg:3sf-tob-noga}.
  Let $\NoFFGExec$ be an equivalent execution of \Cref{algo:prob-ga-fast} which \Cref{lem:equiv-ga2} proves to exist.
  Let $r_i$ be any round and $r_j$ be any round such that $r_j\ge r_i$ and $r_j \in \{\proposing{r_j},\voting{r_j}\}$. 
  We know that $\chainava^{r_i}_i$ is first set by $v_i$ in round $r_i' \leq r_i$ such that $r_i' \in \{\voting{\slot(r_i')},\fastconfirming{\slot(r_i')}\}$.
  From \Cref{cond:5-2,cond:6-2} of \Cref{lem:equiv-ga2}, we know that there exits a slot $t_m \leq r_k$ and a validator $v_m \in H_{\voting{(t_m)}}$ such that $\Chain^{\voting{t_m}}_m \succeq \chainava^{r_i'}_i = \chainava^{r_i}_i$.
  \Cref{lem:ga-confirmed-always-canonical} implies that $\mfcNoFFG^{r_j}_j \succeq \Chain^{\voting{t_m}}_m$.
  From \Cref{lem:equiv-ga2}, we know that $\mfcFFG^{r_j}_j= \mfcNoFFG^{r_j}_j$.
  Hence,  $\mfcFFG^{r_j}_j= \mfcNoFFG^{r_j}_j \succeq \Chain^{\voting{t_m}}_m\succeq \chainava^{r_i'}_i = \chainava^{r_i}_i$.
\end{proof}

\begin{theorem}[Reorg Resilience - Analogous of \Cref{thm:reorg-res-prop-tob}]\label{thm:reorg-res-prop-tob-ffg}
  \Cref{alg:3sf-tob-noga} is $\eta$-reorg-resilient.
\end{theorem}
\begin{proof}
  The proof follows the one for \Cref{thm:reorg-res-prop-tob} with the following changes only.
  \begin{enumerate}
    \item Replace $\Chain$ with $\chainava$,
    \item Replace \Cref{lem:vote-proposal-fast-conf} with \Cref{lem:vote-proposal-fast-conf-ffg}, and
    \item Replace \Cref{lem:ga-confirmed-always-canonical} with \Cref{lem:ga-confirmed-always-canonical-ffg}.\qedhere
  \end{enumerate}
\end{proof}


\begin{theorem}[$\eta$-dynamic-availability - Analogous of \Cref{thm:dyn-avail-fast-conf-tob}]
  \label{thm:dyn-avail-fast-conf-tob-ffg}
\Cref{alg:3sf-tob-noga} is $\eta$-dynamically-available.
\end{theorem}

\begin{proof}
  {Note that \Cref{lem:ga-mfc-proposer-shorter-than-t} holds for \Cref{alg:3sf-tob-noga} as well.}
  Then, the proof follows the one for \Cref{thm:dyn-avail-fast-conf-tob} with the following changes only.
  \begin{enumerate}
    \item Replace {$\Chain$ with $\chainava$},
    \item Replace the reference to \Cref{line:algga-no-ffg-vote-chainava} of \Cref{algo:prob-ga-fast} with \Cref{line:algtob-vote-chainava} of \Cref{alg:3sf-tob-noga},
    \item Replace \Cref{lem:vote-proposal-fast-conf} with \Cref{lem:vote-proposal-fast-conf-ffg},
    \item Replace \Cref{lem:ga-confirmed-always-canonical} by \Cref{lem:ga-confirmed-always-canonical-ffg}, and
    \item In the proof of liveness, consider the following additional case.
    Assume that $r_i \geq \voting{(t+2\kappa)}$ is a fast confirmation round and that $\chainava^{r_i}_i$ is set to $\GJ(\V^{r_i}_i).\chain$ via \Cref{line:algtob-set-chaava-to-bcand}.
    Because the chain of the target checkpoint of an \textsc{ffg-vote} cast by an honest validator $v_\ell$ in round $r_{\ell}$ corresponds to $(\makeFFG{\chainava}^{r_{\ell}}_\ell.\chain,\cdot)$ and we assume $f<\frac{n}{3}$, there exists a slot $t_k \in [0,\slot(r_i)]$ and a validator $v_k \in H_{\voting{(t_k)}}$ such that $\chainava^{\voting{t_k}}_k \succeq \GJ(\V^{r_i}_i).\chain = \chainava^{r_i}_i$.
    We have already established that $\chainava^{\voting{\slot(r_i)}}_i \succeq \chain_p$.
    Then, \Cref{lem:ga-confirmed-always-canonical-ffg} implies that $\mfcvote{\slot(r_i)}_i \succeq \GJ(\V^{r_i}_i).\chain =  \chainava^{r_i}_i $ and $\mfcvote{\slot(r_i)}_i \succeq \chainava^{\voting{\slot(r_i)}}_i \succeq \chain_p$.
    This implies that $\chainava^{\voting{\slot(r_i)}}_i$ and $\GJ(\V^{r_i}_i).\chain$ do not conflict.
    Hence, given that we assume that \Cref{line:algtob-set-chaava-to-bcand} is executed, $\chainava^{\voting{\slot(r_i)}}_i \prec \GJ(\V^{r_i}_i).\chain$.
    Because, $\chainava^{\voting{\slot(r_i)}}_i \succeq \chain_p$, we can conclude that 
    $\chainava^{r_i}_i = \GJ(\V^{r_i}_i).\chain \succ \chainava^{\voting{\slot(r_i)}}_i\succeq \chain_p$.\qedhere
  \end{enumerate}
\end{proof}

{
\begin{lemma}[Analogous of \Cref{thm:fast-liveness-tob}]
\label{thm:fast-liveness-modified-tob}
Take a slot $t$ in which \(|H_{\voting{t}}| \geq \frac{2}{3}n\).
If in slot $t$ an honest validator sends a \textsc{propose} message for chain $\chain_p$, then, for any validator $ v_i \in H_{\fastconfirming{t}}$, $\chainava^{\fastconfirming{t}}_i \succeq \chain_p$.
\end{lemma}

\begin{proof}
    The proof follows the one for \Cref{thm:fast-liveness-tob} {with the only change being replacing $\Chain_i$ with $\chainava_i$}.
\end{proof}
}

\paragraph{Asynchrony Resilience.}
We now move to proving that \Cref{alg:3sf-tob-noga} also satisfies Asynchrony Resilience.
However, compared to \Cref{algo:prob-ga-fast}, due to the FFG component, to prove Asynchrony Resilience for \Cref{alg:3sf-tob-noga} we also need to require that also \Cref{eq:async-condition2} holds.
The reason for this is that we want to avoid a situation where \textsc{ffg-vote}s sent by validators in 
$H_{\voting{(t_a+1)},\voting{(t_a+\pi+1)}} \setminus H_{\voting{(t_a)}}$
during the asynchronous period, together with \textsc{ffg-vote}s for a slot in $[t_a+1,t_a+\pi]$ but that are sent during a later slot $t_k  >t_a+\pi+2$ by validators that are adversarial in $t_k$ but that are honest in $t_a+\pi+2$ and are in $H_{\voting{(t_a)}} \setminus A_{\voting{(t_a+\pi+2)}}$ could justify a checkpoint conflicting with either the chain \textsc{propose}d or the chain confirmed by an honest validator at or before slot $t_a$.

\begin{lemma}[Analogous of \Cref{lem:asyn-induction2}]\label{lem:asyn-induction2-ffg}
  Assume $\pi > 0$ and
  take a slot $t \leq t_a$ such that, in slot $t$, any validator in $H_{\voting{t}}$ casts a \textsc{vote} message for a chain extending $\chain$.
  Then, for any slot $t_i \geq t$, any validator in $W_{\voting{t_i}}$ casts a \textsc{vote} message for a chain extending $\chain$.
\end{lemma}

\begin{proof}
  We proceed by induction on $t_i$ and add the following condition to the inductive hypothesis.
  \begin{enumerate}
    \item\label[condition]{cond:ffg1} For any slot $t'$, let $\X^{t',\chain}$ be the set of validators $v_i$ in $H_{\voting{(t')}}$ such that $\chainava^\voting{t'}_i$ conflicts with $\chain$. Then, for any $t'' \leq t_i$, $\left|X^{t'',\chain} \cup A_\infty\right|<\frac{2}{3}n$.
  \end{enumerate}

  \begin{description}
    \item[Base Case: {$t_i \in [t,t_a]$}.] From \Cref{lem:keep-voting-tob-fast-conf-ffg} which, given that we assume $f<\frac{n}{3}$, also implies \Cref{cond:ffg1}.
    \item[Inductive Step: $t_i > t_a$.]
    We assume that the Lemma holds for any slot $t_i-1$ and prove that it holds also for slot $t_i$.
    Note that \Cref{cond:ffg1} holding for slot $t_i-1$ means that for any validator $v_i \in H_{\voting{(t_i)}}$ and justified checkpoint $\J$ according to $\V^{\voting{t_i}}_i$, $\J$ does not conflict with $\chain$.
    Let us proceed by cases and keep in mind that due to \Cref{line:algtob-vote-chainava}, if in slot $t_i$ a validator $v_i \in H_{\voting{(t_i)}}$ casts a \textsc{vote} message for a chain extending $\chain$, then $\chainava^{\voting{t_i}}_i$ does not conflict with $\chain$.
    \begin{description}
      \item[Case 1: {$t_i \in [t_a+1, t_a+\pi+1]$}.] 
      It should be easy to see that, assuming that \Cref{cond:ffg1} holds for slot $t_i-1$, \Cref{lem:asyn-induction} can be applied to \Cref{alg:3sf-tob-noga} as well.
      This implies that all validators in $W_\voting{t_i}$ cast \textsc{vote} messages for chains extending $\chain$.
      Given Constraint~\eqref{eq:async-condition2}, then \Cref{cond:ffg1} holds for slot $t_i$ as well. 
      \item[Case 2: $t_i = t_a+\pi+2$.]
      Given that we assume that \Cref{cond:ffg1} holds for slot $t_i-1$, it should be easy to see that we can apply here the same reasoning used for this same case in \Cref{lem:asyn-induction2}.
      Then, given that we assume $f<\frac{n}{3}$,  \Cref{cond:ffg1} holds for slot $t_i$ as well.
      \item[Case 3: $t_a \geq t_a+\pi+3$.]
      It should be easy to see that, assuming that \Cref{cond:ffg1} holds for slot $t_i-1$, \Cref{lem:keep-voting-tob-fast-conf} can be applied to \Cref{alg:3sf-tob-noga} as well.
      This implies that all validators in $H_{\voting{t}}$ casts \textsc{vote} messages for chains extending $\chain$.
      Then the proof proceeds as per the same case in \Cref{lem:asyn-induction2}.
      Also, given that we assume $f<\frac{n}{3}$,  \Cref{cond:ffg1} holds for slot $t_i$ as well.\qedhere
    \end{description}
  \end{description} 
\end{proof}

\begin{lemma}[Analogous of \Cref{lem:asyn-induction3}]\label{lem:asyn-induction3-ffg}
  {Assume $\pi > 0$ and} 
  take a slot $t \leq t_a$ such that, in slot $t$, any validator in $H_{\voting{t}}$ casts a \textsc{vote} message for a chain extending $\chain$.
  Then, for any round $r_i \geq \voting{t}$ and validator $v_i \in W_{r_i}$, $\Chain^{r_i}_i$ does not conflict with $\chain$.
\end{lemma}
\begin{proof}
  We can follow the same reasoning used in the proof of \Cref{lem:asyn-induction3} with the following changes only.
  \begin{enumerate}
    \item Replace {$\Chain$ with $\chainava$},
    \item Replace \Cref{lem:asyn-induction2} with  \Cref{lem:asyn-induction2-ffg},
    \item In the proof of Case 1, refer to lines \Cref{line:algtob-vote-chainava,line:algtob-vote,line:algtob-vote-comm} of \Cref{alg:3sf-tob-noga} rather than \Crefrange{line:algga-no-ffg-vote-chainava}{line:algga-no-ffg-vote-comm} of \Cref{algo:prob-ga-fast}.
    \item Consider the following additional sub case for Case 2.
    \begin{description}
      \item[Case 2.3 $\chainava^{r_i}_i = \GJ(\V^{\fastconfirming{\slot(r_i)}}_i).\chain$.]
      Because the chain of the target checkpoint of an \textsc{ffg-vote} cast by an honest validator $v_\ell$ in round $r_{\ell}$ corresponds to $(\makeFFG{\chainava}^{r_{\ell}}_\ell.\chain,\cdot)$ and we assume $f<\frac{n}{3}$, that there exists a slot $t_k \in [0,\slot(r_i)]$ and validator $v_k \in H_{\voting{(t_ k)}}$ such that $\chainava^{\voting{t_k}}_k \succeq \chainava^{\fastconfirming{\slot(r_i)}}_i$.
      From Case 1 (see \Cref{lem:asyn-induction3} for actual the proof of Case1) we have that $\chain \preceq \chainava^{\voting{t_k}}_k \lor \chain \succeq \chainava^{\voting{t_k}}_k$.
      In either case, $\chain$ does not conflict with $\chainava^{\fastconfirming{\slot(r_i)}}_i$.\qedhere
    \end{description}
  \end{enumerate}
\end{proof}

\begin{theorem}[Asynchrony Reorg Resilience - Analogous of \Cref{thm:async-resilience-tob}]
  \label{thm:async-resilience-tob-ffg}
  \Cref{alg:3sf-tob-noga} is $\eta$-asynchrony-reorg-resilient.
\end{theorem}
\begin{proof}
  We can follow the same reasoning used in the proof of \Cref{thm:async-resilience-tob} with the following changes only.
  \begin{enumerate}
    \item Replace \Cref{thm:reorg-res-prop-tob} with \Cref{thm:reorg-res-prop-tob-ffg},
    \item Replace \Cref{lem:vote-proposal-fast-conf} with \Cref{lem:vote-proposal-fast-conf-ffg}, and
    \item Replace \Cref{lem:asyn-induction3} with \Cref{lem:asyn-induction3-ffg}.\qedhere
  \end{enumerate}
\end{proof}

\begin{theorem}[Asynchrony Safety Resilience - Analogous of \Cref{thm:async-safety-resilience-tob}]
  \label{thm:async-safety-resilience-tob-ffg}
  \Cref{alg:3sf-tob-noga} $\eta$-asynchrony-safety-resilient.
\end{theorem}
\begin{proof}
  We can follow the same reasoning used in the proof of \Cref{thm:async-safety-resilience-tob} with the following changes only.
  \begin{enumerate}
    \item Replace {$\Chain$ with $\chainava$},
    \item Replace \Cref{thm:dyn-avail-fast-conf-tob} with \Cref{thm:dyn-avail-fast-conf-tob-ffg}, and 
    \item Replace \Cref{lem:asyn-induction3} with \Cref{lem:asyn-induction3-ffg}.\qedhere
  \end{enumerate}
\end{proof}

\subsubsection{Partial synchrony}
\label{sec:analysis-tob-psync}
In this section, we show that \Cref{alg:3sf-tob-noga} ensures \Cref{prop:never-slashed,prop:chfin,prop:succ-for-ffg-liveness} and hence ensures that the chain $\chainfin$ is always Accountably Safe and is live after time $\max(\GST, \GAT) + {\Delta}$, meaning that \Cref{alg:3sf-tob-noga} is an $\eta$-secure ebb-and-flow protocol {under the assumption that $f<\frac{n}{3}$}.
 
\begin{lemma} \label{lem:never-slashed-3sf-tob-noga}
{\Cref{alg:3sf-tob-noga} satisfies \Cref{prop:never-slashed}.}
\end{lemma}

\begin{proof}

  {Let us prove that Algorithm \ref{alg:3sf-tob-noga} ensures all the conditions listed in \Cref{prop:never-slashed}.}

  \begin{description}
    \item[\Cref{prop:never-slashed-1}.] By  \Cref{line:algtob-vote}, an honest active validator sends only one \textsc{ffg-vote} in any slot, and therefore, \Cref{prop:never-slashed-1} holds.
    \item[\Cref{prop:never-slashed-2}.] Direct consequence of \Cref{line:algtob-set-target-checkpoint}.
    \item[\Cref{prop:never-slashed-3}.] 
    Take any two slots $t$ and $t'$ with $t< t'$.
    Let $\calS_k$ be the source checkpoint of the \textsc{ffg-vote} that validator $v_i$ sent in any slot $k$.
    Now, we prove by induction that $\calS_{\voting{t}} \leq \calS_{t'}$.
    To do so, take any two slots $r$ and $r'$ such that $t \leq r < r' \leq t'$, $v_i \in H_{r}$ and $r'$ is the first slot after $r$ such that $v_i \in H_{r'}$.
    By \Cref{line:algtob-vote}, we know that $S_r = \GJfrozen[\voting{r}]_i$ and by \Cref{line:algtob-prop-if,line:algtob-merge-ch-frozen}, we have that $\mathsf{J}( \GJfrozen[\voting{r}]_i,\V^\voting{r}_i)$.
    Then, note that $r \leq r'-1$ and that because of the joining protocol, we know that $v_i$ is awake at round $\merging{(r'-1)}$.
    Hence, we have that  $\GJfrozen[\voting{r}]_i \leq \GJ(\V^\merging{(r'-1)}_i) = \GJfrozen[\merging{(r'-1)}]_i$.
    Finally, by \Cref{line:algtob-prop-if,line:algtob-prop-merge-gj}, $\GJfrozen[\merging{(r'-1)}]_i \leq \GJfrozen[\voting{r'}]_i = \calS_{r'}$ showing that $\calS_{r} \leq \calS_{r'}$.
    \qedhere
  \end{description}
\end{proof}

{The proof of \Cref{prop:chfin} is trivial so we skip it for now.
Before being able to prove that \Cref{alg:3sf-tob-noga} satisfies \Cref{prop:succ-for-ffg-liveness}, we need to show that the base case of \Cref{lem:vote-proposal-fast-conf} holds even when $\GST >0$.
}

\begin{lemma}
  \label{lem:vote-proposal-fast-conf-ffg-only-base}
  Let $t$ be a slot with an honest proposer $v_p$  such that $\proposing{t} \geq \GST + \Delta$ and assume that~$v_p$ casts a $[\textsc{propose}, \chain_p, \chain^C_p, Q^C_p, \GJ_p, t, v_p]$ message.
  Then, for any validator $v_i \in H_{\voting{t}}$, $v_i$ casts a \textsc{vote} message for chain $\chain_p$ and $\GJfrozen[\voting{t}]_i = \GJ_p$.
\end{lemma}  
\begin{proof}
  Suppose that in slot \(t\), an honest proposer sends a $[\textsc{propose}, \chain_p, \chain^C_p, Q^C_p,\GJ_p, t, v_p]$ message.
  Consider an honest validator \(v_i\) in \(H_{\voting{t}}\).
  Note that due to the synchrony assumption, $\V^{\merging{(t-1)}}_i \subseteq \V^{\proposing{t}}_p$.
  Hence, $\GJ_p \geq \GJfrozen[\merging{(t-1)}]_i$.
  Now, we want to show that  $\chain^{\text{frozen},\voting{t}}_i = \chain^C_p$.
  Let us consider two case.
  \begin{description}
    \item[Case 1: {$\chainfrozen[\merging{(t-1)}]_i\succ \GJfrozen[\merging{(t-1)}]_i.\chain$}.]
    This implies that at time $\merging{(t-1)}$, $v_i$ has received a  quorum of \textsc{vote} messages for $\chainfrozen[\merging{(t-1)}]_i$.
    Let us consider three sub cases.
    \begin{description}
      \item[Case 1.1: {$\chainfrozen[\merging{(t-1)}]_i = \GJ_p.\chain$}.] Given that $\chain^C_p \succeq \GJ_p.\chain$, \Cref{line:algtob-if-setchainfrozen-to-chainc,algo:3sf-noga-setchainfrozen-to-chainc} set  $\chain^{\text{frozen},\voting{t}}_i = \chain^C_p$.
      \item[Case 1.2: {$\chainfrozen[\merging{(t-1)}]_i \succ \GJ_p.\chain$}.] 
      This implies that $\V^{\proposing{t}}_p$ include a quorum of \textsc{vote} messages for $\chain^C_p$.
      Due to the synchrony assumption, the quorum of \textsc{vote} messages for $\chainfrozen[\merging{(t-1)}]_i$ are in the view of validator $v_p$ at time $\proposing{t}$.
      Given that $f<\frac{n}{3}$, this case implies that $\chainfrozen[\merging{(t-1)}]_i = \chain^C_p$, and therefore
      $\chainfrozen[\voting{t}]_i = \chain^C_p$.
      \item[Case 1.3: {$\chainfrozen[\merging{(t-1)}]_i \nsucceq \GJ_p.\chain$}.] In this case, \Crefrange{line:algtob-prop-check-chfrozen-gj}{algo:3sf-noga-setchainfrozen-to-chainc} set $\chain^{\text{frozen},\voting{t}}_i = \chain^C_p$.
    \end{description}
    \item[Case 2: {$\chainfrozen[\merging{(t-1)}]_i = \GJfrozen[\merging{(t-1)}]_i.\chain$}.] Given that $\chain^C_p \succeq \GJ_p$ and $\GJ_p \geq \GJfrozen[\merging{(t-1)}]_i$, \Crefrange{line:algtob-prop-if}{algo:3sf-noga-setchainfrozen-to-chainc} imply that $\chainfrozen[\voting{t}]_i = \chain^C_p$.
  \end{description}
  
  \sloppy{Now, first suppose that at round $\voting{t}$ there exists a chain $\chain \succeq \chain^{\text{frozen}}_i = \chain^C_p$ such that \(\left|\left(\Vfrozen\right)^{\chain,t} \cap \V_i^{\chain,t} \right| > \frac{\left|\mathsf{S}(\V_i,t)\right|}{2}\).}
  By the Graded Delivery property~\cite{streamliningSBFT}, this implies that at time $\proposing{t}$, $\left|\V_p^{\chain,t}\right|>\frac{\left|\mathsf{S}(\V_p,t)\right|}{2}$ meaning that $\chain_p\succeq \chain$ and hence, in slot $t$, $v_i$ casts a \textsc{vote} message for $\chain_p$.
  
  If no such a chain exists, $v_i$ still casts a \textsc{vote} for $\chain_p$ as $\chain_p \succeq \chain^C_p = \chain^\text{frozen}_i$.

  \medskip

  Also, given that $\GJ_p \geq \GJfrozen[\merging{(t-1)}]_i$, due to \Cref{line:algtob-prop-if,line:algtob-prop-merge-gj}, $\GJfrozen[\voting{t}] = \GJ_p$.
\end{proof}

\begin{lemma}
  \label{thm:liveness-tob-no-ga}
 {\Cref{alg:3sf-tob-noga} satisfies \Cref{prop:succ-for-ffg-liveness}.}
\end{lemma}
  
\begin{proof}
  {Let us prove that \Cref{alg:3sf-tob-noga} ensures each condition listed in \Cref{prop:succ-for-ffg-liveness}.}
  \begin{description}
    \item[\Cref{prop:succ-for-ffg-liveness-1-1}.] 
    Let $[\textsc{vote}, \chain, \calS \to \T,t,v_i]$ be the \textsc{vote} message cast by validator $v_i$ in slot $t$.
    We need to show that $\calS.\chain \preceq \T.\chain \preceq \chain$.

    \Cref{line:algtob-vote-chainava,line:algtob-set-target-checkpoint,line:algtob-vote,line:algtob-vote-comm} imply that $\T.\chain \preceq \mfcvote{t}_i$ and $\mfcvote{t}_i \preceq \chain$, hence $\T.\chain \preceq \chain$.

    Then, we are left with showing that $\calS.\chain \preceq \T.\chain$.
    Due to \Cref{line:algtob-set-target-checkpoint,line:algtob-vote}, this amounts to proving that $\GJfrozen[\voting{t}]_i.\chain \preceq \chainava^\voting{t}_i$.
    First, note that \Cref{alg:3sf-tob-noga} ensures that $\chain^{\text{frozen},4 \Delta t}_i \succeq \GJ^{\text{frozen},4 \Delta t}_i$.
    This can be proven by induction on slot $t'$.
    For $t'=0$, the base case, this is obvious from the initialization code.
    For $t'>0$, assume that the statement holds for $t'$.
    Given that $\texttt{fastconfirm}(\V,\cdot)$ always outputs a chain descendant of $\GJ(\V).\chain$, \Cref{line:algtob-merge-ch-frozen,line:algtob-merge-gj-frozen} ensure that $\chain^{\text{frozen},4 \Delta t' + 3 \Delta}_i \succeq \GJ^{\text{frozen},\merging{t'}}_i.\chain$.
    Then, \Cref{line:algtob-prop-check-chfrozen-gj,line:algtob-prop-merge-gj,line:algtob-prop-set-ch-frozen-to-gj} ensure that this relation is maintained for any round until $\voting{(t'+1)}$.
    
    Hence, $\mfc(\cdot, \cdot, \chain^{\text{frozen},\voting{t}}, \cdot) \succeq \GJ^{\text{frozen},4 \Delta t +\Delta}_i.\chain$.
    Therefore, \Cref{line:algtob-vote-chainava} ensures that $\chainava_i^\voting{t} \succeq \GJ(\V^{\text{frozen},4 \Delta t +\Delta}_i).\chain$.
    \item[\Cref{prop:succ-for-ffg-liveness-1-2}.] Due to \Cref{line:algtob-merge-gj-frozen,line:algtob-prop-if}, there exists a set $\V^{\mathsf{FFGvote},t}_i \supseteq \V^{\merging{(t-1)}}_i$ such that $\GJ^{\text{frozen},\voting{t}}_i = \GJ(\V^{\mathsf{FFGvote},t}_i)$.
    \item[\Cref{prop:succ-for-ffg-liveness-1-3}.] \Cref{line:algtob-vote} implies that any \textsc{ffg-vote} $\calS \to \T$ sent by an honest validator during slot $t$ is such that $\T_i.c=t$.
    \item[{\Cref{prop:succ-for-ffg-liveness-2-1}.}] {Follows from \Cref{lem:ga-mfc-proposer-shorter-than-t}.}
    \item[\Cref{prop:succ-for-ffg-liveness-2-2}.]
    {Let the greatest justified checkpoint in the view of~$v_p$ at round $\proposing{t}$ be $\GJ_p$ and $v_i$ be any always-honest validator.
    By \Cref{lem:vote-proposal-fast-conf-ffg-only-base}, $\GJfrozen[\voting{t}]_i = \GJ_p$.
    Hence, by \Cref{line:algtob-vote}, at time $\voting{t}$, validator $v_i$ sends an \textsc{ffg-vote} $\calS \to \T_i$ with the same source checkpoint, namely $\calS = \GJ_p$.
    Finally, as argued when proving \Cref{prop:succ-for-ffg-liveness-1-1}, the chain that $v_i$ \textsc{vote}s for is a descendant of $\T_i.\chain$.
    Then, from \Cref{lem:vote-proposal-fast-conf-ffg-only-base} and \Cref{line:algtob-vote-chainava}, we can conclude that $\T_i.\chain \preceq \chain_p$.}

    \item[\Cref{prop:succ-for-ffg-liveness-2-3-2}.]
    Given that $\voting{t} \geq \GST$, and that \textsc{ffg-vote}s for slot $t$ are sent at round $\voting{t}$, they are received by any always-honest validator $v_i$ at time  $\merging{t}$ and, hence, included in $\V^{\text{frozen},\merging{t}}_i$.
    From \Cref{prop:succ-for-ffg-liveness-1-2}, we then have that $\V^{\text{frozen},\merging{t}}_i \subseteq \V^{\mathsf{FFGvote},t+1}_i$.

    \item[\Cref{prop:succ-for-ffg-liveness-2-3-1}.]
    {Let $v_i$ be any always-honest validator.
    By \Cref{lem:vote-proposal-fast-conf-ffg-only-base}, at time $\voting{t}$, every always-honest validator casts a \textsc{vote} message for chain~$\chain_p$.
    Hence, given that $\fastconfirming{t} \geq \GAT$ and $f < \frac{n}{3}$, for any always-honest validator, $\texttt{fastconfirm}(\V^\fastconfirming{t}_i,t) = (\chain_p, \cdot)$.
    Given that, as established in the proof of \Cref{prop:succ-for-ffg-liveness-2-2}, $\T_i.\chain = \chainava^\voting{t}_i \preceq \chain_p$, $\chainava^\fastconfirming{t}_i$ is set to $\chain_p$ (\Cref{line:algtob-set-chaava-to-bcand}).
    By \Cref{prop:succ-for-ffg-liveness-2-2}, at time $\voting{t}$, $v_i$  sends an \textsc{ffg-vote} with the same source checkpoint and with a target checkpoint $\T_i$ such that $\T_i.\chain \preceq \chain_p$.
    Hence, given $f <\frac{n}{3}$  and \Cref{prop:succ-for-ffg-liveness-1-3}, at time $\voting{(t+1)}$, $\GJ(\V^{\text{frozen},\voting{(t+1)}}_i).c = t \land \GJ(\V^{\text{frozen},\voting{(t+1)}}_i).\chain \preceq \chain_p$.
    Given that $\kappa >1$, $t+1-\kappa < t$.
    Therefore, given $\chainava^\fastconfirming{t}_i = \chain_p$ and $\chain_p.p = t$, in round $\voting{(t+1)}$, due to \Cref{line:algtob-vote-chainava}, $\chainava_i$  is unchanged, \ie, $\chainava_i = \chain_p$.
    Then, due to \Cref{line:algtob-vote,line:algtob-set-target-checkpoint}, the \textsc{ffg-vote} sent by $v_i$ in slot $t+1$ has a target $\T$ such that $\T.\chain =\chain_p$.}

    \item[\Cref{prop:succ-for-ffg-liveness-2-4-1}.] Same reasoning as per \Cref{prop:succ-for-ffg-liveness-2-3-2}.
    \item[\Cref{prop:succ-for-ffg-liveness-2-4-2}.]  Condition (i) follows from the synchrony assumption.
    {As shown in the proof of \Cref{prop:succ-for-ffg-liveness-2-3-1}, for any always-honest validator $v_i$, $\chainava^\voting{(t+1)}_i = \chain_p$.
    By \Cref{line:algtob-set-mfc,line:algtob-vote-chainava,line:algtob-vote-comm,line:algtob-vote}, this implies that $v_i$ \textsc{vote}s for a chain extending $\chain_p$.
    Then, by following the same argument outlined in the proof of \Cref{prop:succ-for-ffg-liveness-2-3-1} we can conclude the following.
    At time $\fastconfirming{(t+1)}$, $\chainava^\fastconfirming{(t+1)}_i \succeq \chain_p$.
    At time $H_\voting{(t+2)}$, $\GJ(\V^{\text{frozen},\voting{(t+2)}}_i).c = t+1 \land \GJ(\V^{\text{frozen},\voting{(t+2)}}_i).\chain = \chain_p$.
    \Cref{line:algtob-vote-chainava} ensures then 
    that $\chainava^\voting{(t+2)}_i \succeq \chain_p$ and, hence, by \Cref{line:algtob-vote-chainava}, any always-honest validator  \textsc{vote}s for a chain extending $\chain_p$.
    Hence, $\chainava^\fastconfirming{(t+2)}_i \succeq \chain_p$.
    This and \Cref{line:algotb-set-chfin-fast} prove condition (ii).}
    \qedhere
  \end{description}
\end{proof}

\begin{theorem}\label{thm:ga-ebb-and-flow}
  \Cref{alg:3sf-tob-noga} is a $\eta$-secure ebb-and-flow protocol.
\end{theorem}
\begin{proof}
  {\Cref{line:algotb-set-chfin-init,line:algotb-set-chfin-vote,line:algotb-set-chfin-fast} and the definition of $\GF$ prove that \Cref{alg:3sf-tob-noga} satisfies \Cref{prop:chfin}.
  Safety of $\chainava$ and $\chainfin$, and  Liveness of $\chainava$ follow from this, \Cref{thm:accountable-safety,thm:dyn-avail-fast-conf-tob-ffg}, and \Cref{lem:never-slashed-3sf-tob-noga}.

  Then, we need to show that $\chainfin$ is live after $\max(\GST,\GAT) + O(\Delta)$ with confirmation time $O(\kappa)$.
  Let $t$ be any slot such that $\proposing{t} \geq \max(\GST,\GAT) + \Delta$.
  There is a high probability of finding a slot \( t_p  \in [t, t + \kappa)\) hosted by an honest proposer (Lemma 2~\cite{rlmd}).
  By \Cref{thm:liveness-ffg-general} and \Cref{thm:liveness-tob-no-ga}, $\txpool^\proposing{t}  \subseteq \txpool^\proposing{t_p} \subseteq \chainfin^\fastconfirming{(t_p+2)}_i$ for any validator honest in $\fastconfirming{(t_p+2)}$ and, clearly, $\fastconfirming{(t_p+2)}-\proposing{t} \in O(\kappa)$.} 

  {
    Finally, we need to show that, for any round $r$ and validator $v_i \in H_r$, $\chainfin^r_i \preceq \chainava^r_i$.
    \Cref{line:algotb-set-chfin-vote} clearly ensures that this is the case when $r$ is a vote round.
    Note that the only other times at which either $\chainava$ or $\chainfin$ are potentially modified are fast confirmation rounds.
    Then, let $r$ be any fast confirmation round.
    \Crefrange{line:algtob-set-fastcand-fconf}{line:algtob-set-chaava-to-bcand} ensure that $\GJ(\V^r_i).\chain \preceq \chainava^r_i$.
    By \Cref{alg:justification-finalization}, 
    $\GJ(\V^r_i)\geq\GF(\V^r_i)$ 
    which, due to \Cref{lem:accountable-jutification} and the assumption that $f < \frac{n}{3}$, implies $\chainfin  \preceq \GJ(\V^r_i).\chain \preceq \chainava^r_i$.
  }
\end{proof}

We conclude this section by showing that the finalized chain of an honest validator grows monotonically. Although this property is not explicitly defined in \Cref{sec:security}, it is frequently desired by Ethereum protocol implementers~\cite{ethereum-properties}.

\begin{lemma}\label{lem:chfin-ga-always-grows}
  For any two round $r' \geq r$ and validator $v_i$ honest in round $r'$, $\chainfin^{r'}_i \succeq \chainfin^r_i$.
\end{lemma}
\begin{proof}
  Let us prove this Lemma by contradiction.
  Let $r$ be the smallest round such that there exist round $r' \geq r$ and a validator $v_i$ honest in round $r'$ such that $\chainfin^{r'}_i \nsucceq \chainfin^r_i$.
  {Clearly $r' > r$.}
  Assume $r'$ to be the smallest such round.
  First, we want to show that there exists a checkpoint $\J$ such that $\chainava^{r'}_i \succeq \J.\chain \land \mathsf{J}(\J, \V^{r'}_i) \land \J \geq \GF(\V^r_i)$.
  Note that both $r$ are $r'$ are either a vote or a fast confirmation round.
  This is due to the minimality of $r$ and $r'$ and the fact that $\chainfin_i$ is only set in these types of rounds.
  Then, let us proceed by cases.

  \begin{description}
    \item[Case 1: $r'$ is a vote round.]
    In this case $\J = \GJfrozen[r']_i$.
    {Note that due to \Cref{line:algtob-merge-ch-frozen,line:algtob-merge-gj-frozen,line:algtob-prop-check-chfrozen-gj,line:algtob-prop-set-ch-frozen-to-gj}, $\chainfrozen[r']_i \succeq \GJfrozen[r']_i.\chain$.
    Then, \Cref{line:algtob-vote-chainava} implies that $\chainava^{r'}_i \succeq \GJfrozen[r']_i.\chain$.
    Also, \Cref{line:algtob-merge-gj-frozen,line:algtob-prop-if,line:algtob-prop-merge-gj} imply that $\mathsf{J}(\GJfrozen[r']_i,\V^{r'}_i)$.
    Finally, \Cref{line:algtob-prop-if,line:algtob-prop-merge-gj,line:algtob-merge-gj-frozen}, and $r' > r$ imply that $\GJfrozen[r']_i \geq \GJ(\V^{r}_i)$.}
    \item[Case 2: $r'$ is a fast confirmation round.]
    In this case, $\J = \GJ(\V^{r'}_i)$.
    Clearly, $\mathsf{J}(\GJ(\V^{r'}_i),\V^{r'}_i)$.
    {\Crefrange{line:algtob-set-fastcand-fconf}{line:algtob-set-chaava-to-bcand} ensure that $\chainava^{r'}_i \succeq \GJ(\V^{r'}_i).\chain$.}
    Then, from \Cref{alg:justification-finalization} and the fact that $\V^r_i\subseteq \V^{r'}_i$, we can conclude that $\GJ(\V^{r'}_i)\geq \GF(\V^{r'}_i) \geq \GF(\V^{r}_i)$.
  \end{description}  
  From \Cref{alg:justification-finalization} and $\V^{r}_i \subseteq \V^{r'}_i$, we also have that $\GF(\V^{r'}_i)\geq \GF(\V^r_i)$.
  Then, \Cref{lem:accountable-jutification} and $f<\frac{n}{3}$ imply $\chainava^{r'}_i \succeq \J.\chain \succeq \GF(\V^r_i).\chain \succeq \chainfin^r_i$ and
  $\GF(\V^{r'}_i).\chain\succeq \GF(\V^r_i).\chain$, which further imply
  $\chainfin^{r'}_i \succeq \GF(\V^r_i).\chain$.
  Hence, $\chainfin^{r'}_i  \succeq \GF(\V^r_i).\chain \succeq \chainfin^r_i$ reaching a contradiction.
\end{proof}

\section{\RLMDGHOST-Based Faster Finality Protocol}
\label{sec:rlmd-based}
In this section, we present a faster finality protocol that builds upon the dynamically-available consensus protocol of D'Amato and Zanolini~\cite{rlmd}, \RLMDGHOST. Specifically, \RLMDGHOST, which serves as the available component outputting $\chainava$, is combined with the FFG component described in \Cref{sec:ffg}, which {finalizes chains prefix of $\chainava$} and outputs $\chainfin$, achieving an $\eta$-secure ebb-and-flow protocol. We start by briefly recalling \RLMDGHOST~\cite{rlmd} and state its most relevant properties.

\subsection{Recalling \RLMDGHOST}
\label{sec:recalling-rlmd}

Before introducing \RLMDGHOST\rs{Is it confusing that we name the protocol and the fork-choice function in the same way?} with fast confirmation in \Cref{alg:rlmd}, 
we define a new fork-choice function in \Cref{alg:rlmd-fc}, 
which will be used similarly to how the fork-choice function $\mfc$ is utilized in Algorithm~\ref{algo:prob-ga-fast}.

The function $\rlmdghost(\V, {B_{\mathrm{start}}}, t)$ in \Cref{alg:rlmd-fc} starts with refining the view $\V$ by calling the filter function $\FIL_{\text{rlmd}}$ that keeps
only the latest, non-equivocating \textsc{vote} messages that fall within the expiry period $[t-\eta, t)$ for the current slot $t$. 
Formally, $\FIL_{\text{rlmd}}$ is defined as a composition of the filters already defined in \Cref{alg:mfc}: $\FIL_{\text{rlmd}}(\V,t) = \FIL_{\text{lmd}}(\FIL_{\eta\text{-exp}}(\FIL_{eq}(\V),t))$.
After this pruning, it proceeds as follows starting from {$B_{\mathrm{start}}$}.
From a given block, it selects the next block by identifying the child block that has the most weight, measured by the number of latest, unexpired, and non-equivocating \textsc{vote} messages for its descendants.
{One it reaches a block with no children, it returns the chain identified by this block.}
This is implemented as the function \GHOST in \Cref{alg:rlmd-fc} .
\footnote{{The $\rlmdghost$ fork-choice function defined in~\cite{rlmd} corresponds to the $\rlmdghost$ fork-choice function of \Cref{alg:rlmd-fc} with $B_\mathrm{start} = \genesis$.}}

\begin{algo}[h!]
\caption{$\rlmdghost$ fork-choice function}
\label{alg:rlmd-fc}
\vbox{
\small
\begin{numbertabbing}\reset
  xxxx\=xxxx\=xxxx\=xxxx\=xxxx\=xxxx\=MMMMMMMMMMMMMMMMMMM\=\kill
\textbf{function} $\rlmdghost(\V,{B_{\mathrm{start}}},t)$ \label{}\\
  \> \textbf{return} $\ghost(\FIL_{\text{rlmd}}(\V, t), {B_{\mathrm{start}}}, t)$ \label{}\\\\

    \textbf{function} $\ghost(\V,{B_{\mathrm{start}}}, t)$ \label{}\\
  \> $B \gets {B_{\mathrm{start}}}$ \label{hfc:root}\\
  \> \textbf{while} $B$ has at least one {child} block $B'$ in $\V$ with $B'.p \leq t$ \textbf{do}\label{}\\
  \>\> $B \gets \underset{B'\in \V, \text{ child of } B}{\text{arg max}} w(B',{\V})$ 
   \` // $w(B,{\V})$ outputs the number of votes in ${\V}$ for $B'\succeq B$ with $B'.p \leq t$\label{line:rlmd-argmax} \\
  \>\> // ties are broken according to a deterministic rule \label{}\\
  \> \textbf{return} $B$ \label{}\\
\\

  \textbf{function} $\FIL_{\text{rlmd}}(\V, t)$ \label{}\\
  \> \textbf{return} $\FIL_{\text{lmd}}(\FIL_{\eta\text{-exp}}(\FIL_{eq}(\V),t))$ \label{}\\[-5ex]
\end{numbertabbing}
}
\end{algo}

Similarly to \Cref{algo:prob-ga-fast} in Section~\ref{sec:ga-based}, 
\Cref{alg:rlmd}
defines the function $\texttt{fastconfirmsimple}(\V,t)$ to returns a chain $\fastcand$ that has garnered support from more than $\frac{2}{3}n$ of the validators that have \textsc{vote}d in slot $t$ according to the view $\V$. 
Specifically, each validator \(v_i\), {via the function $\texttt{fastconfirmsimple}$,} identifies the set $\fastcands$ of all the chains \(\chain\) from \textsc{vote} messages of slot \(t\) that meet a specific voting criterion, \ie, when there are at least \(\frac{2}{3}n\) validators sending \textsc{vote} messages of the form \([\textsc{vote}, \chain', \cdot, t, \cdot]\) in $\V_i$ with \(\chain \preceq \chain'\), and then returns the maximum chain $\fastcand$ in this set.
If such chain does not exist, then it just returns $\genesis$.

Also, {similarly to \Cref{algo:prob-ga-fast},} each honest validator $v_i$ maintains the following state variable:

\begin{itemize}
    \item \textbf{Message Set \(\Vfrozen_i\)}: Each validator stores in \(\Vfrozen_i\) a snapshot of its view $\V_i$ at time \(3\Delta\) in each slot, which is then updated between time 0 and time $\Delta$ of the next slot, if a \textsc{propose} message is received, by merging it with the view included in it.
\end{itemize}
{However, compared to \Cref{algo:prob-ga-fast}, \Cref{alg:rlmd} does not require to the state variable $\chainfrozen_i$.}

{Like \Cref{algo:prob-ga-fast}, each honest validator $v_i$ outputs the following chain.}
\begin{itemize}
  \item {\textbf{Confirmed Chain $\Chain_i$:} The confirmed chain $\Chain_i$ is updated at time $\Delta$ in each slot and also potentially at time $2\Delta$.}
\end{itemize}

{Overall, \Cref{alg:rlmd} proceeds following the structure of \Cref{algo:prob-ga-fast}.}

\begin{description}
    \item[Propose:] At round $4\Delta t$, the propose phase occurs. 
    {During this phase, the proposer $v_p$ for slot $t$ broadcasts the message $[\textsc{propose}, \chain_p, \V_p \cup \{\chain_p\}, t, v_p]$. 
    Here, $\chain_p$ is the chain returned by the function $\texttt{Extend}(\chaincanrlmd,t)$ where $\texttt{Extend}$ is the same function utilized by \Cref{algo:prob-ga-fast} and therefore ensures \Cref{prop:extend}, and $\chaincanrlmd$ corresponds to the output of the $\rlmdghost$ fork-choice $\rlmdghost(\V_i, \genesis,t)$ after removing any block with slot higher or equal to $t$, \ie, we take its $1$-deep prefix\footnote{The reason for taking the $1$-deep prefix of the output of the \rlmdghost{} fork-choice is that, as per our protocol's definition, it is possible for a Byzantine validator to get the \rlmdghost{} fork-choice at propose time $\proposing{t}$ to output a chain with height $t$ by sending a block with slot $t$ and relatives \textsc{vote}s for it. Alternatively, one could avoid taking the $1$-deep prefix by requiring that blocks are signed by the proposer and we only add to the view blocks that are correctly singed by the expected proposer for their slot. We decided not to take this approach to keep the algorithms as compact as possible while maximizing the aspects of the protocols relevant to their properties that are captured by pseudo-code.}.
    Whereas, $\V_p \cup \{\chain_p\}$ denotes the view of $v_p$ combined with the newly \textsc{propose}d chain $\chain_p$.}

    \item[Vote:] During the interval \( [4\Delta t, 4\Delta t + \Delta] \), a validator \( v_i \), upon receiving a proposal $[\textsc{propose}, \chain_p, \V_p, t', v_p]$ from the proposer \( v_p \) for slot~\( t'=t \), sets $\Vfrozen_i \gets \Vfrozen_i \cup \V_p$.

    {During the voting phase, each validator \(v_i\) computes the fork-choice function \(\rlmdghost\). The inputs to $\rlmdghost$ are $v_i$'s {frozen} view $\Vfrozen_i$, the genesis block $\genesis$, and the current slot~$t$.
    
    In the voting phase, validator $v_i$ also sets the confirmed chain $\Chain_i$ to the highest chain among the $\kappa$-deep prefix of the chain output by $\rlmdghost$ and the previous confirmed chain, filtering out this last chain if it is not a prefix of the chain output by $\rlmdghost$.}

    Finally, $v_i$ casts a \textsc{vote} for the output of $\rlmdghost$.

    \item[Fast Confirm:] {Like in \Cref{algo:prob-ga-fast},} validator \(v_i\) checks if a chain $\chain$ can be fast confirmed, i.e., if at least $2/3n$ of the validators cast a \textsc{vote} message for a chain $\chain^\mathrm{cand}$. If that is the case, it sets $\Chain_i$ to this chain.

    \item[Merge:] {Like in \Cref{algo:prob-ga-fast},} at round $4\Delta t + 3\Delta$,  validator $v_i$ updates its variable $\Vfrozen_i$ to be used in the following slot.
\end{description}

\begin{algo}[htb!]
\caption{\RLMDGHOST with fast confirmation -- code for validator $v_i$}
\label{alg:rlmd}
\small
\begin{numbertabbing}\reset
xxxx\=xxxx\=xxxx\=xxxx\=xxxx\=xxxx\=MMMMMMMMMMMMMMMMMMM\=\kill
  {\textbf{Output}} \label{}\\
  \> {\(\Chain_i \gets \genesis\): confirmed chain of validator $v_i$}\label{}\\
  \textbf{State} \label{}\\
  \> \(\V_i^\text{frozen} \gets \{\genesis\} \): frozen view of validator $v_i$ \label{}\\
  {\textbf{function} $\texttt{fastconfirmsimple}(\V, t)$}\label{}\\
  \> {\textbf{let} $\fastcands := \{\chain \colon |\{v_j\colon \exists \chain' \succeq \chain : \ [\textsc{vote}, \chain', \cdot, t,v_j] \in \V_i \}| \geq \frac{2}{3}n\}$}\\
  \> {\textbf{if} $\fastcands \neq \emptyset$ \textbf{then}}\label{}\\
  \>\> {\textbf{return} $\max(\fastcands)$}\label{line:alg2-bcand}\\
  \> {\textbf{else}}\label{}\\
  \>\> {\textbf{return} {$\genesis$}}\label{}\\
  \textsc{Propose}\\
  \textbf{at round} $4\Delta t$ \textbf{do} \label{} \\
  \> \textbf{if} $v_i = v_p^t$ \textbf{then} \label{}\\
  \>\>{\textbf{let} $ \chaincanrlmd := \rlmdghost(\V_i, \genesis, t)^{\lceil 1,t}$} \label{line:alg2-fc1}\\
  \>\> \textbf{let} $ \chain_p := \mathsf{Extend}(\chaincanrlmd,t)$\label{}\\
  \>\> send message [\textsc{propose}, $\chain_p$, $\V_i\,\cup \{\chain_p\}$, $t$, $v_i$] through gossip \label{line:send-propose-rlmd}\\
  \textsc{Vote}\\
  \textbf{at round} $4\Delta t + \Delta$ \textbf{do} \label{}\\
  \> \textbf{let} $ \chaincanrlmd := \rlmdghost(\V_i^{\text{frozen}}, \genesis,t)$ \label{line:alg2-fc2}\\            
  \> $\Chain_i \gets \max(\{\chain \in \{\Chain_i,{\left(\chaincanrlmd\right)^{\lceil\kappa}}\}: \chain \preceq {\chaincanrlmd}\})$\label{line:alg2-vote-chainava-rlmd}\\
  \> send message [\textsc{vote}, {$\chaincanrlmd$}, $\cdot$, $t$, $v_i$] through gossip \label{line:alg2-vote}\\
  \textsc{Fast Confirm}\\
  \textbf{at round} $4\Delta t + 2\Delta$ \textbf{do} \label{}\\
   \> {\textbf{let} $\fastcand := \texttt{fastconfirmsimple}(\V_i,t)$}
  \label{}\\  
  \> \textbf{if} $\fastcand \neq \bot$ \textbf{then}\label{}\\
  \>\> {$\Chain_i \gets \fastcand$}\label{line:alg2-bcand}\\
  \textsc{merge}\\
  \textbf{at round} $4\Delta t + 3\Delta$ \textbf{do} \label{}\\
  \> {$\V_i^{\text{frozen}} \gets \V_i$}\label{}\\
  \\
  \textbf{upon} receiving a gossiped message
  [\textsc{propose}, $\chain_p$, $\V_p$, $t'$, $v_p^t$] \textbf{at any round in} $[4\Delta t, 4\Delta t + \Delta]$ \textbf{do} \label{}\\
  \> $\V_i^{\text{frozen}} \gets \V_i^{\text{frozen}} \cup \V_p$\label{}\\
  [-5ex]
\end{numbertabbing}
\end{algo}

\subsubsection{Analysis}\label{sec:analysis-rlmd}

In this section, we show that \Cref{alg:rlmd}, like \Cref{algo:prob-ga-fast}, satisfied $\eta$ Reorg Resilience {(\Cref{thm:reorg-res-rlmd})}, 
$\eta$ Dynamic Availability (\Cref{thm:dyn-avail-rlmd}), 
$\eta$ Asynchronous Reorg Resilience (\Cref{thm:async-resilience-rlmd}) 
and $\eta$ Asynchronous Safety {Resilience} (\Cref{{thm:async-safety-resilience-rlmd}}).
The analysis presented in this section follows the same structure as the analysis of \Cref{algo:prob-ga-fast} in \Cref{sec:ga-prob-analysis}.
As we will see shortly, the analogous of \Cref{lem:keep-voting-tob-fast-conf,lem:vote-proposal-fast-conf,lem:one-fast-confirm-all-vote-fast-conf} are already proved in \cite{rlmd}.
With these Lemmas, then we can prove that other Lemmas and Theorems by essentially recalling the proofs of the analogous Lemmas and Theorems in \Cref{sec:ga-prob-analysis} with minor modifications.
In the rest of this section, given any slot $t$, we let $\rlmdpropose{t}_i := \rlmdghost(\V^\proposing{t}_i,\genesis,t)^{\lceil 1,t}$ and $\rlmdvote{t}_i := \rlmdghost(\V^\voting{t}_i,\genesis,t)$.

\begin{lemma}
    \label{lem:keep-voting-rlmd-fast-conf}
    If, in slot $t$, all validators in $H_{\voting{t}}$ cast \textsc{vote} messages for chains extending chain \(\chain\),
    then, for any validator $v_i\in H_{\voting{(t+1)}}$, $\rlmdpropose{(t+1)}_i \succeq \chain$ and $\rlmdvote{(t+1)}_i \succeq \chain$, which implies that, in slot $t+1$, all validators in $H_{\voting{(t+1)}}$ cast \textsc{vote} messages for chains extending \(\chain\).
\end{lemma}

\begin{proof}
    See Lemma 4 of~\cite{rlmd}.
\end{proof}

\begin{lemma}
    \label{lem:one-fast-confirm-all-vote-fast-conf-rlmd}
    If an honest validator fast confirms a chain $\chain$ in slot $t$, then, for any slot $t'>t$ and validator $v_i \in H_{\voting{(t')}}$, $\rlmdpropose{t'}_i \succeq \chain$ and $\rlmdvote{t'}_i \succeq \chain$, which implies that, all validators in $H_{\voting{(t')}}$ cast a \textsc{vote} message for a chain extending $\chain$.
\end{lemma}

\begin{proof}
  The proof follows the one for \Cref{lem:one-fast-confirm-all-vote-fast-conf} with the following changes only.
  \begin{enumerate}
    \item The base case follows from Lemma 5 of~\cite{rlmd},
    \item The inductive step follows from \Cref{lem:keep-voting-rlmd-fast-conf}.\qedhere
  \end{enumerate}
\end{proof}

\begin{lemma}
\label{lem:vote-proposal-fast-conf-rlmd}
Let $t$ be a slot with an honest proposer $v_p$ and assume that~$v_p$ casts a $[\textsc{propose}, \chain_p, \V_p, t, v_p]$ message. 
Then, for any slot $t'\geq t$, all validators in $H_{\voting{(t')}}$ cast a \textsc{vote} message for a chain extending {$\chain_p$}.
Additionally, for any slot $t'' > t$ and any validator $v_i\in H_{\voting{(t'')}}$, $\rlmdpropose{t''}_i\succeq \chain_p$ and $\rlmdvote{t''}_i\succeq \chain_p$.
\end{lemma}

\begin{proof}
  The proof follows the one for \Cref{lem:vote-proposal-fast-conf} with the following changes only.
  \begin{enumerate}
    \item The base case follows from Lemma 1 of~\cite{rlmd},
    \item The inductive step follows from \Cref{lem:keep-voting-rlmd-fast-conf}.\qedhere
  \end{enumerate}
\end{proof}

\begin{lemma}
\label{lem:rlmd-confirmed-always-canonical}
    Let $r_i$ be any round and $r_j$ be any round such that $r_j\ge r_i$ and $r_j \in \{\proposing{\slot(r_j)}, \voting{\slot(r_j)}\}$.
    Then, for any validator~$v_i$ honest in round $r_i$ and any validator $v_j \in H_{\slot(r_j)}$, $\Chain^{r_i}_i \preceq \chain^{r_j}_j$.
\end{lemma}

\begin{proof}
  The proof follows the one for \Cref{lem:ga-confirmed-always-canonical} with the following changes only.
  \begin{enumerate}
    \item {Replace $\mfc$ with $\rlmd$,}
    \item Replace \Cref{line:algga-no-ffg-vote-chainava}  of \Cref{algo:prob-ga-fast} with \Cref{line:alg2-vote-chainava-rlmd} of \Cref{alg:rlmd},
    \item Replace \Cref{lem:vote-proposal-fast-conf} with \Cref{lem:vote-proposal-fast-conf-rlmd}, and
    \item Replace \Cref{lem:one-fast-confirm-all-vote-fast-conf} with \Cref{lem:one-fast-confirm-all-vote-fast-conf-rlmd}.\qedhere
  \end{enumerate}
\end{proof}

\begin{theorem}[Reorg Resilience]\label{thm:reorg-res-rlmd}
  \Cref{alg:rlmd} is $\eta$-reorg-resilient.
\end{theorem}

\begin{proof}
  {
  The proof follows the one for \Cref{thm:reorg-res-prop-tob} with the following changes only.
  \begin{enumerate}
    \item Replace $\mfc$ with $\rlmd$,
    \item Replace \Cref{lem:vote-proposal-fast-conf} with \Cref{lem:vote-proposal-fast-conf-rlmd}, and
    \item Replace \Cref{lem:ga-confirmed-always-canonical} with \Cref{lem:rlmd-confirmed-always-canonical}.\qedhere
  \end{enumerate}
  }
\end{proof}

{
\begin{lemma}\label{lem:rlmd-proposer-shorter-than-t}
  For any slot $t \geq 1$ and  validator $v_i \in H_\proposing{t}$, $\rlmdpropose{t}_i.p < t$.
\end{lemma}
\begin{proof}
  Obvious by the definition of $\rlmdpropose{t}_i$.
\end{proof}
}

\begin{theorem}
    \label{thm:dyn-avail-rlmd}
\Cref{alg:rlmd} is $\eta$-dynamically-available.
\end{theorem}

\begin{proof}
  {
    The proof follows the one for \Cref{thm:dyn-avail-fast-conf-tob} with the following changes only.
    \begin{enumerate}
      \item Replace $\mfc$ with $\rlmd$,
      \item Replace \Cref{algo:prob-ga-fast} with \Cref{alg:rlmd},
      \item Replace \Cref{lem:vote-proposal-fast-conf} with \Cref{lem:vote-proposal-fast-conf-rlmd},
      \item Replace \Cref{lem:ga-confirmed-always-canonical} with \Cref{lem:rlmd-confirmed-always-canonical}, and
      \item {Replace \Cref{lem:ga-mfc-proposer-shorter-than-t} with \Cref{lem:rlmd-proposer-shorter-than-t}.}\qedhere
    \end{enumerate}
  }
\end{proof}

{
\begin{lemma}
\label{thm:fast-liveness-rlmd}
Take a slot $t$ in which \(|H_{\voting{t}}| \geq \frac{2}{3}n\).
If in slot $t$ an honest validator sends a \textsc{propose} message for chain $\chain_p$, then, for any validator $ v_i \in H_{\fastconfirming{t}}$, $\chainava^{\fastconfirming{t}}_i \succeq \chain_p$.
\end{lemma}

\begin{proof}
    The proof follows the one for \Cref{thm:fast-liveness-tob} with the only change {being replacing \Cref{lem:vote-proposal-fast-conf} with \Cref{lem:vote-proposal-fast-conf-rlmd}}.
\end{proof}
}

\paragraph{Asynchrony Resilience.}
Here, we demonstrate that \Cref{alg:rlmd} also provides Asynchrony Reorg Resilience and Asynchrony Safety Resilience, as defined in Section~\ref{sec:security}.

\begin{lemma}
  \label{lem:asyn-induction-rlmd}
  Assume $\pi > 0$ and
    let $t$ be any slot in $[t_a, t_a + \pi]$.
    Assume that in any slot $t' \in [t_a,t]$, all validators in $W_{\voting{t'}}$ cast \textsc{vote} messages for chains extending chain $\chain$.
    Then, in slot $t+1$, all validators in $W_\voting{(t+1)}$ also cast \textsc{vote} messages for chains extending chain $\chain$.
\end{lemma}

\begin{proof}
The proof follows the same reasoning of \Cref{lem:asyn-induction}.
\sloppy{{Specifically, note that the condition $\left|\left(\Vfrozen_i\right)^{\chain,t+1}\cap \V_i^{\chain,t+1} \right|> \frac{\left|\mathsf{S}(\V_i,t+1)\right|}{2}$ proved in the proof of \Cref{lem:asyn-induction} is stronger than the condition need for $\rlmdghost(\V_i,\genesis,t+1)$ to output a chain extending $\chain$.}}
\end{proof}

\begin{lemma}\label{lem:asyn-induction2-rlmd}
  {Assume $\pi > 0$ and}
  take a slot $t \leq t_a$ such that, in slot $t$, any validator in $H_{\voting{t}}$ casts a \textsc{vote} message for a chain extending $\chain$.
  Then, for any slot $t_i \geq t$, any validator in $W_{\voting{t_i}}$ casts a \textsc{vote} message for a chain extending $\chain$.
\end{lemma}
\begin{proof}

  The proof follows the one for \Cref{lem:asyn-induction2} with the following changes only.
  \begin{enumerate}
    \item {Replace \Cref{lem:keep-voting-tob-fast-conf} with }\Cref{lem:keep-voting-rlmd-fast-conf}, and 
    \item {Replace \Cref{lem:asyn-induction} with}  \Cref{lem:asyn-induction-rlmd}.\qedhere
  \end{enumerate}
\end{proof}

\begin{lemma}\label{lem:asyn-induction3-rlmd}
  Assume $\pi > 0$ and
  take a slot $t \leq t_a$ such that, in slot $t$, any validator in $H_{\voting{t}}$ casts a \textsc{vote} message for a chain extending $\chain$.
  Then, for any round $r_i \geq \voting{t}$ and validator $v_i \in W_{r_i}$, $\Chain^{r_i}_i$ does not conflict with $\chain$.
\end{lemma}
\begin{proof}

  The proof follows the one for \Cref{lem:asyn-induction3} with the following changes only.
  \begin{enumerate}
    \item Replace \Cref{line:algga-no-ffg-vote-chainava} of \Cref{algo:prob-ga-fast} with \Cref{line:alg2-vote-chainava-rlmd} of \Cref{alg:rlmd},
    \item Replace \Cref{lem:asyn-induction2} with \Cref{lem:asyn-induction2-rlmd}, and
    \item \sloppy{{Note that when considering \Cref{alg:rlmd}, Case 2.2 implies $\Chain_i = \max(\{\chain \colon |\{v_j\colon \exists \chain' \succeq \chain : \ [\textsc{vote}, \chain', \cdot, t,v_j] \in \V_i \}| \geq \frac{2}{3}n\})$, from which the rest of the proof for this case follows unaltered.}}\qedhere
  \end{enumerate}
\end{proof}

\begin{theorem}[Asynchrony Reorg Resilience]
  \label{thm:async-resilience-rlmd}
  \Cref{alg:rlmd} is $\eta$-asynchrony-reorg-resilient.
\end{theorem}
\begin{proof}
  {
    Follows from the proof of \Cref{thm:async-resilience-tob} with the following modifications only.
    \begin{enumerate}
      \item Replace \Cref{thm:reorg-res-prop-tob} with \Cref{thm:reorg-res-rlmd},
      \item Replace \Cref{lem:vote-proposal-fast-conf} with \Cref{lem:vote-proposal-fast-conf-rlmd}, and
      \item Replace \Cref{lem:asyn-induction3} with \Cref{lem:asyn-induction3-rlmd}.\qedhere
    \end{enumerate}
  }
\end{proof}

\begin{theorem}[Asynchrony Safety Resilience]
  \label{thm:async-safety-resilience-rlmd}
  \Cref{alg:rlmd} is $\eta$-asynchrony-safety-resilient.
\end{theorem}
\begin{proof}
  {
    Follows from the proof of \Cref{thm:async-safety-resilience-tob} with the following modifications only.
    \begin{enumerate}
      \item Replace $\mfc$ with $\rlmd$,
      \item Replace \Cref{thm:dyn-avail-fast-conf-tob} with \Cref{thm:dyn-avail-rlmd},
      \item Replace \Cref{lem:ga-confirmed-always-canonical} with \Cref{lem:rlmd-confirmed-always-canonical}, and
      \item Replace \Cref{lem:asyn-induction3} with \Cref{lem:asyn-induction3-rlmd}.\qedhere
    \end{enumerate}
  }
\end{proof}

\subsection{Faster finality protocol execution}

In this section, we present \Cref{alg:rlmd-ffg} which integrates the $\eta$-dynamically-available and reorg-resilient protocol of \Cref{alg:rlmd} with the FFG component introduced in \Cref{sec:ffg} to obtain a $\eta$-secure ebb-and-flow protocol.

We describe \Cref{alg:rlmd-ffg} by discussing the main differences compared to \Cref{alg:rlmd}{, which are quite similar to the differences between \Cref{alg:3sf-tob-noga} and \Cref{algo:prob-ga-fast}.}

\medskip

{To start with, like in \Cref{alg:3sf-tob-noga},} the logic employed by \Cref{alg:rlmd-ffg} to determine fast confirmed chains is encoded in the function $\texttt{fastconfirm}(\V, t)$ which returns the same output of $\texttt{fastconfirmsimple}(\V,t)$ from \Cref{alg:rlmd} as long as the returned fast confirmed chain extends $\GJ(\V).\chain$.
Otherwise, it just returns $\GJ(\V).\chain$.

\medskip

Subsequently, \Cref{alg:rlmd-ffg} maintains the same state variables as those defined in \Cref{alg:rlmd}.

\medskip

{However, compared to \Cref{alg:rlmd},  like  \Cref{alg:rlmd-ffg}, \Cref{alg:rlmd-ffg}outputs two chains.
\begin{itemize}
  \item \textbf{Available Chain $\chainava_i$:} This roughly corresponds to the confirmed chain $\Chain_i$ of \Cref{algo:prob-ga-fast}.
  \item \textbf{Finalized Chain $\chainfin_i$:}
At any fast confirmation round, $\chainfin$ is set to $\GF(\V_i).\chain$, \ie, the chain of the greatest justified checkpoint.
    
  Also, at any vote round, the finalized chain $\chainfin_i$ is set to 
  the longest chain such that of $\chainfin_i \preceq \GF(\V_i).\chain \land \chainfin_i \preceq \chainava_i$.
\end{itemize}}

\medskip

Next, compared to \Cref{alg:rlmd}, \Cref{alg:rlmd-ffg} proceeds as follows.


\begin{description}
    \item[Propose:] The propose phase remains identical to that in \Cref{alg:rlmd}, with the only difference being the input to the fork-choice function $\rlmdghost$. Specifically, instead of using the genesis block as input, validator $v_i$ inputs $GJ(\V_i).\chain$\footnote{{The output of $\rlmdghost(\V_i,\GJ(V_i).\chain,t)$ corresponds to the output of $\hfc(\V_i, t)$ with $\hfc$ being the fork-choice defined in~\cite{DBLP:conf/esorics/DAmatoZ23}.}}.

    \item[Vote:] 
    {First, similarly to  the propose phase, instead of using the genesis block as input to the fork-choice function, validator $v_i$ inputs $GJ(\Vfrozen_i).\chain$.}

    {Second,} in \Cref{alg:rlmd-ffg}, honest validators also include an \textsc{ffg-vote} in the \textsc{vote} messages that they send. 
    Specifically, the source checkpoint of such \textsc{ffg-vote} sent corresponds to $\GJ(\Vfrozen_i)$.
    The target checkpoint is determined by selecting the highest chain among 
    the previous available chain, the $\kappa$-deep prefix of the chain output by the fork-choice function $\rlmdghost$, and $\GJ(\Vfrozen_i).\chain$, filtering out any chain if it is not a prefix of the chain output by the fork-choice function $\rlmdghost$.

    Validator $v_i$ also sets the available chain output, $\chainava_i$ to such chain.

    \item[Fast Confirm:] {Exactly like in \Cref{alg:3sf-tob-noga}, validator $v_i$ checks whether $\chainava_i$ is a prefix of the chain $\mathrm{fast}^{\mathrm{cand}}$ output by $\texttt{fastconfirm}(\V_i, t)$ or it conflicts with it, and,
    in either case, $v_i$ updates $\chainava_i$ to $\mathrm{fast}^{\mathrm{cand}}$.}

    \item[Merge:] The propose phase remains identical to that in \Cref{alg:rlmd}.
\end{description}

\begin{algo}[htb!]
\caption{Faster finality protocol -- code for validator $v_i$}
\label{alg:rlmd-ffg}
\small
\begin{numbertabbing}\reset
xxxx\=xxxx\=xxxx\=xxxx\=xxxx\=xxxx\=MMMMMMMMMMMMMMMMMMM\=\kill
  {\textbf{Output}} \label{}\\
  \> {\(\chainava_i \gets \genesis\): available chain of validator $v_i$}\label{}\\
  \> {\(\chainfin_i \gets \genesis\): finalized chain of validator $v_i$}\label{line:alg2-set-chfin-init}\\
  \textbf{State} \label{}\\
  \> \(\V_i^\text{frozen} \gets \{\genesis\} \): frozen view of validator $v_i$ \label{}\\ 
  {\textbf{function} $\texttt{fastconfirm}(\V,t)$}\label{}\\
  \> {\textbf{let} $\fastcand := \texttt{fastconfirmsimple}(\V,t)$}\label{}\\
  \> {\textbf{if} $\fastcand \neq \bot \land \fastcand \succeq \GJ(\V).\chain$}\label{}\\
  \>\> {\textbf{return} $\fastcand$}\label{}\\
  \> {\textbf{else}}\label{}\\
  \>\> {\textbf{return} $\GJ(\V).\chain$}\label{}\\
  \textsc{Propose}\\
  \textbf{at round} $4\Delta t$ \textbf{do} \label{} \\
  \> \textbf{if} $v_i = v_p^t$ \textbf{then} \label{}\\
  \>\>{\textbf{let} $ \chaincanrlmd := {\rlmdghost}(\V_i, GJ(\V_i).\chain, t)^{\lceil 1,t}$} \label{line:alg2ffg-fc1}\\
  \>\> {\textbf{let} $ \chain_p := \mathsf{Extend}(\chaincanrlmd,t)$}\label{}\\
  \>\> send message [\textsc{propose}, {$\chain_p$}, $\V_i\,\cup \{\chain_p\}$, $t$, $v_i$] through gossip \label{line:send-propose-rlmd-ffg}\\
  \textsc{Vote}\\
  \textbf{at round} $4\Delta t + \Delta$ \textbf{do} \label{}\\
  \> \textbf{let} $ \chaincanrlmd := {\rlmdghost}(\V_i^{\text{frozen}}, \GJ(\V_i^{\text{frozen}}).\chain,t)$ \label{line:alg2-fc2ffg}\\            
  \> $\chainava_i \gets \max(\{\chain \in \{\chainava_i,{\left(\chaincanrlmd\right)^{\lceil\kappa}},\GJ(\V_i^\text{frozen}).\chain\}: \chain \preceq {\chaincanrlmd}\})$\label{line:alg2-vote-chainava}\\
  \> {$\chainfin_i \gets \max(\{\chain \colon \chain \preceq \chainava_i \land \chain \preceq \GF(\V_i).\chain\})$}\label{line:alg2-set-chfin-vote} \\
  \>  \textbf{let} $ \T := (\chainava_i,t)$\label{line:set-target-checkpoint}\\
  \> send message [\textsc{vote}, {$\chaincanrlmd$}, $\GJ(\V_i^\text{frozen}) \to \T$, $t$, $v_i$] through gossip \label{line:alg2ffg-vote}\\
  \textsc{Fast Confirm}\\
  \textbf{at round} $4\Delta t + 2\Delta$ \textbf{do} \label{line:rlmd-at-confirm}\\
   \> {\textbf{let} $\fastcand := \texttt{fastconfirm}(\V_i,t)$}
  \label{line:rlmdffg-set-fastcand-fconf}\\  
  \> \textbf{if} $\chainava_i \nsucceq \fastcand$ \textbf{then}\label{}\\
  \>\> $\chainava_i \gets \fastcand$\label{line:alg2ffg-set-chaava-to-bcand}\\
  \> {$\chainfin_i \gets \max(\{\chain \colon \chain \preceq \chainava_i \land \chain \preceq \GF(\V_i).\chain\})$}\label{line:alg2-set-chfin-fast}\\
  \textsc{merge}\\
  \textbf{at round} $4\Delta t + 3\Delta$ \textbf{do} \label{}\\
  \> {$\V_i^{\text{frozen}} \gets \V_i$}\label{}\\
  \\
  \textbf{upon} receiving a gossiped message
  [\textsc{propose}, $\chain_p$, $\V_p$, $t$, $v_p^t$] \textbf{at any round in} $ [4\Delta t, 4\Delta t + \Delta]$ \textbf{do} \label{}\\
  \> $\V_i^{\text{frozen}} \gets \V_i^{\text{frozen}} \cup \V_p$\label{}\\
  [-5ex]
\end{numbertabbing}
\end{algo}

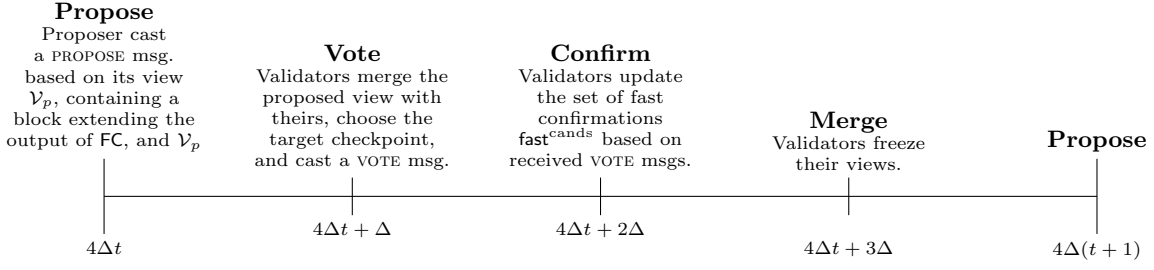
\begin{figure}
    \centering
\begin{tikzpicture}
\scriptsize
\pgfmathsetmacro{\segwidth}{3.3}

\draw (0,0) -- (4*\segwidth,0);

\draw (0,-0.5) -- (0,0.5);
\draw (\segwidth,-0.25) -- (\segwidth,0.25);
\draw (2*\segwidth,-0.25) -- (2*\segwidth,0.25);
\draw (3*\segwidth,-0.25) -- (3*\segwidth,0.25);
\draw (4*\segwidth,-0.5) -- (4*\segwidth,0.5);

\node[above, align=center, text width=2.7cm] at (0,0.5) {{\small \textbf{Propose}} \\ Proposer cast a \textsc{propose} msg. based on its view $\V_p$, containing a block extending the output of $\FC$, and $\V_p$};
\node[above, align=center, text width=2.7cm] at (\segwidth,0.25) {{\small \textbf{Vote}} \\ Validators merge the proposed view with theirs, choose the target checkpoint, and cast a \textsc{vote} msg.};
\node[above, align=center, text width=2.7cm] at (2*\segwidth,0.25) {{\small \textbf{Confirm}} \\ Validators update the set of fast confirmations $\fastcands$ based on received \textsc{vote} msgs.};
\node[above, align=center, text width=2.3cm] at (3*\segwidth,0.25) {{\small \textbf{Merge}} \\ Validators freeze their views.};
\node[above, align=center, text width=2.3cm] at (4*\segwidth,0.5) {{\small \textbf{Propose}}};

\node[below, align=center, text width=2cm] at (0,-0.5) {\textbf{$4\Delta t$}};
\node[below] at (\segwidth,-0.25) {\textbf{$4\Delta t + \Delta$}};
\node[below] at (2*\segwidth,-0.25) {\textbf{$4\Delta t + 2 \Delta$}};
\node[below] at (3*\segwidth,-0.5) {\textbf{$4\Delta t + 3\Delta$}};
\node[below] at (4*\segwidth,-0.5) {\textbf{$4\Delta (t+1)$}};

\end{tikzpicture}
\caption{Slot $t$ of the protocol, with its four phases.}
\label{fig:pvm}
\end{figure}

\subsection{Faster finality protocol analysis}
\label{sec:analysis-rlmd-ffg}

\Cref{alg:rlmd-ffg}, as \Cref{alg:3sf-tob-noga}, works in the generalized partially synchronous sleepy model, and is in particular a $\eta$-secure ebb-and-flow protocol.
In \Cref{sec:ffg-analysis} we have already shown that the the finalized chain $\chainfin$ is $\frac{n}{3}$-accountable, and thus always safe if $f < \frac{n}{3}$.
Now, for $\GST = 0$, we show in Section~\ref{sec:sync-rlmd} that, if the execution is $\eta$-compliant in this stronger sense, then all the properties of {\Cref{alg:rlmd}}, i.e., $\eta$ Dynamic Availability and $\eta$ Reorg Resilience, keep holding. 
Finally, in \Cref{sec:psync-rlmd} we show that, after $\max(\GST, \GAT) + 4\Delta$, \Cref{alg:rlmd-ffg} guarantees the required conditions listed in Property~\ref{prop:succ-for-ffg-liveness} to ensure that the finalized chain is live.






\subsubsection{Synchrony}
\label{sec:sync-rlmd}
In this section, we prove that chain {$\chainava$} of \Cref{alg:rlmd-ffg} is $\eta$-dynamically-available, $\eta$-reorg-resilient, $\eta$-asynchrony-reorg-resilient and $\eta$-asynchrony-safety-resilient.
Throughout this part of the analysis, we assume $\GST=0$, that less than one-third of the entire validator set is ever controlled by the adversary (\ie, $f<\frac{n}{3}$), and that Constraint~\eqref{eq:pi-sleepiness} holds.
%
Observe that, in \RLMDGHOST with fast confirmations (Appendix B~\cite{rlmd}), the assumption $f<\frac{n}{3}$ is strictly needed for safety (and only for clients which use fast confirmations), but for example not for Reorg Resilience or liveness, because fast confirmations do not affect the chain output by the fork-choice function. On the other hand, \Cref{alg:rlmd-ffg} utilizes confirmations as a prerequisite for justification, and justification does affect the chain output by the fork-choice function. 
{This is because in \Cref{alg:rlmd-ffg}, $\rlmdghost$ takes in input the chain of the greatest justified checkpoint and  filters out any branch conflicting with it.} 
Therefore, we require that less than $\frac{n}{3}$ validators are ever {controlled by the adversary, and hence slashable,} for all of the properties which we are going to prove.

Our faster finality protocol implemented in \Cref{alg:rlmd-ffg} uses the $\rlmdghost$ fork-choice function, dealing with checkpoints and justifications. However, one could implement it using also different fork-choice functions. 

We will now adopt a similar approach to the one used in Section~\ref{sec:tob-analysis} to demonstrate that the integration of the FFG component into \Cref{alg:rlmd} does not alter the protocol's behavior in any way that could compromise the properties previously established for \Cref{alg:rlmd} in Section~\ref{sec:recalling-rlmd}. 
To do so, as we did in Section~\ref{sec:tob-analysis}, we take any execution of \Cref{alg:rlmd-ffg} and show that there exists an adversary that induce an execution of \Cref{alg:rlmd} where 
(i) the messages sent by any honest validator in the two executions match except only for the FFG component, 
(ii) the fork-choice outputs of any honest validator in the two executions match and 
(iii) any available chain output by \Cref{alg:rlmd-ffg} is also a confirmed chain of an honest validator in the execution of \Cref{alg:rlmd}.
This then allows us to show that the result of the various Lemmas and Theorems presented in \Cref{sec:recalling-rlmd} for \Cref{alg:rlmd} also hold for \Cref{alg:rlmd-ffg}. 
In the following, if $e$ is an execution of \Cref{alg:rlmd}, then $\specifyExec{e}{\rlmdpropose{t}}$ corresponds to the definition provided in \Cref{sec:analysis-rlmd}. If $e$ is an execution of \Cref{alg:3sf-tob-noga}, then we define $\specifyExec{e}{\rlmdpropose{t}_i} := \specifyExec{e}{\rlmdghost}(\V^{\proposing{t}}_i, \GJ(\V^{\proposing{t}}_i).\chain, t)^{\lceil 1, t}$ and $\specifyExec{e}{\rlmdvote{t}_i} := \specifyExec{e}{\rlmdghost}(\V^{\voting{t}}_i, \GJ(\V^{\voting{t}}_i).\chain, t)$.

{
\begin{definition}
    Let $\FFGExec$ be any $\eta$-compliant execution of \Cref{alg:rlmd-ffg} and $\NoFFGExec$ be an $\eta$-compliant execution  of \Cref{alg:rlmd}.
    We say that $\FFGExec$ and $\NoFFGExec$ are \emph{\rlmdghost-equivalent} if and only if the following constraints hold:
    \begin{enumerate}
    
        \item\label[condition]{cond:0b-2-t-rlmd} $\FFGExec$ and $\NoFFGExec$ are honest-output-dynamically-equivalent up to any round
        \item for any slot $t_i$,
        \begin{enumerate}[label*=\arabic*]
        \item \label[condition]{cond:3-2-rlmd} $\rlmdproposeNoFFG{t_i}_i = \hfcproposeFFG{t_i}_i$
        \item\label[condition]{cond:4-2-rlmd} $\rlmdvoteNoFFG{t_i}_i = \hfcvoteFFG{t_i}_i$
        \item\label[condition]{cond:5-2-rlmd} for any slot $t_j \leq t_i$ and validator $v_j \in H_{\voting{(t_j)}}$, there exists a slot $t_k \leq t_j$ and a validator $v_k \in H_{\voting{(t_k)}}$ such that $\Chain^{\voting{t_k}}_k \succeq \chainava^{\voting{t_j}}_j$.
        \item\label[condition]{cond:6-2-rlmd} for any slot $t_j \leq t_i$ and validator $v_j \in H_{\voting{(t_j)}}$, there exists a slot $t_k \leq t_j$ and a validator $v_k \in H_{\voting{(t_k)}}$ such that $\Chain^{\voting{t_k}}_k \succeq \chainava^{\fastconfirming{t_j}}_j$.
        \end{enumerate}
    \end{enumerate}
\end{definition}

\begin{lemma}\label{lem:equiv-rlmd}
    Let $\FFGExec$ be any $\eta$-compliant execution of \Cref{alg:rlmd-ffg}.
    There exists an $\eta$-compliant execution $\NoFFGExec$ of \Cref{alg:rlmd} such that $\FFGExec$ and $\NoFFGExec$ \rlmdghost-equivalent.
\end{lemma}
}
\begin{proof}
  {
    Follows from the proof of \Cref{lem:equiv-ga2} with the following modifications only.
    \begin{enumerate}
      \item Replace $\mfc$ with $\rlmd$,
      \item Replace \Cref{lem:ga-confirmed-always-canonical} with \Cref{lem:rlmd-confirmed-always-canonical},
      \item Replace \Cref{line:algtob-vote-chainava} of \Cref{alg:3sf-tob-noga} with \Cref{line:alg2-vote-chainava} of \Cref{alg:rlmd-ffg},
      \item Replace $\GJfrozen[\voting{t_i}]_i$ with $\GJ(\V^{\voting{t_i}}_i)$,
      \item In the treatment of Case 2 of \Cref{cond:5-2}, note that, due to the change above, we do not need to refer to any line to show that $\mathsf{J}(\GJ(\V^{\voting{t_i}}_i),\V^{\voting{t_i}}_i)$,
      \item Replace \Cref{line:algga-no-ffg-vote-chainava} of \Cref{algo:prob-ga-fast} with \Cref{line:alg2-vote-chainava-rlmd} of \Cref{alg:rlmd},
      \item Replace the reference to \Crefrange{line:algotb-at-confirm}{line:algtob-set-chaava-to-bcand} of \Cref{alg:3sf-tob-noga} with reference to \Crefrange{line:rlmd-at-confirm}{line:alg2ffg-set-chaava-to-bcand} of \Cref{alg:rlmd-ffg}, and 
      \item In the treatment of Case 3 of \Cref{cond:6-2}, replace $(\chain^C, Q) = \texttt{fastconfirm}(\VFFG^{\fastconfirming{t_i}}_i,t_i) \land Q\neq \emptyset$ with $\chain^C = \texttt{fastconfirm}(\VFFG^{\fastconfirming{t_i}}_i,t_i) \land \chain^C\neq \GJ(\VFFG^\merging{t_i}_i.\chain)$ and let $\texttt{fastconfirm}$ be as defined in \Cref{alg:rlmd-ffg}.\qedhere
    \end{enumerate}
  }
\end{proof}

As a result of \Cref{lem:equiv-rlmd}, we can derive analogous lemmas to \Cref{lem:keep-voting-rlmd-fast-conf}, \Cref{lem:one-fast-confirm-all-vote-fast-conf-rlmd}, and \Cref{lem:vote-proposal-fast-conf-rlmd}, similar to the process followed with \Cref{alg:3sf-tob-noga}. 

In the following Lemmas and Theorems, unless specified, we refer to executions of \Cref{alg:rlmd-ffg}.

\begin{lemma}[Analogous of \Cref{lem:keep-voting-rlmd-fast-conf}]
  \label{lem:keep-voting-tob-fast-conf-ffg-rlmd}
  If, in slot $t$, all validators in $H_{\voting{t}}$ cast \textsc{vote} messages for chains extending chain \(\chain\),
  then, for any validator $v_i\in H_{\voting{(t+1)}}$, $\hfcpropose{(t+1)}_i \succeq \chain$ and $\hfcvote{(t+1)}_i \succeq \chain$, which implies that, in slot $t+1$, all validators in $H_{\voting{(t+1)}}$ cast \textsc{vote} messages for chains extending \(\chain\).
\end{lemma}
\begin{proof}
  {
    Follows from the proof of \Cref{lem:keep-voting-tob-fast-conf-ffg} with the following modifications only.
    \begin{enumerate}
      \item Replace $\mfc$ with $\hfcvp$,
      \item Replace \Cref{algo:prob-ga-fast} with \Cref{alg:rlmd},
      \item Replace \Cref{alg:3sf-tob-noga} with \Cref{alg:rlmd-ffg}, and
      \item Replace \Cref{lem:equiv-ga2} with \Cref{lem:equiv-rlmd}
      \qedhere
    \end{enumerate}
  }
\end{proof}

\begin{lemma}[Analogous of \Cref{lem:vote-proposal-fast-conf-rlmd}]
  \label{lem:vote-proposal-fast-conf-ffg-rlmd}
  Let $t$ be a slot with an honest proposer $v_p$ and assume that~$v_p$ casts a $[\textsc{propose}, \chain_p, \V_p, t, v_p]$ message. 
  Then, for any slot $t'\geq t$, all validators in $H_{\voting{(t')}}$ cast a \textsc{vote} message for a chain extending {$\chain_p$}.
\end{lemma}
\begin{proof}
  {
    Follows from the proof of \Cref{lem:vote-proposal-fast-conf-ffg} with the following modifications only.
    \begin{enumerate}
      \item Replace $\mfc$ with $\hfcvp$,
      \item Replace \Cref{algo:prob-ga-fast} with \Cref{alg:rlmd},
      \item Replace \Cref{alg:3sf-tob-noga} with \Cref{alg:rlmd-ffg},
      \item Replace \Cref{lem:vote-proposal-fast-conf} with \Cref{lem:vote-proposal-fast-conf-rlmd}, and
      \item Replace \Cref{lem:equiv-ga2} with \Cref{lem:equiv-rlmd}.
      \qedhere
    \end{enumerate}
  }
\end{proof}

We now prove $\eta$ Reorg Resilience. Observe that D'Amato and Zanolini define $\eta$-reorg-resilient in a slightly different, albeit related, manner. The authors in~\cite{rlmd} state that $\chainava$ is $\eta$-reorg-resilient if any honest proposal $B$ from a slot $t$ is always included in the chain outputted by the fork-choice function of all active validators at any propose and vote rounds equal or higher than $\voting{t}$ in all $\eta$ compliant executions. This definition differs from ours, as we define Reorg Resilience with respect to the confirmed chain, rather than the chain outputted by the fork-choice function. However, we will now demonstrate how the definition provided by D'Amato and Zanolini~\cite{rlmd} implies our definition, obtaining the following result.

We begin by demonstrating that the confirmed chain at round \( r \), denoted as \( \chainava^r \), is always a prefix of the chain outputted by the fork-choice function of any active validator in any round \( r' \ge r \). Specifically, we show that if \( r \) is a propose or vote round, then the confirmed chain is always a prefix of the chain outputted by the fork-choice function in every round of slot \(\slot(r)\). Conversely, if \( r \) is a confirm or merge round, then the confirmed chain is always a prefix of the chain outputted by the fork-choice function in every round of slot \( t' > \slot(r) \). This result will then be used in the proof of $\eta$ Reorg Resilience.

\begin{lemma}[Analogous of \Cref{lem:rlmd-confirmed-always-canonical}]\label{lemma:chava-prefix-canonical-ffg}
  Let $r_i$ be any round and $r_j$ be any round such that $r_j\ge r_i$ and $r_j \in \{\proposing{\slot(r_j)},\voting{\slot(r_j)}\}$. Then, for any validator~$v_i$ honest in round $r_i$ and any validator $v_j \in H_{\slot(r_j)}$, 
  {$\chainava^{r_i}_i \preceq \rlmd^{r_j}_j$}.
\end{lemma}

\begin{proof}

  The proof follows the one for \Cref{lem:ga-confirmed-always-canonical-ffg} with the following changes only.
  \begin{enumerate}
    \item {Replace \Cref{algo:prob-ga-fast} with \Cref{alg:rlmd},}
    \item {Replace \Cref{alg:3sf-tob-noga} with \Cref{alg:rlmd-ffg},}
    \item Replace \Cref{lem:equiv-ga2} with \Cref{lem:equiv-rlmd}, and
    \item Replace \Cref{lem:ga-confirmed-always-canonical} with \Cref{lem:rlmd-confirmed-always-canonical}. \qedhere
  \end{enumerate}
\end{proof}

\begin{theorem}[Reorg Resilience]\label{thm:reorg-res-prop-rlmd-ffg}
  The chain $\chainava$ output by \Cref{alg:rlmd-ffg} is $\eta$-reorg-resilient.
\end{theorem}

\begin{proof}
  {
    As in the proof of \Cref{thm:reorg-res-prop-tob-ffg}, we can follow the proof of \Cref{thm:reorg-res-prop-tob} with the following changes only.
    \begin{enumerate}
      \item Replace {$\Chain$ with $\chainava$},
      \item Replace \Cref{lem:vote-proposal-fast-conf} with \Cref{lem:vote-proposal-fast-conf-ffg-rlmd}, and
      \item Replace \Cref{lem:ga-confirmed-always-canonical} with \Cref{lemma:chava-prefix-canonical-ffg}.\qedhere
    \end{enumerate}
  }
\end{proof}

\begin{theorem}[$\eta$-dynamic-availability]
  \label{thm:dyn-avail-rlmd-ffg}
\Cref{alg:rlmd-ffg} is $\eta$-dynamically-available.
\end{theorem}

\begin{proof}
      {Note that \Cref{lem:rlmd-proposer-shorter-than-t} holds for \Cref{alg:rlmd-ffg} as well.}
      Then, as in the proof of \Cref{thm:dyn-avail-fast-conf-tob-ffg}, we can follow the proof of \Cref{thm:dyn-avail-fast-conf-tob} with the following changes only.
      \begin{enumerate}
        \item Replace {$\Chain$ with $\chainava$},
        \item Replace the reference to \Cref{line:algga-no-ffg-vote-chainava} of \Cref{algo:prob-ga-fast} with \Cref{line:alg2-vote-chainava} of \Cref{alg:rlmd-ffg},
        \item Replace \Cref{lem:vote-proposal-fast-conf} with \Cref{lem:vote-proposal-fast-conf-ffg-rlmd},
        \item Replace \Cref{lem:ga-confirmed-always-canonical} with \Cref{lemma:chava-prefix-canonical-ffg},
        \item {Replace \Cref{lem:ga-mfc-proposer-shorter-than-t} with \Cref{lem:rlmd-proposer-shorter-than-t}, and}
        \item In the proof of liveness, consider the additional case mentioned ad point 5 of the proof of \Cref{thm:dyn-avail-fast-conf-tob-ffg} with the following changes only.
        \begin{enumerate}[label*=\arabic*]
          \item Replace $\mfc$ with $\hfcvp$ and
          \item Replace \Cref{lem:ga-confirmed-always-canonical-ffg} with \Cref{lemma:chava-prefix-canonical-ffg}.\qedhere
        \end{enumerate}
      \end{enumerate}
\end{proof}

\paragraph{Asynchrony Resilience.}

We now proceed to demonstrate that \Cref{alg:rlmd-ffg} also satisfies Asynchrony Resilience, similar to the proof provided for \Cref{alg:3sf-tob-noga} in Section~\ref{sec:ga-based}. 

Due to the FFG component, establishing Asynchrony Resilience for \Cref{alg:rlmd-ffg} necessitates that Constraint~\eqref{eq:async-condition2} holds, just as it was required for \Cref{alg:3sf-tob-noga}. 

\begin{lemma}[Analogous of \Cref{lem:asyn-induction2-rlmd}]\label{lem:asyn-induction2-rlmd-ffg}
  Assume $\pi > 0$ and
  take a slot $t \leq t_a$ such that, in slot $t$, any validator in $H_{\voting{t}}$ casts a \textsc{vote} message for a chain extending $\chain$.
  Then, for any slot $t_i \geq t$, any validator in $W_{\voting{t_i}}$ casts a \textsc{vote} message for a chain extending $\chain$.
\end{lemma}

\begin{proof}
  The proof follows the one for \Cref{lem:asyn-induction2-ffg} with the following changes only.
  \begin{enumerate}
    \item {Replace \Cref{alg:3sf-tob-noga} with \Cref{alg:rlmd-ffg}},
    \item Replace the reference to \Cref{line:algtob-vote-chainava} of \Cref{alg:3sf-tob-noga} with the reference to \Cref{line:alg2-vote-chainava} of \Cref{alg:rlmd-ffg},
    \item Replace \Cref{lem:keep-voting-tob-fast-conf-ffg} with \Cref{lem:keep-voting-tob-fast-conf-ffg-rlmd},
    \item Replace \Cref{lem:asyn-induction} with \Cref{lem:asyn-induction-rlmd},
    \item Replace \Cref{lem:asyn-induction2} with \Cref{lem:asyn-induction2-rlmd}, and
    \item Replace \Cref{lem:keep-voting-tob-fast-conf} with \Cref{lem:keep-voting-rlmd-fast-conf}. \qedhere
  \end{enumerate}
\end{proof}

\begin{lemma}[Analogous of \Cref{lem:asyn-induction3-rlmd}]\label{lem:asyn-induction3-rlmd-ffg}
  Assume $\pi > 0$ and
  take a slot $t \leq t_a$ such that, in slot $t$, any validator in $H_{\voting{t}}$ casts a \textsc{vote} message for a chain extending $\chain$.
  Then, for any round $r_i \geq \voting{t}$ and validator $v_i \in W_{r_i}$, $\Chain^{r_i}_i$ does not conflict with $\chain$.
\end{lemma}
\begin{proof}
 {
  As in the proof of \Cref{lem:asyn-induction3-ffg}, we can follow the same reasoning used in the proof of \Cref{lem:asyn-induction3} with the following changes only.
  \begin{enumerate}
    \item Replace {$\Chain$ with $\chainava$},
    \item Replace \Cref{lem:asyn-induction2} with  \Cref{lem:asyn-induction2-rlmd-ffg},
    \item In the proof of Case 1, refer to lines \Cref{line:alg2-vote-chainava,line:set-target-checkpoint,line:alg2ffg-vote} of \Cref{alg:rlmd-ffg} rather than \Crefrange{line:algga-no-ffg-vote-chainava}{line:algga-no-ffg-vote-comm} of \Cref{algo:prob-ga-fast}.
    \item For Case 2, consider the additional sub case and related proof as per point 4 of \Cref{lem:asyn-induction3-ffg}.\qedhere
  \end{enumerate}
 }
\end{proof}

\begin{theorem}[Asynchrony Reorg Resilience - Analogous of \Cref{thm:async-resilience-rlmd}]
  \label{thm:async-resilience-rlmd-ffg}
  \Cref{alg:rlmd-ffg} is $\eta$-asynchrony-reorg-resilient.
\end{theorem}
\begin{proof}
  We can follow the same reasoning used in the proof of \Cref{thm:async-resilience-tob-ffg}.
  {
    As in the proof of \Cref{thm:async-resilience-tob-ffg}, we can follow the same reasoning used in the proof of \Cref{thm:async-safety-resilience-tob} with the following changes only.
    \begin{enumerate}
      \item Replace {$\Chain$ with $\chainava$},
      \item Replace \Cref{thm:dyn-avail-fast-conf-tob} with \Cref{thm:dyn-avail-rlmd-ffg}, and 
      \item Replace \Cref{lem:asyn-induction3} with \Cref{lem:asyn-induction3-rlmd-ffg}.\qedhere
    \end{enumerate}
  }
\end{proof}

\begin{theorem}[Asynchrony Safety Resilience - Analogous of \Cref{thm:async-safety-resilience-rlmd}]
  \label{thm:async-safety-resilience-rlmd-ffg}
  \Cref{alg:rlmd-ffg} is $\eta$-asynchrony-safety-resilient.
\end{theorem}
\begin{proof}
  {
    As in the proof of \Cref{thm:async-safety-resilience-tob-ffg}, we can follow the same reasoning used in the proof of \Cref{thm:async-safety-resilience-tob} with the following changes only.
    \begin{enumerate}
      \item Replace {$\Chain$ with $\chainava$},
      \item Replace \Cref{thm:dyn-avail-fast-conf-tob} with \Cref{thm:dyn-avail-rlmd-ffg}, and 
      \item Replace \Cref{lem:asyn-induction3} with \Cref{lem:asyn-induction3-rlmd-ffg}.\qedhere
    \end{enumerate}
  }
\end{proof}

\subsubsection{Partial synchrony}
\label{sec:psync-rlmd}
{
In this section, we show that \Cref{alg:rlmd-ffg}, like \Cref{alg:3sf-tob-noga}, ensures \Cref{prop:never-slashed,prop:chfin,prop:succ-for-ffg-liveness} and hence ensures that the chain $\chainfin$ is always Accountably Safe and is live after time $\max(\GST, \GAT) + {\Delta}$, meaning that \Cref{alg:rlmd-ffg} is an $\eta$-secure ebb-and-flow protocol.}

\begin{lemma} \label{thm:new-slashed-3sf}
  {\Cref{alg:rlmd-ffg} satisfies \Cref{prop:never-slashed}.}
\end{lemma}

\begin{proof}

  {Let us prove that Algorithm \ref{alg:rlmd-ffg} ensures all the conditions listed in \Cref{prop:never-slashed}.}

  \begin{description}
    \item[\Cref{prop:never-slashed-1}.] By \Cref{line:alg2-vote}, an honest active validator sends only one \textsc{ffg-vote} in any slot.
    \item[\Cref{prop:never-slashed-2}.] Direct consequence of \Cref{line:set-target-checkpoint}.
    \item[\Cref{prop:never-slashed-3}.] Take any two slots $t$ and $t'$ with $t<t'$.
    Due to \Cref{line:alg2-vote},the source checkpoints of \textsc{ffg-vote} sent by an honest validator $v_i$ in slots $t$ and $t'$ are $\GJ(\Vfrozen[\voting{t}]_i)$ and $\GJ(\Vfrozen[\voting{t'}]_i)$, respectively.
    Since $\Vfrozen_i$ only keeps growing from $t$ to $t'$, we have $\GJ(\Vfrozen[\voting{t}]_i) \leq \GJ(\Vfrozen[\voting{t'}]_i)$. 
    \qedhere
  \end{description}
\end{proof}

\begin{lemma}
  \label{lem:vote-proposal-fast-conf-ffg-only-base-rlmd}
  Let $t$ be a slot with an honest proposer $v_p$  such that $\proposing{t} \geq \GST + \Delta$ and assume that~$v_p$ casts a $[\textsc{propose}, \chain_p, \V_p, \GJ_p, t, v_p]$ message.
  Then, for any validator $v_i \in H_{\voting{t}}$, $v_i$ casts a \textsc{vote} message for chain $\chain_p$ and $\GJfrozen[\voting{t}]_i = \GJ_p$.
\end{lemma}  
\begin{proof}
  Follows from Lemma 2 of~\cite{DBLP:conf/esorics/DAmatoZ23}.
  \rs{Lemma 2 of~\cite{DBLP:conf/esorics/DAmatoZ23} is not an exact match so we may need to rewrite this proof, but I think this is good enough for now.}
\end{proof}

\begin{lemma}
  \label{thm:liveness-rlmd-ghost}
  {\Cref{alg:rlmd-ffg} satisfies \Cref{prop:succ-for-ffg-liveness}.}
\end{lemma}

\begin{proof}
  {Let us prove that \Cref{alg:rlmd-ffg} ensures each condition listed in \Cref{prop:succ-for-ffg-liveness}.}

  \begin{description}
    \item[\Cref{prop:succ-for-ffg-liveness-1-1}.]
    Let $[\textsc{vote}, \chain, \calS \to \T,t,v_i]$ be the \textsc{vote} message cast by validator $v_i$ in slot $t$.
    We need to show that $\calS.\chain \preceq \T.\chain \preceq \chain$.

    \Cref{line:alg2-vote-chainava,line:set-target-checkpoint,line:alg2ffg-vote} imply that $\T.\chain \preceq \hfcvote{t}_i = \chain$.

    Then, we are left with showing that $\calS.\chain \preceq \T.\chain$.
    Due to \Cref{line:alg2-vote-chainava,line:set-target-checkpoint,line:alg2ffg-vote}, this amounts to showing that $\chainava_i \succeq \GJ(\Vfrozen_i).\chain$ which is implied by the fact that
    $\rlmdghost(\Vfrozen_i,GJ(\Vfrozen_i).\chain,t)\succeq \GJ(\Vfrozen_i).\chain$

    \item[\Cref{prop:succ-for-ffg-liveness-1-2}.]
    Obvious from \Cref{line:alg2ffg-vote}.
    
    \item[\Cref{prop:succ-for-ffg-liveness-1-3}.] 
    \Cref{line:set-target-checkpoint,line:alg2ffg-vote} imply that any \textsc{ffg-vote} $\calS_i \to \T_i$ sent by an honest validator during a slot $t'$ is such that $\T_i.c=t'$.

    \item[{\Cref{prop:succ-for-ffg-liveness-2-1}.}] {Follows from \Cref{lem:rlmd-proposer-shorter-than-t}.}
    
    \item[\Cref{prop:succ-for-ffg-liveness-2-2}.]
    Follows the same reasoning used in the proof of \Cref{prop:succ-for-ffg-liveness-2-2} in \Cref{thm:liveness-tob-no-ga} with the following changes only.
    \begin{enumerate}
      \item Replace \Cref{lem:vote-proposal-fast-conf-ffg-only-base} with \Cref{lem:vote-proposal-fast-conf-ffg-only-base-rlmd},
      \item Replace \Cref{line:algtob-vote-chainava} of \Cref{alg:3sf-tob-noga} with \Cref{line:alg2-vote-chainava} of \Cref{alg:rlmd-ffg}, and
      \item Replace \Cref{line:algtob-vote}  of \Cref{alg:3sf-tob-noga}  with \Cref{line:alg2ffg-vote} of \Cref{alg:rlmd-ffg}.
    \end{enumerate}

    \item[\Cref{prop:succ-for-ffg-liveness-2-3-2}.] {As per proof of \Cref{prop:succ-for-ffg-liveness-2-3-2} in \Cref{thm:liveness-tob-no-ga}.}

    \item[\Cref{prop:succ-for-ffg-liveness-2-3-1}.]
    Follows the same reasoning used in the proof of \Cref{prop:succ-for-ffg-liveness-2-3-1} in \Cref{thm:liveness-tob-no-ga} with the following changes only.
    \begin{enumerate}
      \item Replace $\texttt{fastconfirm}(\V^\fastconfirming{t}_i,t) = (\chain_p, \cdot)$ with $\texttt{fastconfirm}(\V^\fastconfirming{t}_i,t) = \chain_p$,
      \item Replace \Cref{lem:vote-proposal-fast-conf-ffg-only-base} with \Cref{lem:vote-proposal-fast-conf-ffg-only-base-rlmd},

      \item Replace \Cref{line:algtob-vote-chainava} of \Cref{alg:3sf-tob-noga} with \Cref{line:alg2-vote-chainava} of \Cref{alg:rlmd-ffg},
      \item Replace \Cref{line:algtob-set-target-checkpoint} of \Cref{alg:3sf-tob-noga} with \Cref{line:set-target-checkpoint} of \Cref{alg:rlmd-ffg},
      
      \item Replace \Cref{line:algtob-vote} of \Cref{alg:3sf-tob-noga} with \Cref{line:alg2ffg-vote} of \Cref{alg:rlmd-ffg}, and
      \item Replace \Cref{line:algtob-set-chaava-to-bcand} of \Cref{alg:3sf-tob-noga} with \Cref{line:alg2ffg-set-chaava-to-bcand} of \Cref{alg:rlmd-ffg}.     
    \end{enumerate} 
    \item[\Cref{prop:succ-for-ffg-liveness-2-4-1}.] Same reasoning as per proof of \Cref{prop:succ-for-ffg-liveness-2-3-2}.
    \item[\Cref{prop:succ-for-ffg-liveness-2-4-2}]  
    Follows the same reasoning used in the proof of \Cref{prop:succ-for-ffg-liveness-2-4-2} in \Cref{thm:liveness-tob-no-ga} with the following changes only.
    \begin{enumerate}
      \item Replace \Cref{line:algtob-set-mfc} of \Cref{alg:3sf-tob-noga} with \Cref{line:alg2-fc2ffg} of \Cref{alg:rlmd-ffg},
      \item Replace \Cref{line:algtob-vote-chainava} of \Cref{alg:3sf-tob-noga} with \Cref{line:alg2-vote-chainava} of \Cref{alg:rlmd-ffg},
      \item Remove \Cref{line:algtob-vote-comm},
      \item Replace \Cref{line:algtob-vote} of \Cref{alg:3sf-tob-noga} with \Cref{line:alg2ffg-vote} of \Cref{alg:rlmd-ffg}, and
      \item Replace \Cref{line:algotb-set-chfin-fast} of \Cref{alg:3sf-tob-noga} with \Cref{line:alg2-set-chfin-fast} of \Cref{alg:rlmd-ffg}.\qedhere
    \end{enumerate}
  \end{description}
\end{proof}

\begin{theorem}\label{thm:rlmd-ffg-ebb-and-flow}
  \Cref{alg:rlmd-ffg} is a $\eta$-secure ebb-and-flow protocol.
\end{theorem}
\begin{proof}
{  Follows from the proof of \Cref{thm:ga-ebb-and-flow} with the following modifications only.
  \begin{enumerate}
    \item Replace \Cref{alg:3sf-tob-noga} with \Cref{alg:rlmd-ffg},
    \item Replace \Cref{line:algotb-set-chfin-init,line:algotb-set-chfin-vote,line:algotb-set-chfin-fast} of \Cref{alg:3sf-tob-noga} with \Cref{line:alg2-set-chfin-init,line:alg2-set-chfin-vote,line:alg2-set-chfin-fast} of \Cref{alg:rlmd-ffg} respectively,
    \item Replace \Crefrange{line:algtob-set-fastcand-fconf}{line:algtob-set-chaava-to-bcand} of \Cref{alg:3sf-tob-noga} with \Crefrange{line:rlmdffg-set-fastcand-fconf}{line:alg2ffg-set-chaava-to-bcand} of \Cref{alg:rlmd-ffg},    
    \item Replace \Cref{thm:dyn-avail-fast-conf-tob-ffg} with \Cref{thm:dyn-avail-rlmd-ffg},
    \item Replace \Cref{lem:never-slashed-3sf-tob-noga} with \Cref{thm:new-slashed-3sf}, and
    \item Replace \Cref{thm:liveness-tob-no-ga} with \Cref{thm:liveness-rlmd-ghost}.\qedhere
  \end{enumerate}}
\end{proof}

We conclude this section showing that \Cref{alg:rlmd-ffg} also ensures that the finalized chain of an honest validator grows monotonically.

\begin{lemma}\label{lem:chfin-rlmd-always-grows}
  For any two round $r' \geq r$ and validator $v_i$ honest in round $r'$, $\chainfin^{r'}_i \succeq \chainfin^r_i$.
\end{lemma}
\begin{proof}
  Follows the same reasoning used in the proof of \Cref{lem:chfin-ga-always-grows} with the following changes only.
  \begin{enumerate}
    \item Replace $\GJfrozen[r']_i$ with $\GJ(\Vfrozen[r']_i)$,
    \item \sloppy{In the treatment of Case 1, $\chainava^{r'}_i \succeq \GJ(\Vfrozen[r']_i).\chain$ is derived from \Cref{line:alg2-fc2ffg,line:alg2-vote-chainava} and $\mathsf{J}(\GJ(\Vfrozen[r']_i, \V^{r'}_i))$ is derived from $\Vfrozen[r']_i \subseteq \V^{r'}_i$, and}
    \item Replace $(\chainava^{r'}_i,\cdot) = \texttt{fastconfirm}(\V^{r'}_i,\slot(r')) $ with $\chainava^{r'}_i= \texttt{fastconfirm}(\V^{r'}_i,\slot(r'))$.\qedhere
  \end{enumerate}
\end{proof}
  

\section{Practical considerations}
\label{sec:practical-consideration}

The protocols we just presented in \Cref{sec:ga-based,sec:rlmd-based}, although theoretically sound as demonstrated by the theorems, presents practical challenges. 
In this section, we discuss implementation techniques that can be employed to mitigate such challenges.
We proceed by discussing peculiarities of \Cref{alg:3sf-tob-noga} first, then we move to addressing aspects unique to \Cref{alg:rlmd-ffg}, before discussing mitigations that apply to both protocols.
{We conclude by analyzing message complexity of both protocols.}

\subsection{Optimizations} \label{sec:optimizations}

\begin{enumerate}
    \item \textbf{Optimizations only for \Cref{alg:3sf-tob-noga}:}
    Every validator participating in \Cref{alg:3sf-tob-noga} tracks all the received \textsc{vote} messages, including those that equivocate.
    This approach is necessary due to the way the fork-choice function $\mfc$ operates. 

    The fork-choice function $\mfc$ takes two views, $\V$ and $\V'$ (with the possibility that $\V = \V'$ for the proposer), as input. It then considers the latest and unexpired \textsc{vote} messages sent by validators that never equivocated within each of these views.
    The function checks if the intersection of these filtered views contains more than half of all the senders within $\V'$ as
    Consequently, a generic view $\V$ must contain information regarding \textsc{vote} messages, equivocations, and the senders of these messages.
    This information is crucial for the fork-choice function to determine its output chain.

    However, there is no need to retain unexpired messages within $\V$, as these will be filtered out during the evaluation of the fork-choice function.
    Therefore, at each slot, every validator \(v_i\) can prune their local view \(\V_i\) and \(\Vfrozen_i\) by removing all expired messages and retaining only the unexpired ones.
    This approach streamlines the process by reducing the amount of data that validators need to store and process, enhancing the overall efficiency of the protocol.
    However, we need to note that for a given chain $\chain$, this optimization leads to the set $\V^{\chain,t}$ having potentially more elements than in the original protocol.
    This is because the original definition of $\V^{\chain,t}$ removes first the \textsc{vote}s of any validator that has ever equivocated in $\V$, regardless of the age of their \textsc{vote} messages.
    However, it should be easy to see that this change does not invalidate any of the proofs conducted with respect to the original protocol.

    \item \textbf{Optimizations only for \Cref{alg:rlmd-ffg}:} The way in which views are used in \Cref{alg:rlmd-ffg} is not optimally implementable.
    Recall that every validator executing \Cref{alg:rlmd-ffg} keeps track of a view that is constantly updated with each message the validator receives. Moreover, the proposer for slot \( t \) broadcasts the entire view, containing all the messages ever received, to all the other validators (\Cref{line:send-propose-rlmd}, \Cref{alg:rlmd-ffg}). This is done to achieve key properties such as ensuring that under synchrony with an honest proposer, the view of validators at the the vote round is a subset of the proposer's view at the propose round, leading the output of the fork-choice function ${\rlmdghost}$ at the vote round to match the chain \textsc{propose}d by the proposer in that slot.
    
    Clearly, from a practical standpoint, having a constantly growing set of messages that is forwarded at each slot poses significant issues. This approach leads to inefficiencies and increased overhead.
    
    However, these issues can be greatly mitigated with some considerations. To start, observe that the fork choice function $\rlmdghost$ outputs the chain with the accumulated most weight from unexpired latest messages. This means that a proposer could potentially send only the messages that contribute to this chain as far as the fork choice function is concerned, along with all the \textsc{vote} messages carrying \textsc{ffg-vote}s that justify the greatest justified checkpoint used in the fork choice function. This way, a large number of messages that do not contribute to the output of the fork-choice function are excluded, reducing the message overhead.
    
    The intuition behind this is as follows.
    Take a view $\V$ and a slot $t$, and remove from $\V$ all the \textsc{vote} messages that are not in support of prefixes of the chain ${\hfcvp}(\V,\GJ(\V).\chain,t)$.
    Let $\V'$ be the resulting view.
    Clearly, ${\hfcvp}(\V,\GJ(\V).\chain,t) ={\hfcvp}(\V',GJ(\V').\chain,t)$.
    Therefore, by implementing the optimization above, we still have the property that, under synchrony, the \hfcvp output by honest validators in a given slot matches the chain \textsc{propose}d by an honest validator is that same slot, which is all that we effectively need to ensure all the various properties holding under synchrony.
    Additionally, in protocols like {the} Ethereum's consensus protocols, votes are included in blocks.
    This means that validators can derive the proposer's view -- or more precisely, the useful messages used by the proposer to output its chain from the fork-choice function -- from the chain it proposes.
    
    Thus, \Cref{alg:rlmd-ffg} can be made more practical by removing the overall message overhead of sending all messages all the time. While we used the notion of view to clarify the properties and behavior of our protocol, in practice, this can be avoided.

    \item \textbf{Optimizations for both \Cref{alg:3sf-tob-noga,alg:rlmd-ffg}:} 
    In both \Cref{alg:3sf-tob-noga,alg:rlmd-ffg}, proposers send chains.
    However, this does not translate with the proposer having to send all the blocks in that chain in practice.
    In fact, in both protocols, the proposer extends the chain output by the respective fork-choice function at the time of proposing.
    Hence, due to gossiping and the synchrony assumptions leveraged upon, it is sufficient for the proposer to send only the blocks added to the chain output by the fork-choice function.
    While theoretically this can be more than one block (in the case that the output of the fork-choice function has length strictly lower than the previous slot number), 
    in practice, given that this is a quite unlikely scenario, most probably the proposing {action will be} coded to just add one block to the chain output by the fork-choice function.
    {In addition, some local optimizations in the use of views can still be made in \Cref{alg:3sf-tob-noga,alg:rlmd-ffg}. 
    For instance, there is no need to retain unexpired messages within $\V$, as these will be filtered out during the evaluation of the fork-choice function. 
    Therefore, at each slot, every validator \(v_i\) can prune their local view \(\V_i\) and \(\Vfrozen\) by removing all expired messages and retaining only the unexpired ones. 
    This approach streamlines the process by reducing the amount of data that validators need to store and process, enhancing the overall efficiency of the protocol.}
\end{enumerate}

\subsection{Message complexity}
\label{sec:mess-comp}

{Regarding message complexity, both Algorithm~\ref{alg:3sf-tob-noga} and Algorithm~\ref{alg:rlmd-ffg} exhibit the same expected communication complexity, namely \(\mathit{O}(Ln^3)\), where \(L\) denotes the block size. This matches the complexity of the underlying dynamically available protocols, \TOBSVD and \RLMDGHOST. This result arises because, similar to \TOBSVD and \RLMDGHOST, our protocols require honest validators to forward all received messages in every round. Consequently, each validator’s forwarding of messages they receive contributes to the \(\mathit{O}(Ln^3)\) complexity. The finality gadget layered atop the dynamically available protocol does not alter this complexity, as it is encapsulated within the \textsc{vote} message, cast at round \(4\Delta t + \Delta\). Therefore, the communication complexity of both protocols remains equivalent to that of their underlying dynamically available protocols, with the finality gadget introducing no additional complexity.}

\section{Healing process}\label{sec:healing}
So far, we have proved that, if $\GST = 0$, then both \Cref{alg:3sf-tob-noga} and \Cref{alg:rlmd-ffg} are $\eta$-dynamically-available, $\eta$-reorg-resilient, $\eta$-asynchrony-reorg-resilient and $\eta$-asynchrony-safety-resilient.
In the reminder of this section, we use the term \emph{synchronous properties} to refer to these properties.
We chose to study the synchronous properties assuming $\GST=0$ because this is the usual setting employed to analyze dynamically-available protocols.

However, the assumption that $\GST = 0$ is too strong in practice.
In fact, this would mean that as soon as we have some period of asynchrony, none of the synchronous properties is guaranteed to hold anymore\footnote{Note that the asynchrony-reorg and safety-resiliency properties only guarantee limited resiliency to periods of asynchrony, even for those lasting less than $\eta$.
Specifically, they only offer resiliency towards periods of asynchrony for events, either a chain being \textsc{propose}d by an honest validator or an honest validator confirming a chain, that happened before the period of asynchrony.}.
{This is quite undesirable in practice as} it is reasonable to expect that a protocol meant to run uninterruptedly for decades will experience periods of asynchrony at different points in time.

In this section, we show that, even if $\GST > 0$, once the network becomes synchronous and a set of additional conditions is met, then both protocols start to guarantee all the synchronous properties from that point onward.

We formalize this list of additional conditions below by letting $\theal$ be the first slot satisfying them.

\begin{definition}[Slot $\theal$]\label{def:theal}
  Let $\theal$ be the first slot such that \begin{enumerate}
    \item {$\proposing{\theal}\geq \GST + \Delta$}, 
    \item all validators that are honest in round $\fastconfirming{\theal}$ are also active in round $\fastconfirming{\theal}$ and 
    \item the proposer of slot $\theal$ is honest.
  \end{enumerate}
\end{definition}

Provided that such slot $\theal$ exists, then we have the following results.

\begin{theorem}[Reorg Resilience]\label{thm:reorg-res-prop-tob-ffg-heal}
  Both \Cref{alg:3sf-tob-noga,alg:rlmd-ffg} are $\eta$-reorg-resilient after slot $\theal$ and time $\fastconfirming{\theal}$.
\end{theorem}

\begin{theorem}[$\eta$ Dynamic Availability]
  \label{thm:dyn-avail-fast-conf-tob-ffg-heal}
Both \Cref{alg:3sf-tob-noga,alg:rlmd-ffg} are $\eta$-dynamically-available after time $\fastconfirming{\theal}$.
\end{theorem}

\begin{theorem}[Asynchrony Reorg Resilience]
  \label{thm:async-resilience-tob-ffg-heal}
  Both \Cref{alg:3sf-tob-noga,alg:rlmd-ffg} are asynchrony-reorg-resilient after slot $\theal$ and time $\fastconfirming{\theal}$.
\end{theorem}

\begin{theorem}
  \label{thm:async-safety-resilience-tob-ffg-heal}
  Both \Cref{alg:3sf-tob-noga,alg:rlmd-ffg} are asynchrony-safety-resilient after time $\fastconfirming{\theal}$.
\end{theorem}

In the reminder of this section, we provide the intuition underpinning these results.
We refer the reader to \Cref{sec:healing-detailed} for detailed proofs.

Let us start with considering the dynamically-available protocols \Cref{algo:prob-ga-fast,alg:rlmd} that the faster-finality protocols \Cref{alg:3sf-tob-noga,alg:rlmd-ffg} are based on.
In those protocols, intuitively, all that it needs to happen for the synchronous properties to start holding is that there exists a slot $t_p$ with an honest proposer such that $\merging{(t_p -1)}\geq \GST$.
This is because the only synchrony requirement that \Cref{lem:vote-proposal-fast-conf} relies on is to have synchrony starting from the merge round preceding a slot with an honest proposer.
However, under these conditions only, $\eta$ Safety after time $\theal$, as we have stated it, would not quite hold.
In fact, it would hold only if we limit it to be satisfied just for those validators that have been active in at least one round since the confirm round of slot $\theal$.
Nevertheless, with the additional requirements that all validators honest in round $\fastconfirming{\theal}$ are also active in that round, $\eta$ Safety after time $\theal$ is guaranteed by either \Cref{algo:prob-ga-fast} or \Cref{alg:rlmd}.
This is because, thanks to this additional requirement, in round $\fastconfirming{\theal}$, all validators honest in that round set their confirmed chain to the chain \textsc{propose}d in slot $\theal$.

Now, let us move to \Cref{alg:3sf-tob-noga,alg:rlmd-ffg} and consider the implications of the added logic to handle the FFG component in either protocol.
What can happen during asynchrony is that enough honest validators send \textsc{ffg-vote}s targeting a checkpoint $\C$ conflicting with the chain $\chain_p$ \textsc{propose}d by an honest validator after $\GST$  such that these votes alone do not justify such checkpoint, but they can justify it together with the \textsc{ffg-vote}s that the adversary can send via the validators that it has corrupted.
Such \textsc{ffg-vote}s are then released by the adversary after $\chain_p$ is \textsc{propose}d.
If no checkpoint higher than $\C$ has been justified yet, which can very well happen due to low honest participation, then, after such \textsc{ffg-vote}s are released, honest validators would switch to \textsc{vot}ing and confirming a chain extending $\C$ and, therefore, they would violate $\eta$ Dynamic Availability and Reorg Resilience.

However, what happens in slot $\theal$ is that all honest validators cast \textsc{ffg-vote}s with the same source checkpoint (the greatest justified checkpoint in the view of the honest proposer at the proposing time) and with target a checkpoint $\C'$, possibly different for each validator, with $\C'.\chain \preceq \chain_p$ and $\C'.c = \theal$.
Given that all those validators are also active in the confirm round of slot $\theal$ and that we assume $f<\frac{n}{3}$, this leads to all validators honest in the confirm round of slot $\theal$ to see a checkpoint $\J$, possibly different for each validator, with $\J.\chain \preceq \chain_p$ and $\J.c = \theal$ as the greatest justified checkpoint.
Like in the case of \Cref{algo:prob-ga-fast,alg:rlmd}, this also leads to all validators honest in the confirm round of slot $\theal$ to set $\chainava$ to $\chain_p$.
As a consequence of this, by the vote round of slot $\theal +1$, for any honest validator, it is impossible that a greatest justified checkpoint conflicting with chain $\chain_p$ can emerge.
Hence, all honest validators in $H_{\theal + 1}$ cast \textsc{vote} messages extending chain $\chain_p$ and \textsc{ffg-vote}s for targets that are prefixes of $\chain_p$.
Therefore, given that we assume $f<\frac{n}{3}$, no checkpoint for slot $\theal+1$ conflicting with $\chain_p$ can ever be justified.
Clearly this will keep being the case for any slot onward.
In summary, the key intuition about \Cref{alg:3sf-tob-noga} and \Cref{alg:rlmd-ffg} guaranteeing the synchronous properties after slot $\theal$ or after round $\fastconfirming{\theal}$, depending on the property, is that after the confirm round of slot $\theal$, all honest validators always see greatest justified checkpoints for a slot $\theal$ or higher and no checkpoint for slot $\theal$ or higher conflicting with the chain \textsc{propose}d in slot $\theal$ can ever be justified.

\section{Conclusion}
\label{sec:conclusion}


{
In this work, we present a novel finality gadget that, when integrated with dynamically-available consensus protocols, enables a secure ebb-and-flow protocol with just a single voting phase per proposed chain. This streamlined design enhances the practical throughput of large-scale blockchain networks. By limiting the process to a single voting phase for each chain proposal, our protocols reduce the time interval between proposals. Consequently, this optimization leads to a 20\% increase in transaction throughput compared to the SSF protocol of D'Amato and Zanolini (assuming same block size), all while maintaining essential security guarantees.



Moreover, assuming a uniformly distributed transaction submission time, our protocols achieve a shorter expected confirmation time compared to the SSF protocol of D'Amato and Zanolini. Specifically, for an adversary with power \(\beta = \frac{1}{3}\), the expected confirmation time for 3SF is improved by approximately 11\% over the SSF protocol.

A key trade-off in our approach is a slight delay in chain finalization, which extends to the time required to propose three additional chains, rather than finalizing before the next chain proposal. In practice, this translates to an expected finalization time of \(16\Delta\) (which can be reduced to $13\Delta$) when $\beta = \frac{1}{3}$, compared to \(11\Delta\) in the SSF protocol. However, this delay is offset by providing early confirmation of the proposed chain before the subsequent chain is introduced. Users can be confident that, under conditions of network synchrony and at least 2/3 honest node participation, the chain will be finalized within the time required to propose three additional chains. This feature is particularly beneficial in practical scenarios, where periods of synchrony and robust honest participation often extend well beyond the duration required for finalization in our protocol. Finally, our protocols enhance the practicality of large-scale blockchain networks by enabling the dynamically-available component to recover from extended asynchrony, provided at least two-thirds of validators remain honest and online for sufficient time. This improvement addresses a limitation in the original ebb-and-flow protocol, which assumed constant synchrony to ensure safety, highlighting the importance of resilience in real-world, asynchronous conditions.

}

\bibliography{references}

\begin{thebibliography}{10}

\bibitem{ethereum-properties}
Aditya Asgaonkar, Francesco D'Amato, Roberto Saltini, and Luca Zanolini.
\newblock Ethereum consensus property list.
\newblock URL: \url{https://docs.google.com/document/d/1Q_iDOODIq-glLRPSMnSf3HAKDQGHXsENHzRG9iwok7g}.

\bibitem{casper}
Vitalik Buterin and Virgil Griffith.
\newblock Casper the friendly finality gadget.
\newblock {\em CoRR}, abs/1710.09437, 2017.
\newblock URL: \url{http://arxiv.org/abs/1710.09437}, \href {https://arxiv.org/abs/1710.09437} {\path{arXiv:1710.09437}}.

\bibitem{gasper}
Vitalik Buterin, Diego Hernandez, Thor Kamphefner, Khiem Pham, Zhi Qiao, Danny Ryan, Juhyeok Sin, Ying Wang, and Yan~X Zhang.
\newblock Combining {GHOST} and {Casper}.
\newblock 2020.

\bibitem{DBLP:conf/osdi/CastroL99}
Miguel Castro and Barbara Liskov.
\newblock Practical byzantine fault tolerance.
\newblock In Margo~I. Seltzer and Paul~J. Leach, editors, {\em Proceedings of the Third {USENIX} Symposium on Operating Systems Design and Implementation (OSDI), New Orleans, Louisiana, USA, February 22-25, 1999}, pages 173--186. {USENIX} Association, 1999.

\bibitem{DBLP:conf/sp/DaianGKLZBBJ20}
Philip Daian, Steven Goldfeder, Tyler Kell, Yunqi Li, Xueyuan Zhao, Iddo Bentov, Lorenz Breidenbach, and Ari Juels.
\newblock Flash boys 2.0: Frontrunning in decentralized exchanges, miner extractable value, and consensus instability.
\newblock In {\em 2020 {IEEE} Symposium on Security and Privacy, {SP} 2020, San Francisco, CA, USA, May 18-21, 2020}, pages 910--927. {IEEE}, 2020.
\newblock \href {https://doi.org/10.1109/SP40000.2020.00040} {\path{doi:10.1109/SP40000.2020.00040}}.

\bibitem{DBLP:conf/podc/DAmatoLZ24}
Francesco D'Amato, Giuliano Losa, and Luca Zanolini.
\newblock Asynchrony-resilient sleepy total-order broadcast protocols.
\newblock In Ran Gelles, Dennis Olivetti, and Petr Kuznetsov, editors, {\em Proceedings of the 43rd {ACM} Symposium on Principles of Distributed Computing, {PODC} 2024, Nantes, France, June 17-21, 2024}, pages 247--256. {ACM}, 2024.
\newblock \href {https://doi.org/10.1145/3662158.3662779} {\path{doi:10.1145/3662158.3662779}}.

\bibitem{goldfish}
Francesco D'Amato, Joachim Neu, Ertem~Nusret Tas, and David Tse.
\newblock No more attacks on proof-of-stake ethereum?
\newblock {\em CoRR}, abs/2209.03255, 2022.
\newblock URL: \url{https://doi.org/10.48550/arXiv.2209.03255}.

\bibitem{streamliningSBFT}
Francesco D'Amato, Roberto Saltini, Thanh-Hai Tran, and Luca Zanolini.
\newblock {TOB-SVD}: {T}otal-{O}rder {B}roadcast with {S}ingle-{V}ote {D}ecisions in the {S}leepy {M}odel, 2024.
\newblock URL: \url{https://arxiv.org/abs/2310.11331}, \href {https://arxiv.org/abs/2310.11331} {\path{arXiv:2310.11331}}.

\bibitem{DBLP:conf/esorics/DAmatoZ23}
Francesco D'Amato and Luca Zanolini.
\newblock A simple single slot finality protocol for ethereum.
\newblock In Sokratis~K. Katsikas, Fr{\'{e}}d{\'{e}}ric Cuppens, Nora Cuppens{-}Boulahia, Costas Lambrinoudakis, Joaqu{\'{\i}}n Garc{\'{\i}}a{-}Alfaro, Guillermo Navarro{-}Arribas, Pantaleone Nespoli, Christos Kalloniatis, John Mylopoulos, Annie~I. Ant{\'{o}}n, and Stefanos Gritzalis, editors, {\em Computer Security. {ESORICS} 2023 International Workshops - CyberICS, DPM, CBT, and SECPRE, The Hague, The Netherlands, September 25-29, 2023, Revised Selected Papers, Part {I}}, volume 14398 of {\em Lecture Notes in Computer Science}, pages 376--393. Springer, 2023.
\newblock \href {https://doi.org/10.1007/978-3-031-54204-6\_23} {\path{doi:10.1007/978-3-031-54204-6\_23}}.

\bibitem{rlmd}
Francesco D'Amato and Luca Zanolini.
\newblock Recent latest message driven {GHOST:} balancing dynamic availability with asynchrony resilience.
\newblock In {\em 37th {IEEE} Computer Security Foundations Symposium, {CSF} 2024, Enschede, Netherlands, July 8-12, 2024}, pages 127--142. {IEEE}, 2024.
\newblock \href {https://doi.org/10.1109/CSF61375.2024.00001} {\path{doi:10.1109/CSF61375.2024.00001}}.

\bibitem{DBLP:conf/wdag/GafniL23}
Eli Gafni and Giuliano Losa.
\newblock Brief announcement: Byzantine consensus under dynamic participation with a well-behaved majority.
\newblock In Rotem Oshman, editor, {\em 37th International Symposium on Distributed Computing, {DISC} 2023, October 10-12, 2023, L'Aquila, Italy}, volume 281 of {\em LIPIcs}, pages 41:1--41:7. Schloss Dagstuhl - Leibniz-Zentrum f{\"{u}}r Informatik, 2023.
\newblock URL: \url{https://doi.org/10.4230/LIPIcs.DISC.2023.41}, \href {https://doi.org/10.4230/LIPICS.DISC.2023.41} {\path{doi:10.4230/LIPICS.DISC.2023.41}}.

\bibitem{DBLP:journals/corr/abs-2304-14701}
Andrew Lewis{-}Pye and Tim Roughgarden.
\newblock Permissionless consensus.
\newblock {\em CoRR}, abs/2304.14701, 2023.
\newblock URL: \url{https://doi.org/10.48550/arXiv.2304.14701}, \href {https://arxiv.org/abs/2304.14701} {\path{arXiv:2304.14701}}, \href {https://doi.org/10.48550/ARXIV.2304.14701} {\path{doi:10.48550/ARXIV.2304.14701}}.

\bibitem{DBLP:journals/iacr/MalkhiMR22}
Dahlia Malkhi, Atsuki Momose, and Ling Ren.
\newblock Byzantine {Consensus} under {Fully} {Fluctuating} {Participation}.
\newblock 2022.
\newblock URL: \url{https://eprint.iacr.org/archive/2022/1448/20221024:011919}.

\bibitem{DBLP:conf/ccs/MalkhiM023}
Dahlia Malkhi, Atsuki Momose, and Ling Ren.
\newblock Towards practical sleepy {BFT}.
\newblock In Weizhi Meng, Christian~Damsgaard Jensen, Cas Cremers, and Engin Kirda, editors, {\em Proceedings of the 2023 {ACM} {SIGSAC} Conference on Computer and Communications Security, {CCS} 2023, Copenhagen, Denmark, November 26-30, 2023}, pages 490--503. {ACM}, 2023.
\newblock \href {https://doi.org/10.1145/3576915.3623073} {\path{doi:10.1145/3576915.3623073}}.

\bibitem{milionis2024automatedmarketmakinglossversusrebalancing}
Jason Milionis, Ciamac~C. Moallemi, Tim Roughgarden, and Anthony~Lee Zhang.
\newblock Automated market making and loss-versus-rebalancing, 2024.
\newblock URL: \url{https://arxiv.org/abs/2208.06046}, \href {https://arxiv.org/abs/2208.06046} {\path{arXiv:2208.06046}}.

\bibitem{DBLP:conf/ccs/Momose022}
Atsuki Momose and Ling Ren.
\newblock Constant latency in sleepy consensus.
\newblock In Heng Yin, Angelos Stavrou, Cas Cremers, and Elaine Shi, editors, {\em Proceedings of the 2022 {ACM} {SIGSAC} Conference on Computer and Communications Security, {CCS} 2022, Los Angeles, CA, USA, November 7-11, 2022}, pages 2295--2308. {ACM}, 2022.
\newblock \href {https://doi.org/10.1145/3548606.3559347} {\path{doi:10.1145/3548606.3559347}}.

\bibitem{DBLP:conf/sp/NeuTT21}
Joachim Neu, Ertem~Nusret Tas, and David Tse.
\newblock Ebb-and-flow protocols: {A} resolution of the availability-finality dilemma.
\newblock In {\em 42nd {IEEE} Symposium on Security and Privacy, {SP} 2021, San Francisco, CA, USA, 24-27 May 2021}, pages 446--465. {IEEE}, 2021.
\newblock \href {https://doi.org/10.1109/SP40001.2021.00045} {\path{doi:10.1109/SP40001.2021.00045}}.

\bibitem{ebbandflow}
Joachim Neu, Ertem~Nusret Tas, and David Tse.
\newblock Ebb-and-flow protocols: {A} resolution of the availability-finality dilemma.
\newblock In {\em 42nd {IEEE} Symposium on Security and Privacy}, 2021.
\newblock Forthcoming.
\newblock URL: \url{https://arxiv.org/abs/2009.04987}.

\bibitem{sleepy}
Rafael Pass and Elaine Shi.
\newblock The sleepy model of consensus.
\newblock In {\em {ASIACRYPT} {(2)}}, volume 10625 of {\em Lecture Notes in Computer Science}, pages 380--409. Springer, 2017.

\bibitem{DBLP:conf/fc/Schwarz-Schilling22}
Caspar Schwarz{-}Schilling, Joachim Neu, Barnab{\'{e}} Monnot, Aditya Asgaonkar, Ertem~Nusret Tas, and David Tse.
\newblock Three attacks on proof-of-stake ethereum.
\newblock In Ittay Eyal and Juan~A. Garay, editors, {\em Financial Cryptography and Data Security - 26th International Conference, {FC} 2022, Grenada, May 2-6, 2022, Revised Selected Papers}, volume 13411 of {\em Lecture Notes in Computer Science}, pages 560--576. Springer, 2022.

\bibitem{ghost}
Yonatan Sompolinsky and Aviv Zohar.
\newblock Secure high-rate transaction processing in {Bitcoin}.
\newblock In {\em International Conference on Financial Cryptography and Data Security}, pages 507--527. Springer, 2015.

\bibitem{zamfir}
Vlad Zamfir.
\newblock Casper the friendly ghost. a correct-by-construction blockchain consensus protocol.
\newblock URL: \url{https://github.com/ethereum/research/blob/master/papers/cbc-consensus/AbstractCBC.pdf}.

\end{thebibliography}
\bibliographystyle{plainurl}

\appendix

\section{Healing process - Detailed Proofs}\label{sec:healing-detailed}
In this section, we provide detailed proofs of the results presented in \Cref{sec:healing}.
In the interest of space and time, we will only show the proofs for \Cref{alg:3sf-tob-noga} as the proofs for \Cref{alg:rlmd-ffg} are very similar.
The reader should be convinced of this by considering the similarities of the proofs of the various Lemmas and Theorems for \Cref{alg:rlmd-ffg} to the proof of the analogous Lemmas and Theorems for \Cref{alg:3sf-tob-noga}.



First, we prove that all the synchronous properties hold for \Cref{algo:prob-ga-fast}, the dynamically-available protocol that \Cref{alg:3sf-tob-noga} is based on, from slot $\theal$ onward.
Then, in line with the strategy followed so far, we show that given an execution of \Cref{alg:3sf-tob-noga} there exists an execution of \Cref{algo:prob-ga-fast} that is a dynamically-equivalent to the execution of \Cref{alg:3sf-tob-noga} from $\theal$ onward.
This will then enable us to then show that the synchronous properties hold for \Cref{alg:3sf-tob-noga} from slot $\theal$ onward.
\subsection{Analysis for \Cref{algo:prob-ga-fast}}

\begin{lemma}\label{lem:keep-voting-tob-fast-conf-heal}
  If $\merging{t} \geq \GST$, then \Cref{lem:keep-voting-tob-fast-conf} holds for $\GST > 0$ as well.
\end{lemma}
\begin{proof}
  The only synchrony requirement that the proof of \Cref{lem:keep-voting-tob-fast-conf} relies on is that \textsc{vote} messages sent by validators in $H_{\voting{t}}$ are received by round {$\merging{(t)}$} by all honest validator awake at that round.
\end{proof}

\begin{lemma}\label{lem:one-fast-confirm-all-vote-fast-conf-heal}
  If $\merging{t} \geq \GST$, then \Cref{lem:one-fast-confirm-all-vote-fast-conf} holds for $\GST > 0$ as well.
\end{lemma}
\begin{proof}
  Given \Cref{lem:keep-voting-tob-fast-conf-heal}, the only synchrony requirement that the proof \Cref{lem:one-fast-confirm-all-vote-fast-conf} relies on is that synchrony holds since $\merging{t}$.
\end{proof}

\begin{lemma}\label{lem:vote-proposal-fast-conf-heal}
  If $\merging{t-1} \geq \GST$, then \Cref{lem:vote-proposal-fast-conf} holds for $\GST > 0$ as well.
\end{lemma}
\begin{proof}
  Given \Cref{lem:keep-voting-tob-fast-conf-heal}, the only synchrony requirement that the proof of \Cref{lem:vote-proposal-fast-conf} relies on is that synchrony holds since $\merging{(t-1)} = 4 \Delta t - \Delta\geq \GST$.
\end{proof}

\begin{lemma}\label{lem:ga-confirmed-always-canonical-heal}
  If $r_i \geq \fastconfirming{\theal}$, then \Cref{lem:ga-confirmed-always-canonical} holds for $\GST > 0$ as well.
\end{lemma}
\begin{proof}
  The proof of this Lemma is very similar to the proof of \Cref{lem:ga-confirmed-always-canonical}.
  However, the changes are intertwined enough with the proof that, for the sake of clarity, we are better off writing the entire proof.

  Like in the proof of \Cref{lem:ga-confirmed-always-canonical}, we proceed by contradiction.
  Let $r_i \geq \fastconfirming{\theal}$ be the smallest round such that there exist two honest validators $v_i$ and $v_j$, and round $r_j$ such that $r_j\ge r_i$ and $r_j \in \{\proposing{\slot(r_j)}, \voting{\slot(r_j)}\}$ and 
  $\Chain^{r_i}_i \preceq \mfc^{r_j}_j$.
  Due to the joining protocol, $H_{\theal}$ includes all the validators honest in round $\voting{\theal}$.
  Given this and the minimality of $r_i$, $v_i \in H_{r_i}$ and $\Chain^{r_i-1}_i \neq \Chain^{r_i}_i$.
  This can only happen if $r_i$ is either a voting or a fast confirmation round.
  Let $t_i= \mathrm{slot}(r_i)$ and proceed by analyzing each case separately. 
  \begin{description}
    \item[Case 1: $r_i$ is a vote round.] 
    Due to \Cref{line:algga-no-ffg-vote-chainava}, $\Chain^{r_i}_i \succeq \left(\mfcvote{t_i}_i\right)^{\lceil \kappa} \lor \Chain^{r_i}_i$.
    Let us now consider two sub cases.
      \begin{description}
        \item[Case 1.1: $\Chain^{r_i}_i = \left(\mfcvote{t_i}_i\right)^{\lceil \kappa}$.] 
        We know that with  overwhelming probability (Lemma 2~\cite{rlmd}), there exists at least one slot \(t_p\) in the interval \([\max(t_i - \kappa,\theal), t_i)\) with an honest proposer $v_p$.
        Note that this interval is not empty as, due to the Lemma's conditions, this case implies $\slot(r_i) > \theal$.
        Let $\chain_p$ be the chain \textsc{propose}d by $v_p$ in slot $t_p$.
        Given that $\theal \leq t_p < t_i$, \Cref{lem:vote-proposal-fast-conf-heal} implies that 
        $\mfcvote{t_i}_i \succeq \chain_p$.
        Then, because $t_p\geq t_i-\kappa$, we have that $\chain_p \succeq  \left(\mfcvote{t_i}_i\right)^{\lceil \kappa} = \Chain^{r_i}_i$.
        Because $\theal \leq t_p < \slot(r_j)$, \Cref{lem:vote-proposal-fast-conf-heal} also implies that $\chain^{r_j}_j \succeq \chain_p \succeq \Chain^{r_i}_i$ leading to a contradiction.
        \item[Case 1.2: $\Chain^{r_i}_i \succ \left(\mfcvote{t_i}_i\right)^{\lceil \kappa}$.] 
        This case implies that $\Chain^{r_i}_i  = \Chain^{r_i-1}_i$.
        From the minimality of $r_i$ we reach a contradiction.
      \end{description}
      \item[Case 2: $r_i$ is a fast confirmation round.]
      Note that this implies that $t_i < \slot(r_j)$.
      Let us consider three sub cases.
      \begin{description}
        \item[Case 2.1: $r_i = \fastconfirming{\theal}$.]
        Given that the the proposer $v_p$ of $\theal$ is honest, we can apply \Cref{lem:vote-proposal-fast-conf-heal} to conclude that all validators in $H_\theal$ cast a \textsc{vote}  message for the chain $\chain$ \textsc{propose}d by $v_p$.
        Then, because we assume that $f<\frac{n}{3}$, $H_{\theal}$ includes all the validators honest in round $\voting{\theal}$ and any validator honest in round  $\fastconfirming{\theal}$ is also awake in round $\fastconfirming{\theal}$, then, for any validator $v_i$ honest in round $\fastconfirming{\theal}$, $(\chain_p, Q) = \texttt{fastconfirmsimple}(\V^{\fastconfirming{\theal}}_i,t_i) \land Q \neq \emptyset$ meaning that $v_i$ sets $\Chain^{r_i}_i = \chain_p$.
        Then, we can apply \Cref{lem:one-fast-confirm-all-vote-fast-conf-heal} to conclude that $\chain^{r_j}_j \succeq \Chain^{r_i}_i$ reaching a contradiction.
        \item[Case 2.2: $r_i > \fastconfirming{\theal} \land \Chain^{r_i-1}_i \neq \Chain^{r_i}_i$.] 
        Due to \Crefrange{line:algga-no-ffg-on-confirm}{line:algga-no-ffg-set-chaava-to-bcand}, this case implies $(\Chain^{r_i}_i, Q) = \texttt{fastconfirmsimple}(\V^{r_i}_i,t_i) \land Q \neq \emptyset$.
        Therefore, we can apply \Cref{lem:one-fast-confirm-all-vote-fast-conf-heal} to conclude that $\chain^{r_j}_j \succeq \Chain^{r_i}_i$ reaching a contradiction.
        \item[Case 2.3: $r_i > \fastconfirming{\theal} \land \Chain^{r_i-1}_i = \Chain^{r_i}_i$.] 
        This case implies that $r_i - 1 \geq \fastconfirming{\theal}$.
        Hence, due to the minimality of $r_i$ we reach a contradiction.\qedhere
      \end{description}
  \end{description}  
\end{proof}

\begin{theorem}[Reorg Resilience]\label{thm:reorg-res-prop-tob-heal}
  \Cref{algo:prob-ga-fast} is $\eta$-reorg-resilient after slot $\theal$ and time $\fastconfirming{\theal}$.
\end{theorem}
\begin{proof}
  Follows from \Cref{lem:vote-proposal-fast-conf-heal,lem:ga-confirmed-always-canonical-heal} and the proof of \Cref{thm:reorg-res-prop-tob}.
\end{proof}

\begin{theorem}
  \label{thm:dyn-avail-fast-conf-tob-heal}
\Cref{algo:prob-ga-fast} is $\eta$-dynamically-available after time $\fastconfirming{\theal}$.
\end{theorem}
\begin{proof}
  Follows from \Cref{lem:vote-proposal-fast-conf-heal,lem:ga-confirmed-always-canonical-heal} and the proof of \Cref{thm:dyn-avail-fast-conf-tob}.
\end{proof}

\begin{lemma}[Liveness of fast confirmations]
  \label{thm:fast-liveness-tob-heal}
  If $\merging{t-1} \geq \GST$, then \Cref{thm:fast-liveness-tob} holds for $\GST > 0$ as well.
\end{lemma}
\begin{proof}
  From \Cref{lem:vote-proposal-fast-conf-heal} and the proof of \Cref{thm:fast-liveness-tob}.
\end{proof}

\begin{lemma}\label{lem:asyn-induction-heal}
  If $\voting{t_a} \geq \GST$, then \Cref{lem:asyn-induction} holds for $\GST > 0$ as well.
\end{lemma}
\begin{proof}
  The only synchrony requirement that the proof of \Cref{lem:asyn-induction} relies on is that synchrony holds since $\voting{t_a}$.
\end{proof}

\begin{lemma}\label{lem:asyn-induction2-heal}
  If $\voting{t} \geq \GST$, then \Cref{lem:asyn-induction2} holds for $\GST > 0$ as well.
\end{lemma}
\begin{proof}
  Given \Cref{lem:keep-voting-tob-fast-conf-heal,lem:asyn-induction-heal} the only synchrony requirements that \Cref{lem:asyn-induction3} relies on is that synchrony holds from {$\min(\voting{t_a},\merging{t})$} which, given that we assume $t_a \geq t$, is implied by $\voting{t} \geq \GST$.
\end{proof}

\begin{lemma}\label{lem:asyn-induction3-heal}
  If $\voting{t} \geq \GST$, then \Cref{lem:asyn-induction3} holds for $\GST > 0$ as well.
\end{lemma}
\begin{proof}
  Given \Cref{lem:asyn-induction2-heal}, follow the proof of \Cref{lem:asyn-induction3}.
\end{proof}

\begin{theorem}[Asynchrony Reorg Resilience]
  \label{thm:async-resilience-tob-heal}
  \Cref{algo:prob-ga-fast} is $\eta$-asynchrony-reorg-resilient after slot $\theal$ and time $\fastconfirming{\theal}$.
\end{theorem}
\begin{proof}
  From \Cref{thm:reorg-res-prop-tob-heal}, \Cref{lem:vote-proposal-fast-conf-heal,lem:asyn-induction3-heal}, and the proof of \Cref{thm:async-resilience-tob}.
\end{proof}

\begin{theorem}[Asynchrony Safety Resilience]
  \label{thm:async-safety-resilience-tob-heal}
  \Cref{algo:prob-ga-fast} is $\eta$-asynchrony-safety-resilient after time $\fastconfirming{\theal}$.
\end{theorem}
\begin{proof}
  From \Cref{thm:dyn-avail-fast-conf-tob-heal}, \Cref{lem:ga-confirmed-always-canonical-heal,lem:asyn-induction3-heal}, and the proof of \Cref{thm:async-safety-resilience-tob}.
\end{proof}

\subsection{Analysis for \Cref{alg:3sf-tob-noga}}

Before proceeding, we need to tweak the definition of dynamical-equivalence to capture the fact that we interested in it from round $\voting{\theal}$ only.

\begin{definition}[dynamically-equivalent from round $\voting{\theal}$ up to round $r$]
  We say that two sets of messages are dynamically-equivalent from round $\voting{\theal}$ up to round $r$ if and only if, after removing from both sets all the \textsc{propose} messages with slot field lower or equal to $\theal$ and all the \textsc{vote} message with field lower than $\theal$, the resulting sets are dynamically-equivalent up to round $r$.
  The definition of two executions being honest-output-dynamically-equivalent from round $\voting{\theal}$ up to round $r$ follow naturally from this.
\end{definition}

Consequently, we also need to tweak the definition of $\mfc$-equivalent executions accordingly.

\begin{definition}\label{def:exec-equiv-teal}
  Let $\FFGExec$ by an $\eta$-compliant execution of \Cref{alg:3sf-tob-noga} and $\NoFFGExec$ be an $\eta$-compliant execution of \Cref{algo:prob-ga-fast}.
  We say that $\FFGExec$ and $\NoFFGExec$ are \emph{\theal-\mfc-equivalent} if and only if the following constraints hold:
  \begin{enumerate}
    \item\label[condition]{cond:0b-3-t} $\FFGExec$ and $\NoFFGExec$ are honest-output-dynamically-equivalent from round $\voting{\theal}$ up to any round
    \item for any slot $t_i \geq \theal$,
    \begin{enumerate}[label*=\arabic*]
      \item \label[condition]{cond:3-3} $\mfcproposeNoFFG{t_i}_i = \mfcproposeFFG{t_i}_i$.
      \item\label[condition]{cond:6-3} for any slot $t_j \in [\theal, t_i]$ and validator $v_j \in H_{\voting{(t_j)}}$, there exists round $r_k \in [\fastconfirming{\theal},\fastconfirming{t_j}]$ and a validator $v_k$ honest in round $r_k$ such that $\Chain^{r_k}_k \succeq \chainava^{\fastconfirming{t_j}}_j$.
    \end{enumerate}
    \item for any slot $t_i > \theal$ and validator $v_i \in H_{\voting{(t_i)}}$,
    \begin{enumerate}[label*=\arabic*]
      \item\label[condition]{cond:4-3} $\mfcvoteNoFFG{t_i}_i = \mfcvoteFFG{t_i}_i$
      \item\label[condition]{cond:5-3} for any slot $t_j \in [\theal +1, t_i]$ and validator $v_j \in H_{\voting{(t_j)}}$, there exists round $r_k \in [\fastconfirming{\theal},\voting{t_j}]$ and a validator $v_k $ honest in round $r_k$ such that $\Chain^{r_k}_k \succeq \chainava^{\voting{t_j}}_j$.
    \end{enumerate}
  \end{enumerate}
\end{definition}

\begin{lemma}\label{lem:equiv-ga3}
  Let $\FFGExec$ by any $\eta$-compliant execution of \Cref{alg:3sf-tob-noga}.  
  There exists an $\eta$-compliant execution $\NoFFGExec$ of \Cref{algo:prob-ga-fast} that is $\theal$-\mfc-equivalent to $\FFGExec$.
\end{lemma}

\begin{proof}
  Let $\FFGExec$ by any $\eta$-compliant execution of \Cref{alg:3sf-tob-noga} and let $\AFFG$ be the adversary in such an execution.
  Let $v_p$ be the honest proposer in slot $\theal$ and let $[\textsc{propose}, \chain_p,\cdot,\cdot,\GJ_p,\cdot,\cdot]$ the \textsc{propose} message sent by $v_p$ in slot $\theal$ in execution $\FFGExec$..
  
  Let the adversary $\ANoFFG$ behave as follows.
  \begin{enumerate}[label=(\roman*)]
    \item As per item \theenumi{} in the proof of \Cref{lem:equiv-ga2}.
    \item As per item \theenumi{} in the proof of \Cref{lem:equiv-ga2}.
    \item As per item \theenumi{} in the proof of \Cref{lem:equiv-ga2}.
    \item\label{cond:adv-decision-5}  Up to round $\merging{(\theal-1)}$, 
    \begin{enumerate}[label=(\roman{enumi}.\roman*)]
      \item For any slot $t_r < \theal$, if there exists a chain $\chain_r \preceq \mfcproposeFFG{\theal}_p$ such that $\chain_r.p = t_r$ and the proposer of $t_r$ is adversarial in round $\proposing{t_r}$, then $v_r$ \textsc{propose}s $\chain_r$ in slot $\chain_r.p$.
      \item Other than the above, the adversary does not send any message.
      \item All messages are delivered within $\Delta$ rounds.
    \end{enumerate}
    \item\label{cond:adv-decision-4-heal}  From round $\merging{(\theal-1)}$ onward, $\ANoFFG$ behaves as per \cref{cond:adv-decision-4} in the proof of \Cref{lem:equiv-ga2}.
  \end{enumerate}

  Now, we move to proving by induction on $t_i$ that the execution $\NoFFGExec$ induced by $\ANoFFG$ satisfies all \Cref{def:exec-equiv-teal}'s conditions.
  To do so, we add the following conditions to the inductively hypothesis,

  \begin{enumerate}[start=2]
    \item for any slot $t_i \geq \theal$,%
    \begin{enumerate}[start=3,label*=\arabic*]
      \item\label[condition]{cond:7-3} $\GJ(\VFFG^{\fastconfirming{t_i}}_i).c \geq \theal$
    \end{enumerate}
  \end{enumerate}

  \begin{enumerate}[start=3]
    \item for any slot $t_i > \theal$ and validator $v_i \in H_{\voting{(t_i)}}$,%
    \begin{enumerate}[start=4,label*=\arabic*]
      \item \label[condition]{cond:1-3} For any $\J$ such that $\mathsf{J}(\J,\VFFG^{\proposing{t_i}}_i) \land \J.c \geq \theal$, $\mfcproposeNoFFG{t_i}_i \succeq \J.\chain$.
      \item\label[condition]{cond:2-3} For any $\J$ such that $\mathsf{J}(\J,\VFFG^{\voting{t_i}}_i) \land \J.c \geq \theal$, $\mfcvoteNoFFG{t_i}_i \succeq \J.\chain$.
      \item\label[condition]{cond:8-3} $\GJfrozen[t_i]_i.c \geq \theal$.
    \end{enumerate}
  \end{enumerate} 
  and rephrase \Cref{cond:0b-3-t} as follows
  \begin{enumerate}[start=2]
    \item for any slot $t_i \geq \theal$,%
    \begin{enumerate}[start=4,label*=\arabic*]
      \item\label[condition]{cond:0bt-3}  $\FFGExec$ and $\NoFFGExec$ are honest-output-dynamically-equivalent from round $\voting{\theal}$ up to round $4\Delta (t_i+1)$
    \end{enumerate}
  \end{enumerate}
  \begin{description}
    \item[Base Case: $t_i = \theal$.]
    The only non-vacuously true conditions for this case are \Cref{cond:3-3,cond:0bt-3,cond:7-3,cond:6-3}.
    Let us now prove that each holds.
    \begin{description}
      \item[\Cref{cond:3-3}.]
      It should be quite easy to see that, due to item~\ref{cond:adv-decision-5} of $\ANoFFG$'s set of decisions, execution $\NoFFGExec$ of \Cref{algo:prob-ga-fast} is such that $\mfcproposeNoFFG{\theal}_p = \mfcproposeFFG{\theal}_p$.
      \item[\Cref{cond:0bt-3}.]
      Given the above, in either execution $v_p$ sends a \textsc{propose} message for the same chain $\chain_p$.     
      Then, \Cref{lem:vote-proposal-fast-conf-ffg,lem:vote-proposal-fast-conf-ffg-only-base} show that all honest validators in round $\voting{\theal}$ send the same \textsc{vote} messages for chain $\chain_p$.
      Given that no other messages are sent by honest validators before round $\proposing{(\theal+1)}$, this proves \Cref{cond:0bt-3} for slot $\theal$.

      \item[\Cref{cond:7-3}.]
      \Cref{lem:vote-proposal-fast-conf-ffg-only-base} and \Cref{line:algtob-vote-chainava,line:algtob-set-target-checkpoint,line:algtob-vote} and the definition of $\theal$ imply that al validators honest in round $\voting{\theal}$ send \textsc{ffg-vote}s that have source $\GJ_p$ and target checkpoint $\T$ such that $\T.c = t$ and $\T.\chain \preceq \chain_p$.
      All such votes are in the view of any validator honest in round $\fastconfirming{\theal}$.
      Given that we assume $f<\frac{n}{3}$, this implies that, for any validator $v_i$ honest in round $\fastconfirming{\theal}$, $\GJ(\VFFG^{\fastconfirming{t_i}}_i).\chain \preceq \chain_p$ and $\GJ(\VFFG^{\fastconfirming{t_i}}_i).c = \theal$ proving \Cref{cond:7-3}.

      \item[\Cref{cond:6-3}.]
      Then, given that in execution $\FFGExec$, for any validator $v_i$ honest in round $\voting{\theal}$, $v_i$ cast \textsc{vote} messages for $\chain_p$, $\chain_p \succeq \GJ(\VFFG^{\fastconfirming{t_i}}_i).\chain$, then, for any validator $v_j$ honest in round $\fastconfirming{\theal}$, $\chainava^{\fastconfirming{\theal}}_j = \chain_p$.
      Similarly, for any validator $v_j$ honest in round $\fastconfirming{\theal}$ of execution $\NoFFGExec$, $\Chain^{\fastconfirming{\theal}}_j = \chain_p$.
      Hence, \Cref{cond:6-3} is proved.
    \end{description}
    \item[Inductive Step: $t_i > \theal$.]
    Assume that the Lemma and the additional conditions \Cref{cond:1-3,cond:2-3,cond:0bt-3,cond:8-3,cond:7-3} hold for slot $t_i-1$.
    Below we prove that they also hold for slot $t_i$.
    \begin{description}
      \item[\Cref{cond:1-3,cond:2-3}.]  Let $\J$ be any checkpoint such that  $\mathsf{J}(\J,\VFFG^{\voting{t_i}}_i) \land \J.c \geq \theal$.
      Because the chain of the target checkpoint of an \textsc{ffg-vote} cast by an honest validator $v_\ell$ in round $r_{\ell}$ corresponds to $(\makeFFG{\chainava}^{r_{\ell}}_\ell.\chain,\slot(r_\ell))$ and we assume $f<\frac{n}{3}$, we have that
      $\J.c \in [\theal,t_i)$ and there exists a validator $v_k \in H_{\J.c}$ such that $\chainava^{\voting{\J.c}}_k \succeq \J.\chain$.
      Let us consider two cases.
      \begin{description}
        \item[Case 1: $\J.c = \theal$.]
        \Cref{lem:vote-proposal-fast-conf-ffg,lem:vote-proposal-fast-conf-ffg-only-base} show that all honest validators in round $\voting{\theal}$ send the same \textsc{vote} messages for chain $\chain_p$.
        Then, because we assume that $f<\frac{n}{3}$,
        due to \Crefrange{line:algtob-vote-chainava}{line:algtob-vote-comm}, in execution $\FFGExec$,
        we have that $\chain_p \succeq \J.\chain$.
        Given that $\theal < t_i$, \Cref{lem:vote-proposal-fast-conf-heal} also implies that $\mfcproposeNoFFG{t_i}_i \succeq \chain_p$ and $\mfcvoteNoFFG{t_i}_i \succeq \chain_p$.
        Therefore, $\mfcproposeNoFFG{t_i}_i \succeq \chain_p \succeq \J.\chain$ and $\mfcvoteNoFFG{t_i}_i \succeq \chain_p \succeq \J.\chain$.
        \item[Case 2: $\J.c > \theal$.]  
        By the inductive hypothesis, \Cref{cond:5-3} implies that there exists a slot $t_m \in [\theal+1,\J.c]$ and validator $v_m \in H_{\voting{(t_m)}}$ such that $\Chain^{\voting{t_m}}_m \succeq \chainava^{\voting{\J.c}}_k$.
        Given that $t_m \leq \J.c < t_i$, from \Cref{lem:ga-confirmed-always-canonical}, we know that $\mfcproposeNoFFG{t_i} \succeq \Chain^{\voting{t_m}}_m \succeq \chainava^{\voting{\J.c}}_k \succeq \J.\chain$ and that $\mfcvoteNoFFG{t_i} \succeq \Chain^{\voting{t_m}}_m \succeq \chainava^{\voting{\J.c}}_k \succeq \J.\chain$.
      \end{description}
      Given that the set of justified checkpoints in $\VFFG^{\voting{t_i}}_i$ is a superset of the set of justified checkpoints in $\VFFG^{\proposing{t_i}}_i$, both \Cref{cond:1-3} and \Cref{cond:2-3} are proven.
      \item[\Cref{cond:3-3,cond:4-3}.]
      We know that $\FFGExec$ and $\NoFFGExec$ are honest-output-dynamically-equivalent from round $\voting{\theal}$ up to round $\merging{t_i}$.
      Hence, due to \cref{cond:adv-decision-4-heal} of $\ANoFFG$'s set of decisions, for any $v_i \in H_{\voting{t}}$, $\VFFG^{\proposing{t_i}}_i$ and $\VNoFFG^{\proposing{t_i}}_i$ are dynamically-equivalent from round $\voting{\theal}$ up to round $\merging{t_i}$.
      By \Cref{lem:vote-proposal-fast-conf-heal}, we know that, in execution $\NoFFGExec$, any \textsc{vote} cast by a validator honest in a round in $[\voting{\theal},\voting{(t_i-1)}]$ is for a chain extending $\chain_p$.
      Then, given that $\FFGExec$ and $\NoFFGExec$ are honest-output-dynamically-equivalent from round $\voting{\theal}$ up to round $\merging{t_i}$, this is true for execution $\FFGExec$ as well.
      From the fact that $\FFGExec$ and $\NoFFGExec$ are honest-output-dynamically-equivalent from round $\proposing{\theal}$ up to round $\merging{t_i}$, we also know that, for any validator $v_i$ honest in round $\proposing{t_i}$, any \textsc{vote} message $[\textsc{vote}, \chain_v,\cdot,t_v,\cdot]$ that is in in one of the two views $\VNoFFG^{\proposing{t_i}}_i$ and $\VFFG^{\proposing{t_i}}_i$, for which there does not exists a dynamically-equivalent \textsc{vote} message in the other view, then $t_v < \theal$.
      Hence, given that in both executions, any \textsc{vote} cast by a validator honest in a round in $[\voting{\theal},\voting{(t_i-1)}]$ is for a chain extending $\chain_p$, this implies that $\chain_v \nsucceq \chain_p$.
      Hence, in $\VFFG^{\proposing{t_i}}_i$, for any \textsc{vote} message $[\textsc{vote}, \chain',\cdot,\cdot,\cdot]$ with $\chain' \succeq \chain_p$, there exists a dynamically-equivalent message in $\VNoFFG^{\proposing{t_i}}_i$, and vice-versa.
      From this and \Cref{cond:1-3} follows that \Cref{cond:3-3} holds.

      Therefore, $\VFFG^{\voting{t_i}}_i$ and $\VNoFFG^{\voting{t_i}}_i$ are $\theal$-$\mfc$-equivalent from round $\voting{\theal}$ up to round $\voting{t_i}$.
      Hence, we can apply the same reasoning outlined above and, due to \Cref{cond:2-3}, conclude that \Cref{cond:4-3} holds as well.
      \item[\Cref{cond:8-3}.]
      From the inductive hypothesis, by \Cref{cond:7-3}, $\GJ(\VFFG^{\fastconfirming{(t_i-1)}}_i).c \geq \theal$. From \Cref{line:algtob-prop-if,line:algtob-prop-merge-gj} then we can conclude that $\GJfrozen[\voting{t_i}]_i.c \geq \theal$.
      \item[\Cref{cond:7-3}.]
      From the inductive hypothesis, by \Cref{cond:7-3}, $\GJ(\VFFG^{\fastconfirming{(t_i-1)}}_i).c \geq \theal$.
      Given that $\VFFG^{\fastconfirming{(t_i-1)}}_i \subseteq \VFFG^{\fastconfirming{t_i}}_i$, clearly  $\GJ(\VFFG^{\fastconfirming{t_i}}_i).c \geq \theal$..
      \item[\Cref{cond:5-3}.]
      By \Cref{line:algtob-vote-chainava}, 
      $\chainava^{\voting{t_i}}_i \in \{\chainava^{\fastconfirming{(t_i-1)}}_i,\GJfrozen[\voting{t_i}]_i.\chain,(\mfcvoteFFG{t_i}_i)^{\lceil \kappa}\}$.
      Let us consider each case.
      \begin{description}
        \item[Case 1: $\chainava^{\voting{t_i}}_i = \chainava^{\fastconfirming{(t_i-1)}}_i$.] 
        \Cref{cond:6-3} for slot $t_i-1$ implies 
        \Cref{cond:5-3} for slot $t_i$.
        \item[Case 2: {$\chainava^{\voting{t_i}}_i = \GJfrozen[\voting{t_i}]_i.\chain$}]
        Due to \Cref{line:algtob-merge-ch-frozen,line:algtob-prop-if,line:algtob-prop-merge-gj}, $\mathsf{J}(\GJfrozen[\voting{t_i}]_i, \V^{\voting{t_i}})$.
        And from \Cref{cond:8-3}, $\GJfrozen[\voting{t_i}]_i.c \geq \theal$.
        Hence, by following the reasoning applied when discussing \Cref{cond:1-3,cond:2-3},
        $\GJfrozen[\voting{t_i}]_i.c \in [\theal,t_i)$ and
        there exists a validator $v_k \in H_{\GJfrozen[\voting{t_i}]_i.c}$ such that $\chainava^{\voting{\GJfrozen[\voting{t_i}]_i.c}}_k\succeq \GJfrozen[\voting{t_i}]_i.\chain = \chainava^{\voting{t_i}}_i$.
        Consider two sub cases.
        \begin{description}
          \item[Case 2.1: {$\GJfrozen[\voting{t_i}]_i.c = \theal$}.] 
          Given that $\mathsf{J}(\GJfrozen[\voting{t_i}]_i, \V^{\voting{t_i}})$,
          from the proof of \Cref{cond:6-3} for slot $\theal$ and the proof of Case 1 of \Cref{cond:1-3,cond:2-3}, we know that, for any validator $v_j$ honest in round $\fastconfirming{\theal}$, $\Chain^{\fastconfirming{\theal}}_j  = \chain_p \succeq \GJfrozen[\voting{t_i}]_i.\chain$.
          Hence, $\Chain^{\fastconfirming{\theal}}_j \succeq \GJfrozen[\voting{t_i}]_i.\chain = \chainava^{\voting{t_i}}_i$.
          \item[Case 2.2: {$\GJfrozen[\voting{t_i}]_i.c > \theal$}.] 
          In this case, we can apply the inductive hypothesis.
          Specifically, \Cref{cond:5-3} for slot $\GJfrozen[\voting{t_i}]_i.c < t_i$ implies 
          \Cref{cond:5-3} for slot $t_i$.
        \end{description}
        \item[Case 3: {\normalfont $\chainava^{\voting{t_i}}_i = (\mfcvoteFFG{t_i}_i)^{\lceil \kappa}$}.] 
        \sloppy{\Cref{line:algga-no-ffg-vote-chainava} of \Cref{algo:prob-ga-fast} implies that $\Chain^{\voting{t_i}}_i \succeq \mfcvoteNoFFG{t_i}_i$.}
        Then, 
        \Cref{cond:4-3} implies that $\Chain^{\voting{t_i}}_i \succeq \mfcvoteFFG{t_i}_i$.
        From \Cref{line:algtob-vote-chainava} of \Cref{alg:3sf-tob-noga}, we can conclude that $\Chain^{\voting{t_i}}_i \succeq \mfcvoteFFG{t_i}_i \succeq \chainava^\voting{t_i}_i$.
      \end{description}
      \item[\Cref{cond:6-3}.]
      By \Crefrange{line:algotb-at-confirm}{line:algtob-set-chaava-to-bcand}, $\chainava^{\fastconfirming{t_i}}_i \in \{\chainava^{\voting{t_i}}_i, \GJ(\VFFG^{\fastconfirming{t_i}}).\chain, \chain^C\}$ with $(\chain^C, Q) = \texttt{fastconfirm}(\VFFG^{\fastconfirming{t_i}}_i,t_i) \land Q\neq \emptyset$.
      Let us consider each case.
      \begin{description}
        \item[Case 1: $\chainava^{\fastconfirming{t_i}}_i = \chainava^{\voting{t_i}}_i$.] In this case, given that $t_i > \theal$, \Cref{cond:5-3} implies \Cref{cond:6-3}.
        \item[Case 2: \normalfont$\chainava^{\fastconfirming{t_i}}_i = \GJ(\VFFG^{\fastconfirming{t_i}}).\chain$.] 
        \sloppy{From \Cref{cond:7-3}, we know that $\GJ(\VFFG^{\fastconfirming{t_i}}).c \geq \theal$.}
        By following the reasoning applied when discussing \Cref{cond:1-3,cond:2-3}, this means that $\GJ(\VFFG^{\fastconfirming{t_i}}) \in [\theal, t_i)$ and there exists a validator $v_k \in H_{\GJ(\VFFG^{\fastconfirming{t_i}}).c}$ such that $\chainava^{\voting{\GJ(\VFFG^{\fastconfirming{t_i}}).c}}_k \succeq \chainava^{\fastconfirming{t_i}}_i$.
        Consider two sub cases.
        \begin{description}
          \item[Case 2.1: \normalfont$\GJ(\VFFG^{\fastconfirming{t_i}}).c = \theal$.]
          Similar to the  proof of Case 2.1 of \Cref{cond:5-3}.
          From the proof of \Cref{cond:6-3} for slot $\theal$ and the proof of Case 1 of \Cref{cond:1-3,cond:2-3}, we know that, for any validator $v_j$ honest in round $\fastconfirming{\theal}$, $\Chain^{\fastconfirming{\theal}}_j  = \chain_p \succeq \GJ(\VFFG^{\fastconfirming{t_i}}).\chain$.
          Hence, $\Chain^{\fastconfirming{\theal}}_j \succeq \GJ(\VFFG^{\fastconfirming{t_i}}).\chain = \chainava^{\fastconfirming{t_i}}_i$.
          \item[Case 2.2: \normalfont$\GJ(\VFFG^{\fastconfirming{t_i}}).c > \theal$.]  
          In this case, \Cref{cond:5-3} for slot $\GJ(\VFFG^{\fastconfirming{t_i}}).c < t_i$ implies \Cref{cond:6-3} for slot $t_i$.
        \end{description}
        \item[Case 3: $\chainava^{\fastconfirming{t_i}}_i = \chain^C$ with $(\chain^C, Q) = \texttt{fastconfirm}(\VFFG^{\fastconfirming{t_i}}_i,t_i) \land Q\neq \emptyset$.]
        This implies that $\VFFG^{\fastconfirming{t_i}}_i$ includes a quorum of \textsc{vote} message for chain $\chain^C$ and slot $t_i$.
        Given \Cref{cond:4-3},  $\VNoFFG^{\fastconfirming{t_i}}_i$ also includes a quorum of \textsc{vote} message for chain $\chain^C$ and slot $t_i$.
        Hence, $\Chain^{\fastconfirming{t_i}}_i = \chainava^{\fastconfirming{t_i}}_i$. 
        \end{description}
      \item[\Cref{cond:0bt-3}.]
      As argued in the proof of \Cref{cond:3-3,cond:4-3}, we know that $\FFGExec$ and $\NoFFGExec$ are honest-output-dynamically-equivalent from round $\voting{\theal}$ up to round $\merging{t_i}$.
      This, \Cref{cond:3-3,cond:4-3} and \cref{cond:adv-decision-4-heal} of $\ANoFFG$'s set of decisions clearly imply \Cref{cond:0bt-3} up to  round $\proposing{(t_i+1)}$.\qedhere
    \end{description}
  \end{description}
\end{proof}

\begin{lemma}
  \label{lem:keep-voting-tob-fast-conf-ffg-heal}
  If $t \geq \theal$, then \Cref{lem:keep-voting-tob-fast-conf-ffg} holds for $\GST > 0$ as well.
\end{lemma}
\begin{proof}
  Due to \Cref{lem:vote-proposal-fast-conf-heal}, we can use the proof of \Cref{lem:keep-voting-tob-fast-conf-ffg}, with the following changes.
  \begin{enumerate}
    \item Replace \Cref{lem:equiv-ga2} with \Cref{lem:equiv-ga3},
    \item Replace ``dynamically-equivalent'' with ``$\theal$-dynamically-equivalent'',
    \item Replace ``honest-output-dynamically-equivalent'' with ``honest-output-dynamically-equivalent from round $\voting{\theal}$'', and
    \item Replace \Cref{cond:3-2,cond:4-2} of \Cref{lem:equiv-ga2} with \Cref{cond:3-3,cond:4-3} of \Cref{lem:equiv-ga3}.\qedhere
  \end{enumerate}
\end{proof}

\begin{lemma}
  \label{lem:vote-proposal-fast-conf-ffg-heal}
  If $t \geq \theal$, then \Cref{lem:vote-proposal-fast-conf-ffg} holds for $\GST > 0$ as well.
\end{lemma}
\begin{proof}
    Given \Cref{lem:keep-voting-tob-fast-conf-ffg-heal}, follow the proof of \Cref{lem:vote-proposal-fast-conf-ffg} with the following changes.
    \begin{enumerate}
      \item Replace \Cref{lem:equiv-ga2} with \Cref{lem:equiv-ga3},
      \item Replace ``dynamically-equivalent'' with ``$\theal$-dynamically-equivalent'',
      \item Replace ``honest-output-dynamically-equivalent'' with ``honest-output-dynamically-equivalent from round $\voting{\theal}$'', and
      \item Replace \Cref{cond:5-2,cond:6-2} of \Cref{lem:equiv-ga2} with \Cref{cond:5-3,cond:6-3} of \Cref{lem:equiv-ga3}.\qedhere
    \end{enumerate}
\end{proof}

\begin{lemma}
\label{lem:ga-confirmed-always-canonical-ffg-heal}
  If $r_i \geq \fastconfirming{\theal}$, then \Cref{lem:ga-confirmed-always-canonical-ffg} holds for $\GST > 0$ as well.
\end{lemma}
\begin{proof}
  Given \Cref{lem:ga-confirmed-always-canonical-heal}, follow the proof of \Cref{lem:ga-confirmed-always-canonical-ffg} with the following changes.
  \begin{enumerate}
    \item Replace \Cref{lem:equiv-ga2} with \Cref{lem:equiv-ga3},
    \item Replace ``dynamically-equivalent'' with ``$\theal$-dynamically-equivalent'', and
    \item Replace \Cref{cond:5-2,cond:6-2} of \Cref{lem:equiv-ga2} with \Cref{cond:5-3,cond:6-3} of \Cref{lem:equiv-ga3}.\qedhere
  \end{enumerate}
\end{proof}

\begin{theorem}[Reorg Resilience]
  \Cref{alg:3sf-tob-noga} is $\eta$-reorg-resilient after slot $\theal$ and time $\fastconfirming{\theal}$.
\end{theorem}
\begin{proof}
  From \Cref{lem:vote-proposal-fast-conf-ffg-heal,lem:ga-confirmed-always-canonical-ffg-heal} and the proof of \Cref{thm:reorg-res-prop-tob-ffg}.
\end{proof}


\begin{theorem}[$\eta$-dynamic-availability]
\Cref{alg:3sf-tob-noga} is $\eta$-dynamically-available from time $\fastconfirming{\theal}$.
\end{theorem}
\begin{proof}
  Given \Cref{lem:vote-proposal-fast-conf-ffg-heal,lem:ga-confirmed-always-canonical-ffg-heal}, it follows from the proof of \Cref{thm:dyn-avail-fast-conf-tob}.
\end{proof}

\begin{lemma}\label{lem:asyn-induction2-ffg-heal}
  If $t\geq \theal$, then \Cref{lem:asyn-induction2-ffg} holds for $\GST > 0$ as well.
\end{lemma}
\begin{proof}
  As in the proof of \Cref{lem:asyn-induction2-ffg},
  we proceed by induction on $t_i$ and 
  and the following conditions to the inductive hypothesis.
  \begin{enumerate}
    \item\label[condition]{cond:ffg1-heal} For any slot $t'$, let $\X^{t',\chain}$ be the set of validators $v_i$ in $H_{\voting{(t')}}$ such that $\chainava^\voting{t'}_i$ conflicts with $\chain$. Then, for any $t'' \in (t,t_i]$, $\left|X^{t'',\chain} \cup A_\infty\right|<\frac{2}{3}n$.
  \end{enumerate}
  However, in this proof, we also need to add the following condition to the inductive hypothesis.
  \begin{enumerate}[start=2]
    \item For any slot $t_j \in [t,t_i]$ and validator $v_j \in W_{\voting{t_j}}$, 
    \begin{enumerate}[label*=\arabic*.,ref=\theenumi.\arabic*]
      \item\label[condition]{cond:ffg2-heal} $\GJ(\V^{\fastconfirming{t_j}}_j).c \geq t$ and 
      \item\label[condition]{cond:ffg3-heal} $\GJ(\V^{\fastconfirming{t_j}}_j)$ does not conflict with $\chain$.
    \end{enumerate}
  \end{enumerate}

  \begin{description}
    \item[Base Case: {$t_i \in [t,t_a]$}.]
    \Cref{cond:ffg2-heal} is proved in the proof of \Cref{lem:equiv-ga3}.
    
    From \Cref{lem:keep-voting-tob-fast-conf-ffg-heal} which, given that we assume $f<\frac{n}{3}$, also implies \Cref{cond:ffg1-heal}.
    This and \Cref{line:algtob-vote-chainava,line:algtob-vote,line:algtob-vote-comm} imply \Cref{cond:ffg3-heal}.
    \item[Inductive Step: $t_i > t_a$.]
    We assume that the Lemma and the additional conditions hold for any slot $t_i-1$ and prove that it holds also for slot $t_i$.
    It should be easy to see that if \Cref{cond:ffg2-heal} holds for a given slot, it will keep holding for any future slot.
    Given that we have already proved that it holds in the base case and that $[t,t_a]$ is not any empty set, then \Cref{cond:ffg2-heal} is clearly true for slot $t_i$ as well.
    Now, note that due to \Cref{line:algtob-prop-if}, for any honest validator $v_i \in H_{\voting{(t_i)}}$, $\GJfrozen[\voting{t_i}]_i \geq \GJ(\V^{\fastconfirming{(t_i-1)}}_i)$.
    Hence, given that $t_i-1\geq t$, \Cref{cond:ffg1-heal,cond:ffg2-heal,cond:ffg3-heal} holding for slot $t_i-1$ means that for any validator $v_i \in H_{\voting{(t_i)}}$, $\GJfrozen[\voting{t_i}]_i$ does not conflict with $\chain$.
    Then, observe that if Lemma's statement holds for slot $t_i$, then \Cref{cond:ffg3-heal} clearly holds as well.
    Hence, we are left only with proving the Lemma's statement and \Cref{cond:ffg1-heal}.
    Let us now proceed by cases and keep in mind that due to \Cref{line:algtob-vote-chainava}, if in slot $t_i$ a validator $v_i \in H_{\voting{(t_i)}}$ casts a \textsc{vote} message for a chain extending $\chain$, then $\chainava^{\voting{t_i}}_i$ does not conflict with $\chain$.
    \begin{description}
      \item[Case 1: {$t_i \in [t_a+1, t_a+\pi+1]$}.] 
      Given that for any honest validator $v_i \in H_{\voting{(t_i)}}$, $\GJfrozen[\voting{t_i}]_i$ does not conflict with $\chain$, due to \Cref{lem:asyn-induction-heal}, it should be easy to see that,  \Cref{lem:asyn-induction} can be applied to \Cref{alg:3sf-tob-noga} as well.
      This implies that all validators in $W_\voting{t_i}$ cast \textsc{vote} messages for chains extending $\chain$.
      Given Constraint~\eqref{eq:async-condition2}, then \Cref{cond:ffg1-heal} holds for slot $t_i$ as well. 
      \item[Case 2: $t_i = t_a+\pi+2$.]
      Given that for any honest validator $v_i \in H_{\voting{(t_i)}}$, $\GJfrozen[\voting{t_i}]_i$ does not conflict with $\chain$, due to \Cref{lem:asyn-induction2-heal}, it should be easy to see that we can apply here the same reasoning used for this same case in \Cref{lem:asyn-induction2}.
      Then, given that we assume $f<\frac{n}{3}$,  \Cref{cond:ffg1-heal} holds for slot $t_i$ as well.
      \item[Case 3: $t_a \geq t_a+\pi+3$.]
      Given that for any honest validator $v_i \in H_{\voting{(t_i)}}$, $\GJfrozen[\voting{t_i}]_i$ does not conflict with $\chain$ and due to \Cref{lem:keep-voting-tob-fast-conf-heal},
      it should be easy to see that \Cref{lem:keep-voting-tob-fast-conf} can be applied to \Cref{alg:3sf-tob-noga} as well.
      This implies that all validators in $H_{\voting{t}}$ casts \textsc{vote} messages for chains extending $\chain$.
      Then, due to \Cref{lem:asyn-induction2-heal}, the proof proceeds as per the same case in \Cref{lem:asyn-induction2}.
      Also, given that we assume $f<\frac{n}{3}$,  \Cref{cond:ffg1-heal} holds for slot $t_i$ as well.\qedhere
    \end{description}
  \end{description} 
\end{proof}

\begin{lemma}\label{lem:asyn-induction3-ffg-heal}
  If $t\geq \theal$, then \Cref{lem:asyn-induction3-ffg} holds for $\GST > 0$ as well.
\end{lemma}
\begin{proof}
  Given, \Cref{lem:asyn-induction2-ffg-heal,lem:asyn-induction3-heal}, we can just follow the proof of \Cref{lem:asyn-induction3-ffg}.
\end{proof}

\begin{theorem}
  [Asynchrony Reorg Resilience]
  \Cref{alg:3sf-tob-noga} is asynchrony resilient after slot $\theal$ and time $\fastconfirming{\theal}$.
\end{theorem}
\begin{proof}
  Given \Cref{thm:reorg-res-prop-tob-ffg-heal}, and \Cref{lem:vote-proposal-fast-conf-ffg-heal,lem:asyn-induction3-ffg} we can just follow the proof of \Cref{thm:async-resilience-tob-ffg}.
\end{proof}

\begin{theorem}
  \Cref{alg:3sf-tob-noga} ensures Asynchrony Safety Resilience after time $\fastconfirming{\theal}$.
\end{theorem}
\begin{proof}
  Given \Cref{thm:dyn-avail-fast-conf-tob-ffg-heal} and \Cref{lem:asyn-induction3-ffg-heal}, we can just follow the proof of \Cref{thm:async-safety-resilience-tob-ffg}.
\end{proof}

\section{Two-Slot Finality}\label{sec:two-slot-finality}

In this section, we explore a further trade-off between the SSF protocol~\cite{DBLP:conf/esorics/DAmatoZ23} and our protocol as presented in the main text.
Specifically, by introducing a second vote round corresponding to  the third vote round from the SSF protocol, we obtain a protocol that can finalize chains proposed by honest proposers one slot earlier, \ie, by the end of the next slot without increasing the slot length even if vote aggregation is employed.

\paragraph{Acknowledgment Messages.} As per the SSF protocol, we introduce \emph{acknowledgment} messages that have the form $[\textsc{ack},\C,t,v_i]$ where $\C$ is a checkpoint with $\C.c = t$\footnote{Note that in this context, the parameter $t$ is redundant in the \textsc{ack} message since it is already incorporated within $\C$. However, for the sake of clarity, we have chosen to include it explicitly in the \textsc{ack} message.}.
A \emph{supermajority acknowledgment of $\C$} is a set of at least $\frac{2}{3}n$ distinct \textsc{ack} messages for $\C$.
At round $\fastconfirming{t}$, for any honest validator $v_i$, if $\GJ(\V^\fastconfirming{t}_i).c = t$, then it broadcasts $[\textsc{ack},\GJ(\V^\fastconfirming{t}_i),t,v_i]$.
Any observer that receives a supermajority acknowledgment for a \emph{justified} checkpoint $\C$ considers $\C$ to be finalized.

\paragraph{Slashing Rule.} The acknowledgment vote $[\textsc{ack},\C,t,v_i]$ can be interpreted as an \textsc{ffg-vote} $\C \to \C$.
Then, as in the SSF protocol, we introduce the third slashing $\mathbf{E_3}$: If a validator $v_i$ has sent an \textsc{ffg-vote} $\C_1 \to \C_2$, an \textsc{ack} vote for $\C_a$, and $\C_1 < C_a \land \C_a.c < \C_2.c$, then $v_i$ is slashable.
We do not need to introduce any slashing rule for the case that $\C_2 \neq \C_a \land \C_2.c = C_a.c$ as, to finalize $\C_a$, $\C_a$ must first be justified and, hence, such a condition is already covered by the $\mathbf{E_1}$ slashing condition.

\begin{lemma}
  \label{lem:accountable-safety-ack}
  Let $f \in [0,n]$ and allow for finalization of checkpoints via \textsc{ack} message as well.
  If two conflicting chains are finalized according to any two respective views, then at least $\frac{n}{3}$ validators can be detected to have violated either $\mathbf{E_1}$, $\mathbf{E_2}$, or $\mathbf{E_3}$.
\end{lemma}
\begin{proof}
  Assume that there exists two conflicting finalized checkpoints $\C_1$ and $\C_2$.
  This implies that both checkpoints are justified as well.
  Without loss of generality, assume $\C_2.c \geq \C_1.c$.
  Let us consider two cases.
  \begin{description}
    \item[Case 1: $\C_1$ is finalized via supermajority links.] We can apply \Cref{lem:accountable-jutification} to reach a contradiction.
    \item[Case 2: $\C_1$ is finalized via supermajority acknowledgment.] 
    Let $\C_j$ be the smallest checkpoint such that $\C_j.c > \C_1.c$ and $\C_j.\chain \nsucceq \C_1.\chain$.
    By our assumptions, we know that such a checkpoint exists.
    We also know that there exits a validator $v_i$ that has cast both an \textsc{ack} vote for $\C_1$ and a valid \textsc{ffg-vote} $\C'_j \to \C_j$.
    Given that $\C_j.\chain$ conflicts with $\C_1.\chain$, this implies that $\C'_j.\chain \nsucceq \C_1.\chain$ which, by the minimality of $\C_j$ implies that $\C'_j.c \leq \C_1.c$.
    If $\C'_j.c = \C_1.c$, then given that as proven in \Cref{lem:accountable-jutification}, $\C'_j.\chain$ cannot conflict with $\C_1.chain$ and thath $\C'_j.\chain \nsucceq \C_1.\chain$, $\C'_j.\chain \preceq \C_1.\chain$.
    Hence, $\C'_j.c \leq \C_1.c$, implies $\C'_j < \C_1$.
    Given that $\C'_j.c > \C_1.c$, $v_i$ violates condition $\mathbf{E_3}$.\rs{Do we prefer the style used in this proof where we never give name to each checkpoints' component or the style of \Cref{lem:accountable-jutification}? It would be better to unify, but not a big deal}\qedhere
  \end{description}
\end{proof}

\begin{lemma} \label{lem:never-slashed-generalized-ack}
  Allow for finalization of checkpoints via \textsc{ack} message as well
  If \Cref{prop:never-slashed} holds, then honest validators are never slashed. 
\end{lemma}

\begin{proof}
  Honest validators only ever send \textsc{ack} messages for the greatest justified checkpoint in their view at time $\fastconfirming{t}$.
  This implies that they will never send an \textsc{ffg-vote} in a slot strictly higher than $t$ with a source checkpoint lower than the checkpoint that they send an \textsc{ack} message for.
\end{proof}

\begin{theorem}[Accountable Safety with Acknowledgments]
  Let $f \in [0,n]$ and allow for finalization of checkpoints via \textsc{ack} message as well.
  If \Cref{prop:never-slashed,prop:chfin} hold, then the finalized chain $\chainfin$ is $\frac{n}{3}$-accountable.
\end{theorem}
\begin{proof}
  Follows from \Cref{lem:accountable-safety-ack,lem:never-slashed-generalized-ack}.
\end{proof}

\begin{theorem}[Liveness with Acknowledgments]
  \label{thm:liveness-ffg-general-ack}
  Allow for finalization of checkpoints via \textsc{ack} message as well and 
  assume that \Cref{prop:succ-for-ffg-liveness} holds (and $f< \frac{n}{3}$).
  Let $v_p$ be any validator honest by the time it \textsc{propose}s chain $\chain_p$ in a slot $t$ such that $\proposing{t} \geq \max(\GST, \GAT) + 4\Delta$.
  Then, chain $\chain_p$ is justified and finalized during slot~$t+1$.
  In particular, $\chain_p \preceq \chainfin^{4\Delta (t+1) + 3 \Delta}_i$ for any validator $v_i \in H_{4\Delta (t+1) + 3 \Delta}$.
\end{theorem}
\begin{proof}
  Follows from the proof of \Cref{thm:liveness-ffg-general} and the fact that \textsc{ack} messages for checkpoint $(\chain,t+1)$ are sent at time $\fastconfirming{(t+1)}$ and therefore received by time $\merging{(t+1)} = 4\Delta(t+1) + 3\Delta$.
\end{proof}

The following Corollary shows that it takes $6\Delta$ for a chain proposed by an honest validator to become finalized in the global view, meaning that no honest validator can ever finalize any conflicting chain.
If we assume that vote rounds take $2\Delta$, then this latency becomes $8\Delta$ as Fast Confirmation must wait to receive the \textsc{ffg-vote}s sent during the vote round, but there is no need at any point to wait for \textsc{ack} messages.

\begin{corollary}\label{cor:liveness-ffg-general}
  Allow for finalization of checkpoints via \textsc{ack} message as well and  and 
  assume that \Cref{prop:succ-for-ffg-liveness} holds (and $f< \frac{n}{3}$).
  Let
  $\chainfin^r_\mathsf{G}$ be the longest finalized chain according to view $\V^r_\mathsf{G}$ and
  $v_p$ be any validator honest by the time it \textsc{propose}s chain $\chain_p$ in a slot $t$ such that $\proposing{t} \geq \max(\GST, \GAT) + 4\Delta$.
  Then, $\chain_p \preceq \chainfin^{4\Delta (t+1)+2\Delta}_\mathsf{G}$.
\end{corollary}
\begin{proof}
  Follows from the proof of \Cref{thm:liveness-ffg-general-ack}.
\end{proof}}

\end{document}